\documentclass[12pt, a4paper]{article}
\usepackage[margin = 2.0cm]{geometry}
\usepackage{amsfonts, amsthm, amsmath, amssymb}
\usepackage[T1]{fontenc}
\usepackage{array}
\usepackage{enumitem}
\usepackage[many]{tcolorbox} %
\usepackage{scalerel}
\usepackage{soul}
\usepackage{color, xcolor}
\usepackage{multirow}
\usepackage{algorithm,algorithmic}
\usepackage{booktabs}
\usepackage{txfonts} %
\usepackage{tgpagella} %
\usepackage[round]{natbib} %
\usepackage[colorlinks, citecolor = custom_light_blue, urlcolor = custom_light_blue, linkcolor = custom_light_blue, bookmarks = false, hypertexnames = true]{hyperref} %
\usepackage[flushleft]{threeparttable} %
\usepackage{placeins} %

\makeatletter
\newcommand{\indep}{\perp \!\!\! \perp}

\newcommand{\customlabel}[2]{%
\protected@write \@auxout {}{\string \newlabel {#2}{{#1}{\thepage}{#1}{#2}{}}}
\hypertarget{#2}{}
}
\makeatother

\usepackage{float} %
\def\arraystretch{0.75} %

\usepackage[title]{appendix}

\usepackage{chngpage}

\usepackage{adjustbox}

\usepackage{physics} %

\newtheoremstyle{break}
  {\topsep}{\topsep}%
  {\itshape}{}%
  {\bfseries}{}%
  {\newline}{}%
\theoremstyle{break}

\theoremstyle{plain} %
\newtheorem{innercustomplaingeneric}{\customplaingenericname}
\providecommand{\customplaingenericname}{}
\newcommand{\newtheoremplain}[2]{%
  \newenvironment{#1}[1]
  {%
   \ifdefined\crefalias\crefalias{innercustomplaingeneric}{#2}\fi
   \renewcommand\customplaingenericname{#2}%
   \renewcommand\theinnercustomplaingeneric{##1}%
   \innercustomplaingeneric
  }
  {\endinnercustomplaingeneric}%
  \ifdefined\crefname\crefname{#2}{#2}{#2s}\fi
}
\newtheoremplain{thm}{Theorem}
\newtheoremplain{coro}{Corollary}
\newtheoremplain{prop}{Proposition}
\newtheoremplain{lemma}{Lemma}

\theoremstyle{definition}
\newtheorem{innercustomdefinitiongeneric}{\customdefinitiongenericname}
\providecommand{\customdefinitiongenericname}{}
\newcommand{\newtheoremdefinition}[2]{%
  \newenvironment{#1}[1]
  {%
   \ifdefined\crefalias\crefalias{innercustomdefinitiongeneric}{#2}\fi
   \renewcommand\customdefinitiongenericname{#2}%
   \renewcommand\theinnercustomdefinitiongeneric{##1}%
   \innercustomdefinitiongeneric
  }
  {\endinnercustomdefinitiongeneric}%
  \ifdefined\crefname\crefname{#2}{#2}{#2s}\fi
}
\newtheoremdefinition{deff}{Definition}
\newtheoremdefinition{assu}{Assumption}
\newtheoremdefinition{nota}{Notation}
\newtheoremdefinition{property}{Property}
\newtheoremdefinition{remark}{Remark}

\theoremstyle{remark}
\newtheorem{innercustomremarkgeneric}{\customremarkgenericname}
\providecommand{\customremarkgenericname}{}
\newcommand{\newtheoremremark}[2]{%
  \newenvironment{#1}[1]
  {%
   \ifdefined\crefalias\crefalias{innercustomremarkgeneric}{#2}\fi
   \renewcommand\customremarkgenericname{#2}%
   \renewcommand\theinnercustomremarkgeneric{##1}%
   \innercustomremarkgeneric
  }
  {\endinnercustomremarkgeneric}%
  \ifdefined\crefname\crefname{#2}{#2}{#2s}\fi
}
\newtheoremremark{remk}{Remark}

\DeclareMathOperator*{\supp}{supp}

\usepackage{mathtools} %
\makeatletter
\newcommand{\biggg}{\bBigg@{4}}
\newcommand{\Biggg}{\bBigg@{5}}
\makeatother

\usepackage{tikz} %
\usetikzlibrary{trees}
\usetikzlibrary{shapes.geometric, arrows}

\newcolumntype{L}[1]{>{\raggedright\let\newline\\\arraybackslash\hspace{0pt}}p{#1}}
\newcolumntype{C}[1]{>{\centering\let\newline\\\arraybackslash\hspace{0pt}}m{#1}}
\newcolumntype{R}[1]{>{\raggedleft\let\newline\\\arraybackslash\hspace{0pt}}m{#1}}

\definecolor{custom_blue}{rgb}{0.0627, 0.1765, 0.3686}
\definecolor{custom_darker_blue}{rgb}{0.0078, 0.0667, 0.1373}
\definecolor{custom_light_blue}{rgb}{0.18, 0.49, 0.72}
\definecolor{highlight_grey}{rgb}{0.85, 0.85, 0.85} %

\definecolor{highlight_blue}{rgb}{0.7510, 0.9216, 1}

\usetikzlibrary{calc}

\usepackage{etoolbox}

\makeatletter
\pretocmd{\@makefntext}{\linespread{1}\selectfont}{}{}
\makeatother

\begin{document}

\title{Simple Specification-Free Causal Inference\\ in Large Panels\thanks{This paper was previously circulated under the title ``Inferring Treatment Effects in Large Panels by Uncovering Latent Similarities''. The first version of this paper was prepared for the 2025 ESWC submission on January 31, 2025, and is also available \href{https://www.homepages.ucl.ac.uk/~uctpzel/inferring_TE.pdf}{here}. We thank Alberto Abadie, Timothy B. Armstrong, Andrew Chesher, Timothy Christensen, Guido Imbens, Aureo de Paula, Kirill Ponomarev, Liyang Sun, and participants in numerous seminars and conferences for their valuable comments and suggestions.}}

\author{Ben Deaner\thanks{Department of Economics, University College London, and CeMMAP.} \and Chen-Wei Hsiang\footnotemark[2] \and Andrei Zeleneev\footnotemark[2]}

\date{\small This version: September, 2026\\ First Circulated: January 31, 2025}
\maketitle

\begin{abstract}

We propose a simple and transparent approach to causal inference using panel data with a large number of units and periods. Instead of imposing a particular parametric specification on the untreated potential outcomes, we assume they follow a nonparametric nonseparable factor model. We use long pre-treatment outcome histories to construct a measure of similarity between the latent factors of different units. Using this measure, we impute missing counterfactual means and propensity scores through kernel smoothing. Under weak smoothness conditions, these estimates can attain the optimal nonparametric convergence rate up to logarithmic factors, provided that the pre-treatment history grows sufficiently quickly with the sample size. We construct a doubly robust estimator of the period-specific average treatment effect on the treated (ATT) and provide conditions under which it is $\sqrt{N}$-consistent and asymptotically normal, centered at the true ATT. Simulations demonstrate accurate inference for a wide range of linear and nonlinear data-generating processes.

\bigskip

\noindent \textbf{Keywords:} treatment effects, panel data, nonlinear factor models, doubly robust estimation

\end{abstract}

\newpage

\section{Introduction} \label{sec:1}

The presence of unobserved confounding factors presents a major challenge for causal inference. The availability of rich panel data can help address this problem. Panel datasets contain multiple observations for each individual and can thus provide a means of controlling for time-invariant confounders (fixed effects). In addition, panel data can enable the imputation of individual counterfactual outcomes. Methods for causal inference with panel data are widely used in empirical work, and their development is an active and rapidly growing area of research (see, e.g., \citealp{arkhangelsky2024causal}).

Many popular methods for causal inference using panel data assume that untreated potential outcomes have a linear or affine structure. For example, difference-in-differences (DiD) and two-way fixed effects (TWFE) methods rely on parallel-trends assumptions, whereas synthetic control (SC) and most matrix-completion methods require untreated potential outcomes to follow a linear factor model. While such restrictions can simplify estimation and inference, and may be necessary given limited data availability, it is natural to ask whether they are fundamental for identifying treatment effects in large panels and, most importantly, whether they can be relaxed while retaining a theoretically disciplined and practical approach that performs well in empirically relevant sample sizes.

This paper addresses this question by proposing a new method for estimating and conducting inference on the causal effects of a binary treatment $w_{i,t}$ in large panels. We assume that the untreated potential outcome of unit $i$ in period $t$ follows the possibly nonlinear and nonseparable factor model
\begin{align}
    \label{eq:Y_0 model}
    Y_{i,t}(0) = \mu^{(0)}(\alpha_i, \lambda_t) + u_{i,t},
\end{align}
where $\alpha_i$ and $\lambda_t$ are unobserved unit and time effects, $\mu^{(0)}(\cdot, \cdot)$ is an \emph{unknown} function, and $u_{i,t}$ is an exogenous residual.\footnote{It is also conceptually straightforward to allow for a linear adjustment for observed covariates in \eqref{eq:Y_0 model}. We abstract from this extension for clarity of exposition.} We assume that, conditional on the latent factors, untreated potential outcomes are independent of treatment assignment. Importantly, we do not assume that $\mu^{(0)}(\cdot, \cdot)$ has any particular parametric form. Thus, model~\eqref{eq:Y_0 model} is substantially more general than the commonly used linear formulations $\mu^{(0)}(\alpha_i, \lambda_t) = \alpha_i + \lambda_t$ and $\mu^{(0)}(\alpha_i, \lambda_t) = \alpha_i' \lambda_t$ imposed by TWFE and SC methods, respectively.

While the generality of formulation~\eqref{eq:Y_0 model} is appealing, it complicates the identification and estimation of treatment effects because $\alpha_i$ cannot be identified without strong additional restrictions imposed on the form of $\mu^{(0)}(\cdot, \cdot)$. Rather than attempting to model $\mu^{(0)}(\cdot, \cdot)$ and estimate $\alpha_i$ directly, we adapt an approach first developed by \citet{zhang2017estimating} for graphon estimation in network data. Specifically, we use a long history of pre-treatment outcomes to construct a pseudo-distance $\hat{d}_{i,j}$ that is informative about the similarity between the latent characteristics $\alpha_i$ and $\alpha_j$ of units $i$ and $j$. Using this pseudo-distance, we impute the \emph{values} of the conditional means of untreated potential outcomes and the propensity scores (both conditional on the latent factors) in a given post-treatment period $t$ \emph{for all units} through kernel smoothing. Notably, we can do so even though $\alpha_i$ \emph{is neither observed nor identified} in our setting.

We construct estimators of the untreated conditional means $\mu^{(0)}_{i,t}$ and the propensity scores $p_{i,t}$, and establish their rates of convergence (uniformly over $i$) as both the number of units $N$ and the number of pre-treatment periods $T_0$ go to infinity. If $T_0$ grows sufficiently quickly relative to $N$, these estimators attain the optimal nonparametric rate of convergence established by \citet{stone1980optimal}, up to a logarithmic factor, as if the latent characteristics $\alpha_i$ were observed.

Even at the optimal rate, however, these nonparametric estimators generally converge too slowly for conventional plug-in inference on the average treatment effect on the treated (ATT). We therefore combine them in a doubly robust estimator of the ATT for each post-treatment period $t$ (e.g., \citealp{robins1994estimation,chernozhukov2018double}). We propose a cross-fitting scheme that, together with double robustness, helps ensure that the resulting estimator is $\sqrt{N}$-consistent, asymptotically normal, and centered at the true ATT. Remarkably, under certain regularity conditions, and if $T_0$ grows sufficiently fast, our estimator attains the semiparametric efficiency bound that would apply if the latent characteristics $\alpha_i$ were observed.

Compared to some competing approaches, an advantage of our method is its simplicity and transparency. We avoid the need to decompose an incomplete matrix of potential outcomes or to regularize the components of that decomposition. Instead, we impute untreated potential outcomes for treated units by averaging the outcomes of untreated units, placing greater weight on those with similar pre-treatment outcome histories. Researchers can readily inspect the pseudo-distances and corresponding kernel weights to see which comparison units drive each imputation.

We provide extensive simulation evidence that our procedure performs well across a broad range of linear and nonlinear data-generating processes. When the additive TWFE model is correctly specified, our estimator delivers confidence-interval coverage comparable to that of TWFE. Under interactive and nonlinear factor structures, it continues to provide accurate coverage, whereas TWFE can perform poorly. Our method is also competitive with---and often outperforms---the local PCA approach of \citet{feng2026causal} in small and moderate panels, particularly when the pre-treatment history is relatively short. Finally, in an empirical illustration based on the Walmart sales panel, we conduct empirically calibrated Monte Carlo exercises in which our estimator performs nearly as well as an infeasible benchmark that uses the true latent characteristics. Taken together, our theoretical and numerical results show that relaxing linear factor restrictions in large panels need not come at the expense of reliable inference in empirically relevant sample sizes.

To assist empirical practitioners who wish to apply our method, we provide an R implementation at \url{https://github.com/chenweihsiang/PanelLatSim}.

\bigskip

\noindent{\bf Related literature} This paper contributes to the vast and rapidly growing literature on causal inference in panels. We refer the reader to \citet{arkhangelsky2024causal} for an excellent recent overview of this field. DiD and TWFE are widely used and continue to grow in popularity in applied work, perhaps because of their practicality and perceived transparency (\citealp{goldsmith2024tracking}); see \citet{de2023two} and \citet{roth2023s} for extensive reviews of the DiD and TWFE literature. A major focus of recent econometric work on these methods has been accommodating heterogeneous treatment effects (e.g., \citealp{de2020two,callaway2021difference,sun2021estimating,wooldridge2021two,borusyak2024revisiting}). However, their validity still relies crucially on parallel-trends assumptions. Several recent papers have re-examined the restrictiveness of these assumptions and, in some cases, proposed relaxations at the cost of giving up point identification of the causal parameters of interest (e.g., \citealp{manski2018right,rambachan2023more,ghanem2022selection}). By contrast, our framework allows for rich heterogeneity in treatment effects and selection mechanisms while retaining point identification without imposing parallel trends. To achieve this, however, we require a sufficiently long pre-treatment history.

Alternative popular approaches to causal inference in panels include synthetic control (SC) and matrix-completion methods; see \citet{abadie2021using} for a recent review of SC methods. Recent contributions to the SC and causal matrix-completion literatures include, among others, \citet{arkhangelsky2021synthetic,cattaneo2021prediction,chernozhukov2021exact,chernozhukov2018t} and \citet{athey2021matrix,bai2021matrix,agarwal2023causal,xiong2023large,choi2024matrix,abadie2024doubly}, respectively. Similarly to the interactive fixed effects panel literature (e.g., \citealp{pesaran2006estimation,bai2009panel}), SC and matrix-completion methods generalize the TWFE framework by allowing the untreated potential outcomes to follow a factor model. As in this paper, these methods use a long history of pre-treatment outcomes to impute missing counterfactual means. Most of these methods, however, rely on a linear factor model, $\mu^{(0)}(\alpha_i,\lambda_t)=\alpha_i'\lambda_t$, to establish statistical guarantees and justify inference, whereas our framework allows for a general nonparametric and nonlinear factor structure.

Some recent matrix-completion methods also accommodate nonlinear factor structures of the type considered in this paper by approximating the corresponding matrix of counterfactual means with a low-rank matrix. The resulting approximation error is typically controlled by imposing relatively strong smoothness restrictions on $\mu^{(0)}(\cdot,\cdot)$, such as requiring it to belong to a sufficiently smooth H\"older class (e.g., \citealp{fernandez2021low,choi2024inference,abadie2025causal}).\footnote{Among these papers, \citet{choi2024inference} develops an inference procedure for average treatment effects, whereas \citet{fernandez2021low} and \citet{abadie2025causal} establish estimation results but do not provide corresponding inference procedures.} By contrast, our imputation method does not rely on a low-rank approximation of $\mu^{(0)}(\cdot,\cdot)$. Under the substantially weaker smoothness assumption requiring Lipschitz continuity only in the first argument of $\mu^{(0)}(\cdot,\cdot)$, our imputation method achieves a nearly optimal nonparametric rate of convergence (given the imposed smoothness conditions) and allows us to construct a semiparametrically efficient estimator of the ATT.\footnote{\citet{dwivedi2025counterfactual} provide entrywise non-asymptotic error bounds for a nonlinear factor model under a similar Lipschitz-continuity assumption. Their imputation method, however, relies crucially on homoskedastic idiosyncratic errors. Moreover, they do not investigate whether their estimator achieves the optimal nonparametric rate of convergence and do not provide semiparametrically efficient estimators of average treatment effects.} %

Our imputation procedure builds on the pseudo-distance proposed by \citet{zhang2017estimating}, which allows us to find units with similar latent characteristics even though those characteristics are neither observed nor identified. Variants of this pseudo-distance have been employed for nonparametric graphon estimation \citep{zhang2017estimating,zeleneev2020identification,nowakowicz2024nonparametric,sun2026flexible}, controlling for unobservables using network data \citep{auerbach2022identification}, and estimation of nonlinear factor models \citep{feng2024optimal}. Related distance-based methods have also been applied to treatment-effect estimation in network and panel settings by \citet{hoshino2024estimating} and \citet{athey2025identificationaveragetreatmenteffects}. Unlike this paper, however, these studies neither establish asymptotic normality of their treatment-effect estimators nor provide asymptotically valid confidence intervals.\footnote{In particular, \citet{athey2025identificationaveragetreatmenteffects} focus on identification of the counterfactual means and do not provide asymptotic theory for their treatment effect estimator. \citet{hoshino2024estimating} provide error bounds for various average dyad treatment effects but do not establish their asymptotic distribution. None of these papers consider treatment effect estimation in a doubly robust form.} %

Imputation of the propensity scores alongside the counterfactual means is an important and distinctive ingredient of our approach that allows us to estimate the ATT using doubly robust moments (e.g., \citealp{robins1994estimation,robins1995semiparametric,hahn1998role,scharfstein1999adjusting,chernozhukov2018double,chernozhukov2022debiased}). We demonstrate that, under suitable conditions, the resulting doubly robust estimator is $\sqrt{N}$-consistent, asymptotically normal, and centered at the true ATT, whereas conventional imputation estimators in the matrix-completion literature generally inherit the slower nonparametric convergence rate of the estimated counterfactual means.

Although propensity scores are routinely estimated and used in doubly robust moments to adjust for \emph{observed} covariates in panel settings (e.g., \citealp{sant2020doubly,arkhangelsky2022doubly,chen2025efficient}), this strategy remains underexplored in latent factor models, where standard matrix-completion methods focus almost exclusively on imputing counterfactual means. A notable recent exception is \citet{abadie2024doubly}, who estimate propensity scores by applying matrix-completion methods to the observed treatment-assignment matrix. However, our propensity-score imputation approach is fundamentally different from theirs: we use pre-treatment outcome histories rather than rely on a low-rank approximation of the propensity-score matrix.\footnote{Outside of the panel causal inference literature, recent applications of doubly robust estimation in linear panel regression models with unobserved heterogeneity having a nonlinear factor form also include \citet{beyhum2024inference} and \citet{freeman2026inference}.}

Finally, \citet{feng2026causal} studies causal inference with unobserved confounders when a large number of auxiliary ``proxies'' for those confounders are available; in panel settings, pre-treatment outcomes can serve as such proxies. As in this paper, \citet{feng2026causal} estimates counterfactual means and propensity scores and combines them in doubly robust estimators of causal effects. Their imputation procedure, however, is quite different from ours. \citet{feng2026causal} first constructs $K$-nearest-neighbor neighborhoods using the pseudo-distance of \citet{zhang2017estimating} or another distance, applies local principal component analysis (PCA) to extract factor loadings within each neighborhood, and then uses the estimated loadings as regressors in local quasi-maximum-likelihood estimation (QMLE). \citet{feng2026causal} establishes asymptotic results under higher-order smoothness conditions.  Our method instead directly estimates the counterfactual means and propensity scores using kernel smoothing based on the pseudo-distance, avoiding local PCA, the selection of the number of components, and the subsequent QMLE step. A benefit of the resulting procedure is that it is simple and transparent, with counterfactual means and propensity scores expressed directly as smoothly weighted averages of observed outcomes and treatments. The simplicity of the procedure and our focus on panel data also carry over to the theoretical analysis. Under primitive conditions on the nonlinear factor structure and conditions on the serial dependence structure of the panel data, we establish nearly optimal uniform convergence rates for the estimated conditional means and propensity scores under Lipschitz continuity of the model, as opposed to higher-order smoothness.\footnote{While our imputation method achieves the nearly optimal rate of convergence under Lipschitz continuity, it does not take advantage of the potential higher-order smoothness, whereas the approach of \citet{feng2026causal} can achieve faster rates of convergence in the latter case.} We show that under our primitive conditions, the resulting ATT estimator is $\sqrt{N}$-consistent and asymptotically normal. Our simulation results suggest that the simplicity of our method does not come at the expense of finite-sample performance: our estimator is competitive with local PCA across the designs considered and often delivers coverage closer to nominal in both small and moderate panels, particularly when the pre-treatment history is relatively short.\footnote{We note that \citet{feng2026causal} also discusses a simpler estimator that extracts a single principal component, which is described as analogous to local-constant smoothing. However, the one-component procedure still uses local PCA and QMLE and therefore does not generally reduce to conventional nearest-neighbor averaging.}

\bigskip

The rest of the paper is organized as follows. Section~\ref{sec:2} introduces the framework and describes our estimator. Section~\ref{sec:3} provides formal statistical guarantees. Section~\ref{sec:4} presents simulation evidence, and Section~\ref{sec:5} provides an empirical illustration based on Monte Carlo exercises calibrated to the Walmart sales panel. All auxiliary lemmas and proofs are provided in Appendix~\ref{app:proof}.

\section{Setup and Proposed Method} \label{sec:2}
\subsection{Setup}
Our sample consists of individuals indexed by $i=1,...,N$ and time periods indexed by $t=1,...,T$. For each individual $i$ and period $t$, we observe an outcome $Y_{i,t}$ and a binary indicator $w_{i,t}$. The indicator $w_{i,t}$ is equal to one if individual $i$ is treated at or prior to time $t$ and zero otherwise.\footnote{For concreteness, we assume that treatment is an absorbing
state, which is typical for staggered treatment adoption applications.
Our approach also accommodates treatment switching by restricting
the period-$t$ comparison to currently treated and never-treated
units, provided that the maintained assumptions hold for this
comparison sample. In this extension, the equality in
Assumption~\ref{asm:1} is required for all units during the first
$T_0$ periods and for never-treated controls thereafter.} We assume throughout that no individual in the population is treated prior to a period $T_0+1$ with $T_0<T$. We let $Y_{i,s:t}$ denote the vector of outcomes for individual $i$ from periods $s$ to $t$ inclusive and similarly for other variables.

Let $Y_{i,t}(0)$ denote the potential outcome of individual $i$ at time $t$ under a counterfactual in which no individual is ever treated. We assume that if $w_{i,t}=0$ then $Y_{i,t}=Y_{i,t}(0)$. This precludes both anticipation and spillover effects from treated to untreated units (see, e.g., \citealp{Abbring2003, borusyak2024revisiting, sun2021estimating}). It does not rule out spillover effects between treated units. We state this formally in Assumption \ref{asm:1} below.

\begin{assu}{1}[No anticipation/spillovers] \label{asm:1}
If $w_{i,t}=0$ then $Y_{i,t}=Y_{i,t}(0)$.
\end{assu}

Central to our analysis is the assumption of a nonlinear factor model for the untreated potential outcomes. Let $\alpha_i$ be some latent and time-invariant characteristics of individual $i$, and let $\lambda_t$ be period-specific factors. $\alpha_i$ may capture, say, unobserved demographic characteristics, individual preferences, or innate ability. $\lambda_t$ may capture, e.g., unobserved macro-economic conditions, government policy, or environmental factors. These time-specific factors may be \emph{non-stationary} (distributed differently over time) and in particular, their distribution in the post-treatment period may \emph{differ} entirely from that in the pre-treatment period. We model the untreated potential outcomes and realized treatments as follows.
\begin{align}
Y_{i,t}(0) &= \mu^{(0)}(\alpha_i, \lambda_t) + u_{i,t},&\mathbb{E}[u_{i,t}\vert\alpha_i, \lambda_t]=0.\label{outcome_model}\\
w_{i,t} &= p_t(\alpha_i) + \epsilon_{i,t},&\mathbb{E}[\epsilon_{i,t}\vert\alpha_i,\lambda_t
]=0.\label{treatment_model}
\end{align}
The residuals $u_{i,t}$ and $\epsilon_{i,t}$ are unobserved. The functions $\mu^{(0)}(\cdot,\cdot)$ and $p_t(\cdot)$ are unknown, and, importantly, we do not assume that they have any particular functional form.  Note that the model above implies that the conditional mean function $\mu^{(0)}(\cdot,\cdot)$ is the same for every individual and every period. However, the conditional mean $\mu^{(0)}(\alpha_i,\lambda_t)$ may vary with both $i$ and $t$ through the latent factors. Notice that $p_t(\cdot)$, the conditional mean function of $w_{i,t}$, is allowed to vary over time, which reflects the fact that this variable is increasing and therefore subject to secular trends that need not be captured in $\lambda_t$.%

\begin{remark}{1}[Nonseparable models]
The additively separable models \eqref{outcome_model} and \eqref{treatment_model} may be derived from non-separable models under primitive conditions. For example, suppose that $Y_{i,t}(0)=q^{(0)}(\alpha_i,\lambda_t,v_{i,t})$ for some function $q^{(0)}(\cdot)$ and random variable $v_{i,t}$ that is independent of $\alpha_i$ and $\lambda_t$. If $v_{i,t}$ is exchangeable (separately in both $i$ and $t$), then \eqref{outcome_model} holds with
\begin{align*}
    \mu^{(0)}(\alpha_i, \lambda_t) \coloneqq \mathbb E[q^{(0)}(\alpha_i,\lambda_t,v_{i,t})|\alpha_i,\lambda_t],
\end{align*}
and the residual $u_{i,t}$ is mean-independent of $\alpha_i$ and $\lambda_t$ by construction.
\end{remark}

Without additional assumptions, causal quantities of interest are not identified from the model above. This would remain true even if the latent factors $\alpha_i$ and $\lambda_t$ were observed. This is because the residual in the treatment model (\ref{treatment_model}) may be correlated with the residual in the outcome equation (\ref{outcome_model}). In particular, there may be factors that influence both the outcome and treatment status (that is, confounding factors) and which are included in neither $\alpha_i$ nor $\lambda_t$. In order to achieve identification, we make the key assumption that the latent factors $\alpha_i$ and $\lambda_t$ together account for all confounding between the outcome and treatment. Formally, we assume that after controlling for these latent factors, there is no residual dependence between the untreated potential outcome and treatment.

\begin{assu}{2}[Latent unconfoundedness] \label{asm:2}
$Y_{i,t}(0)\indep w_{i,t}|\alpha_i,\lambda_t$.
\end{assu}
The assumption of latent unconfoundedness is common in the literature (e.g. \citealp{abadie2024doubly, agarwal2023causal, arkhangelsky2022doubly, athey2021matrix, fernandez2021low}).  Note that, since, in the spirit of the DiD and TWFE methods, we focus on estimating average treatment effects for the treated units, Assumption~\ref{asm:2} does not place restrictions on the treated potential outcomes. If we additionally require the treated potential outcomes to be conditionally independent of treatment, our method straightforwardly extends to cover other counterfactuals including average treatment effects (ATEs).  %

Under Assumption \ref{asm:2}, key counterfactual quantities of interest can be written in terms of the latent factors and functions $\mu^{(0)}(\cdot,\cdot)$ and $p_t(\cdot)$. For notational convenience, let $p_{i,t}:=p_t(\alpha_i) $ and $\mu^{(0)}_{i,t} \coloneqq \mu^{(0)}(\alpha_i, \lambda_t)$. Under Assumption \ref{asm:2} and a support overlap condition, the average effect of treatment on the treated at time $t \geq T_0 + 1$, which is defined as $\mathrm{ATT}_t \coloneqq \mathbb{E}_t[Y_{i,t}-Y_{i,t}(0)|w_{i,t}=1]$, can be expressed in the following doubly robust form:
\[
\mathrm{ATT}_t=\frac{1}{\mathbb{E}_t[w_{i,t}]}\mathbb{E}_t\bigg[Y_{i,t}w_{i,t}-\frac{(1-w_{i,t})Y_{i,t}p_{i,t}+(w_{i,t}-p_{i,t})\mu^{(0)}_{i,t}}{1-p_{i,t}}\bigg].
\]
The time subscript indicates that each expectation is taken with respect to the period $t$-specific distribution of the observables (i.e., conditional on $\lambda_t$). Related doubly robust approaches
to estimating ATT-type parameters in panel and DiD settings are studied
by, among others,
\citet{sant2020doubly,arkhangelsky2022doubly,chen2025efficient}. %

\begin{remark}{2}
For clarity and concreteness, we focus on the period-specific average treatment effect on the treated
$\mathrm{ATT}_t$. Under staggered treatment adoption, this parameter
coincides with the calendar-time average treatment effect considered by
\citet{callaway2021difference}. In particular, letting $G_i$ denote the
period in which unit $i$ is first treated and defining
\[
\mathrm{ATT}(g,t)
\coloneqq
\mathbb{E}_t
\left[
Y_{i,t}-Y_{i,t}(0)
\mid G_i=g
\right],
\]
we have
\[
\mathrm{ATT}_t
=
\sum_{g\leq t}
\mathbb{P}_t(G_i=g\mid G_i\leq t)\,
\mathrm{ATT}(g,t).
\]
Our analysis can accommodate other standard ATT-type estimands
analogously: one can estimate the cohort-time effects
$\mathrm{ATT}(g,t)$ using the corresponding doubly robust moments and
aggregate them with appropriate weights to obtain event-time,
cohort-specific, or overall treatment effects. However, in order to achieve identification in this case, Assumption \ref{asm:2} must be strengthened slightly to $Y_{i,t}(0)\indep G_{i}|\alpha_i,\lambda_t$ and one must impose a suitable support overlap condition.
\end{remark}

In this paper, we provide new methods for estimating the \emph{values} of $p_{i,t}$ and $\mu^{(0)}_{i,t}$. Given these estimates, one can construct a doubly robust estimate of $\mathrm{ATT}_t$ and perform inference on this object. Let $\hat{p}_{i,t}$ and $\hat{\mu}^{(0)}_{i,t}$ be estimates of $p_{i,t}$ and $\mu^{(0)}_{i,t}$ respectively. Then a corresponding doubly robust estimate of $\mathrm{ATT}_t$ takes the form%
\begin{equation}
\hat{\mathrm{ATT}}_t=\frac{1}{N_{1,t}}\sum_{i=1}^N\bigg(Y_{i,t}w_{i,t}-\frac{(1-w_{i,t})Y_{i,t}\hat{p}_{i,t}+(w_{i,t}-\hat{p}_{i,t})\hat{\mu}^{(0)}_{i,t}}{1-\hat{p}_{i,t}}\bigg), \label{ATT_estimator}
\end{equation}
where $N_{1,t}$ is the number of individuals treated by time $t$.

The arguments we use to consistently estimate $p_{i,t}$ and $\mu^{(0)}_{i,t}$ require an additional condition. Consider that exposure to aggregate/macroeconomic conditions may induce dependence between the potential outcomes of different individuals. We assume, in effect, that the time-specific factors $\lambda_t$ capture the influence of all such factors on untreated potential outcomes. The condition also imposes that at least in the absence of treatment, the outcomes and latent factors of an individual $j$ cannot impact the untreated outcomes of any distinct individual $i$.
\begin{assu}{3}[Conditional cross-sectional independence] \label{asm:3}
For any two individuals $i\neq j$, we have $Y_{i,t}{(0)}\indep Y_{j,t}{(0)}|\alpha_i,\alpha_j,\lambda_t$ and $Y_{i,t}{(0)}\indep \alpha_j|\alpha_i,\lambda_t$.
\end{assu}
Assumption \ref{asm:3} is crucial to our arguments. The condition ensures that for any two distinct individuals $i$ and $j$,
\[\mathbb E[\langle Y_{i,1:T_0},Y_{j,1:T_0}\rangle |\alpha_i,\alpha_j]= \mathbb E[\langle \mu^{(0)}_{i,1:T_0},\mu^{(0)}_{j,1:T_0}\rangle|\alpha_i,\alpha_j],\]
where $\langle\cdot,\cdot\rangle$ is the Euclidean inner product (i.e., $\langle a,b\rangle=a'b$ for two column vectors $a$ and $b$ of the same length). Thus, this assumption allows us to relate observable quantities, i.e., the history of pre-treatment outcomes, to the history of $\mu^{(0)}_{i,t}$, and thus to the individual-specific latent factors.

\subsection{Proposed method}

Our proposed estimation method uses pre-treatment outcomes to find untreated individuals whose latent factors $\alpha_i$ are similar to those of treated individuals. To motivate our approach, let us first suppose $\alpha_i$ were observed. Under Assumptions \ref{asm:1} and \ref{asm:2}, we have
\begin{equation*}
\mu^{(0)}(\alpha,\lambda_t)=\mathbb{E}_t[Y_{i,t}|\alpha_{i}=\alpha,w_{i,t}=0],
\hspace{40pt}
p_t(\alpha)=\mathbb{E}_t[w_{i,t}|\alpha_{i}=\alpha].
\end{equation*}

The objects on the right-hand sides of each equation above are regression functions. If $\alpha_{i}$ were observable, we could nonparametrically estimate $\mu^{(0)}(\cdot,\lambda_t)$ and $p_t(\cdot)$ (as functions of $\alpha$), and thus also $p_{i,t}$ and $\mu^{(0)}_{i,t}$. For example, we could obtain the following Nadaraya-Watson estimates:
\begin{align}
\label{eq:NW_observed_alpha}
\hat{\mu}_{i,t}^{(0)} &= \frac{\sum\limits_{j; w_{j,t} = 0} K(\|\alpha_i-\alpha_j\|/h)Y_{j,t}}{\sum\limits_{j; w_{j,t} = 0}K(\|\alpha_i-\alpha_j\|/h)},
\hspace{40pt}
\hat{p}_{i,t} = \frac{\sum\limits_{j} K(\|\alpha_i-\alpha_j\|/h) w_{j,t}}{\sum\limits_{j}K(\|\alpha_i-\alpha_j\|/h)},
\end{align}
where $K(\cdot)$ and $h$ are some kernel and bandwidth.

The estimates above are infeasible because in practice, we do not directly observe $\alpha_i$ for any individual $i$. In order to obtain feasible estimates, we replace the infeasible distance $\|\alpha_i-\alpha_j\|$ in the expressions above, with the following feasible pseudo-distance:%
\begin{align}
\label{eq:hat_d_def}
\hat{d}_{i,j}= \frac{1}{T_0} \max\limits_{k \notin \{i, j\}} \vert \langle Y_{k,1:T_0}, Y_{i,1:T_0}\rangle -\langle  Y_{k,1:T_0},Y_{j,1:T_0} \rangle \vert.
\end{align}
A pseudo-distance of the form above was first employed in \cite{zhang2017estimating} for nonparametric graphon estimation.

To motivate the use of this pseudo-distance, suppose Assumption \ref{asm:3} holds and that the aggregate factors $\lambda_t$ in each period $t\leq T_0$ are jointly independent of the individual-specific factors. Then the difference in conditional expected values of the inner-products in the pseudo-distance can be expressed as  %
\[
\mathbb{E}[\langle Y_{k,1:T_0}, Y_{i,1:T_0}\rangle|\alpha_k,\alpha_i] - \mathbb{E}[\langle Y_{k,1:T_0},Y_{j,1:T_0} \rangle|\alpha_k,\alpha_j]=\sum_{t=1}^{T_0} \int \mu^{(0)}(\alpha_k, \lambda) \big( \mu^{(0)}(\alpha_i, \lambda) - \mu^{(0)}(\alpha_j, \lambda) \big) \dd \pi_t(\lambda),
\]
where $\pi_t$ denotes the (unconditional) distribution of $\lambda_t$, which may vary with $t$, i.e., $\lambda_t$ does \emph{not} need to be stationary. As such, we can understand the pseudo-distance as the sample analogue of the population pseudo-distance %
\begin{align*}
d_{i,j} &= \sup_{\alpha \in \mathcal{A}} \left\vert \frac{1}{T_0}\sum_{t=1}^{T_0}\int \mu^{(0)}(\alpha, \lambda) \big( \mu^{(0)}(\alpha_i, \lambda) - \mu^{(0)}(\alpha_j, \lambda) \big) \dd \pi_t(\lambda) \right\vert\\
&= \sup_{\alpha \in \mathcal{A}} \left\vert \int \mu^{(0)}(\alpha, \lambda) \big( \mu^{(0)}(\alpha_i, \lambda) - \mu^{(0)}(\alpha_j, \lambda) \big) \dd \bar{\pi}_{T_0}(\lambda) \right\vert\\
&\geq \frac{1}{2}\int \big( \mu^{(0)}(\alpha_i, \lambda) - \mu^{(0)}(\alpha_j, \lambda) \big)^2 \dd \bar{\pi}_{T_0}(\lambda),
\end{align*}
where $\mathcal{A}$ is the support of $\alpha_i$, and $\bar{\pi}_{T_0} \coloneqq \frac{1}{T_0}\sum_{t=1}^{T_0}\pi_t$ is the average of the probability laws of $\lambda_t$ over the pre-treatment period; notice that we suppress the dependence of $d_{i,j}$ on $T_0$ for notational convenience.\footnote{If $\lambda_t$ is stationary during the pre-treatment period, we trivially have $\pi_t = \pi^*$ and $\overline \pi_{T_0} = \pi^*$, where $\pi^*$ denotes the stationary distribution of $\lambda_t$ for $t = 1, \ldots, T_0$.} The inequality above relates the size of the population pseudo-distance $d_{i,j}$ to a  squared $L^2$ distance (with respect to the measure $\bar{\pi}_{T_0}$) between the functions $\mu^{(0)}(\alpha_i, \cdot)$ and $\mu^{(0)}(\alpha_j, \cdot)$.\footnote{The inequality follows from taking $\alpha=\alpha_i$ and $\alpha=\alpha_j$ in the supremum and applying the triangle inequality.} Thus $d_{i,j}$ measures similarity of the latent factors to the extent that they impact untreated outcomes. Finally, while the population pseudo-distance $d_{i,j}$ is not directly observed, it can be consistently estimated by the sample pseudo-distance $\hat d_{i,j}$. This means that $\hat d_{i,j}$ also effectively measures the similarity between the functions $\mu^{(0)}(\alpha_i, \cdot)$ and $\mu^{(0)}(\alpha_j, \cdot)$, which, under an additional injectivity condition, also translates into the similarity between $\alpha_i$ and $\alpha_j$.

\begin{remark}{3}
    It is worth contrasting the sample pseudo-distance above with the (normalized) Euclidean distance between the histories of pre-treatment outcomes, i.e.,  $\frac{1}{T_0} \|Y_{i,1:T_0} - Y_{j,1:T_0}\|^2$. Under the previously employed conditional independence assumption, its population analogue can be expressed as %
\begin{align*}
\mathbb{E}\left[\frac{1}{T_0}\|Y_{i,1:T_0} - Y_{j,1:T_0}\|^2 \mid \alpha_i,\alpha_j\right]&=\int \left( \mu^{(0)}(\alpha_i, \lambda) - \mu^{(0)}(\alpha_j, \lambda)\right)^2\dd \bar{\pi}_{T_0}(\lambda)\\
&+\frac{1}{T_0}\sum_{s=1}^{T_0}\int \left(\mathbb E[u_{i,s}^2|\alpha_i,\lambda_s=\lambda]+\mathbb E[u_{j,s}^2|\alpha_j,\lambda_s=\lambda]\right) \dd {\pi}_{s}(\lambda).
\end{align*}
The above demonstrates that the (population) Euclidean distance is increasing in the conditional residual variance $\mathbb E[u_{j,t}^2|\alpha_j,\lambda_t=\lambda]$. Therefore, in the presence of conditional heteroskedasticity, a small Euclidean distance between a fixed unit $i$ and unit $j$ might not necessarily indicate that $\alpha_j$ is similar to $\alpha_i$: instead, it may simply reflect that unit $j$'s outcomes have a low residual variance. It is the need to be robust to conditional heteroskedasticity that motivates our use of the pseudo-distance \eqref{eq:hat_d_def}.
\end{remark}

Given the constructed pseudo-distance $\hat d_{ij}$, we obtain feasible estimates of $p_{i,t}$ and $\mu^{(0)}_{i,t}$:
\begin{align}
\label{eq:feasible_nonparametric_estimators}
\hat{\mu}_{i,t}^{(0)} &= \frac{\sum\limits_{j; w_{j,t} = 0} K(\hat{d}_{i,j}/h)Y_{j,t}}{\sum\limits_{j; w_{j,t} = 0}K(\hat{d}_{i,j}/h)},
\hspace{40pt}
\hat{p}_{i,t} = \frac{\sum\limits_{j} K(\hat{d}_{i,j}/h) w_{j,t}}{\sum\limits_{j}K(\hat{d}_{i,j}/h)}.
\end{align}
In principle, one could directly plug the estimates above into the formula for the doubly robust $\mathrm{ATT}_t$ estimate \eqref{ATT_estimator}. However, to further de-bias our estimate, we propose additionally combining \eqref{ATT_estimator} with a cross-fitting scheme. This is in line with the extensive literature on double machine learning, which demonstrates the utility of cross-fitting for reducing bias and obtaining valid inference (e.g., \citealp{chernozhukov2018double, abadie2024doubly}). The full algorithm with cross-fitting is detailed below along with a variance estimate and confidence interval. Our asymptotic results assume a fixed number of folds for cross-fitting and a deterministic bandwidth sequence. However, in practice, we select the bandwidth $h$ by cross-validation.\footnote{Implementation of the proposed method in R is also available at \url{https://github.com/chenweihsiang/PanelLatSim}.}

\begin{algorithm}[h!]
\caption{Doubly Robust Estimation and Inference with Cross-Fitting} \label{alg1}
\textbf{Inputs:} Number of folds $\mathcal{K}$, bandwidth $h$, confidence level $1-\zeta$ \\
\textbf{Returns:} estimate $\hat{\mathrm{ATT}}_t$, asymptotic variance estimate $\hat{V}_t$, and $1-\zeta$ level confidence interval\\\vspace{-20pt}
\begin{algorithmic}[1]
\STATE Randomly partition $[N] = \{1, 2, ..., N\}$ into $\mathcal{K}$ folds $\{\mathcal{I}_k\}_{k=1}^{\mathcal{K}}$ of size $\approx N / \mathcal{K}$. Let $\mathcal{I}_{-k}=[N]\setminus\mathcal{I}_k$
\FORALL{$k\in[\mathcal{K}]$}
\FORALL{$i \in \mathcal{I}_{k}$ and $j \in \mathcal{I}_{-k}$}
\STATE Calculate $\hat{d}_{i,j}$, the pseudo-distance between $i \in \mathcal{I}_{k}$ and $j \in \mathcal{I}_{-k}$, as follows:
\begin{align*}
\hat{d}_{i,j} \leftarrow
\frac{1}{T_0} \max_{\ell \notin \{i, j\} ; \ell \in \mathcal{I}_{-k}} \left\vert \left\langle Y_{\ell, 1:T_0}, Y_{i, 1:T_0} - Y_{j, 1:T_0} \right\rangle \right\vert
\end{align*}
\ENDFOR
\ENDFOR
\FORALL{ $k\in[\mathcal{K}]$ and $i \in \mathcal{I}_{k}$}
\STATE Calculate $\hat{\mu}_{i,t}^{(0)}$ and $\hat{p}_{i,t}$ as follows:
\begin{align*}
\hat{\mu}_{i,t}^{(0)} &\leftarrow \frac{\sum\limits_{j\in\mathcal{I}_{-k}; w_{j,t} = 0} K\left(\hat{d}_{i,j}/h\right)Y_{j,t}}{\sum\limits_{j\in\mathcal{I}_{-k}; w_{j,t} = 0}K\left(\hat{d}_{i,j}/h\right)},
\hspace{15pt}
\hat{p}_{i,t} \leftarrow \frac{\sum\limits_{j\in\mathcal{I}_{-k}} K\left(\hat{d}_{i,j}/h\right) w_{j,t}}{\sum\limits_{j\in\mathcal{I}_{-k}}K\left(\hat{d}_{i,j}/h\right)}
\end{align*}
\ENDFOR
\STATE Construct $\hat{\mathrm{ATT}}_t$ and $\hat V_t$ using the formulas below
\begin{align*}
\hat{\mathrm{ATT}}_t&\leftarrow\frac{1}{N_{1,t}}\sum_{i=1}^N\bigg(Y_{i,t}w_{i,t}-\frac{(1-w_{i,t})Y_{i,t}\hat{p}_{i,t}+(w_{i,t}-\hat{p}_{i,t})\hat{\mu}^{(0)}_{i,t}}{1-\hat{p}_{i,t}}\bigg)\\
\hat{V}_t&\leftarrow\frac{N}{N_{1,t}^2}\sum_{i=1}^N\bigg(Y_{i,t}w_{i,t}-\frac{(1-w_{i,t})Y_{i,t}\hat{p}_{i,t}+(w_{i,t}-\hat{p}_{i,t})\hat{\mu}^{(0)}_{i,t}}{1-\hat{p}_{i,t}}-w_{i,t}\hat{\mathrm{ATT}}_t\bigg)^2
\end{align*}
\STATE Construct $1-\zeta$ level confidence interval
$\big[\hat{\mathrm{ATT}}_t \pm z_{1-\zeta/2}\sqrt{\hat{V}_t/N}\big]$, where $z_{1-\zeta/2}$ is the $(1 - \zeta/2)$-th quantile of the standard normal distribution.
\end{algorithmic}
\end{algorithm}

\section{Large Sample Theory} \label{sec:3}

The estimator $\hat{\mathrm{ATT}}_t$ in Algorithm~\ref{alg1} is doubly robust and employs cross-fitting. An extensive literature provides sufficient conditions for $\sqrt{N}$-consistency of such estimators and asymptotically correct coverage of the corresponding confidence intervals (e.g., \citealp{robins2008higher, chernozhukov2018double, abadie2024doubly}). A key condition is that the first stage nuisance-parameter estimates (in our case $\hat{\mu}^{(0)}_{i,t}$ and $\hat{p}_{i,t}$) have sufficiently fast rates of convergence. In particular, this requirement is satisfied if %
\begin{align}
\label{eq:sufficient_rates}
\max_{i \in [N]} \left\vert
\hat{\mu}_{i,t}^{(0)} - \mu_{i,t}^{(0)}
\right\vert = o_p(N^{-1/4})\hspace{30pt}\text{and}\hspace{30pt} \max_{i \in [N]} \left\vert
\hat{p}_{i,t}- p_{i,t}
\right\vert =o_p(N^{-1/4}).
\end{align}

We establish convergence rates for the first-stage estimates. These rates may be of interest per se because $\mu_{i,t}^{(0)}$  is the optimal prediction of individual $i$'s untreated potential outcome given the latent factors. In particular, we show that, under certain conditions, $\hat{\mu}_{i,t}^{(0)}$ and $\hat p_{i,t}$ achieve, up to logarithmic factors, the optimal convergence rate of \citet{stone1980optimal} for nonparametric regression on $\alpha_i$ under Lipschitz continuity. That is, provided that the number of pre-treatment periods $T_0$ grows sufficiently quickly with $N$, our feasible estimators attain, up to logarithmic factors, the optimal pointwise rate available when $\alpha_i$ is observed.

\subsection{First-stage rates of convergence}%
\label{sec:3-1}

In order to derive convergence rates for our first-stage estimates $\hat \mu_{i,t}^{(0)}$ and $\hat p_{i,t}$, we impose the following additional assumptions.

\begin{assu}{4}[Model and latent factors] \label{asm:4}
The models (\ref{outcome_model}) and (\ref{treatment_model}) hold and satisfy the following conditions.
\hspace{0pt}\\\vspace{-25pt}
\begin{enumerate}[label=(\roman*), ref=(\roman*)]
\item\label{asm:4-4} The factors $\alpha_1,\alpha_2,...,\alpha_N,\lambda_1,\lambda_2,...,\lambda_T$ are jointly independent, and $\alpha_1,\alpha_2,...,\alpha_N$ are identically distributed. For each $i$, $\alpha_i$ has support $\mathcal{A}\subset\mathbb{R}^{d_\alpha}$ and there are constants $0<\underline{c}<\bar{c}<\infty  $ and $\bar{\eta}>0$ so that for any $0<\eta\leq\bar{\eta}$ and $\alpha\in\mathcal{A}$,  $\underline{c}\eta^{d_\alpha}\leq\mathbb{P}(\|\alpha-\alpha_{i}\|\leq\eta)\leq\bar{c}\eta^{d_\alpha}$.

\item\label{asm:4-2} There is some $\underline{c}>0$ so that $\underline{c}\leq p_{i,t} \leq 1-\underline{c}$ almost surely for all $i$ and $t>T_0$.
\item\label{asm:4-3} For some $\delta > 0$, $\mathbb{E}[\exp(\eta u_{i,t}) | \{\alpha_j\}_{j=1}^N, \{\lambda_s\}_{s=1}^T,\{w_{j,t}\}_{j=1}^N] \leq \exp(\delta \eta^2)$ for all $\eta \in \mathbb{R}$ almost surely for all $i$ and $t$. For each $i$ and $t$, we have $u_{i,t},\epsilon_{i,t}\indep\{\{u_{j,s},\epsilon_{j,s}\}_{s=1}^{t},\alpha_{j}\}_{j\neq i},\{u_{i,s}\}_{s=1}^{t-1},\{\lambda_{s}\}_{s\neq t}\big\vert\hspace{3pt}\alpha_{i},\lambda_t$.

\item\label{asm:4-5} The functions $\mu^{(0)}(\cdot,\cdot)$ and $p_t(\cdot)$ are uniformly bounded (on the relevant supports) and satisfy, for some ${L_0, L_p <\infty}$, %
\begin{align*}
 \vert \mu^{(0)}(\alpha_1, \lambda) - \mu^{(0)}(\alpha_2, \lambda) \vert &\leq L_0 \Vert \alpha_1 - \alpha_2 \Vert, \\
\vert p_t(\alpha_1) - p_t(\alpha_2) \vert &\leq L_p \Vert \alpha_1 - \alpha_2 \Vert,
\end{align*}
for all $\alpha_1, \alpha_2 \in \mathcal{A}$, $\lambda \in \mathcal{L}$, and $t$.\footnote{Here $\mathcal L:=\bigcup_t\operatorname{supp}(\lambda_t)$,
and we use $\operatorname{supp}(\lambda)$ as shorthand for $\mathcal L$
throughout.} %
\item\label{asm:4-6} $K: \mathbb{R}_{+} \rightarrow \mathbb{R}$ is weakly positive and strictly positive at zero, supported on a compact set and bounded by $\bar{K} < \infty$. Moreover, $\vert K(z) - K(z') \vert \leq \bar{K}' \vert z - z' \vert$ for all $z, z' \in \mathbb{R}_{+}$ for some $\bar{K}' > 0$.
\item\label{asm:4-8}
There exist constants $c_1,c_2>0$ so that if $T_0$ is sufficiently large, then for all $\alpha_1, \alpha_2 \in \mathcal{A}$
\begin{align}
 \int_{\lambda \in \supp(\lambda)}  \big( \mu^{(0)}(\alpha_1, \lambda) - \mu^{(0)}(\alpha_2, \lambda)  \big)^2\dd \bar{\pi}_{T_0}(\lambda)  \geq c_1\|\alpha_1-\alpha_2\|^2 \label{lowB}
\end{align}
and
\begin{align}
&\hspace{5pt} \sup_{\alpha\in\mathcal{A}}\big|\int_{\lambda \in \supp(\lambda)}  \mu^{(0)}(\alpha,\lambda)\big( \mu^{(0)}(\alpha_1, \lambda) - \mu^{(0)}(\alpha_2, \lambda)  \big)\dd \bar{\pi}_{T_0}(\lambda)\big| \nonumber\\
\geq&\hspace{5pt} c_2\sqrt{\int_{\lambda \in \supp(\lambda)}  \big( \mu^{(0)}(\alpha_1, \lambda) - \mu^{(0)}(\alpha_2, \lambda)  \big)^2\dd \bar{\pi}_{T_0}(\lambda)}.\label{closeness}
\end{align}

\end{enumerate}
\end{assu}

Assumption \ref{asm:4}\ref{asm:4-4} requires the individual-specific latent factors to be independent and identically distributed as if units are (independently) drawn from the underlying population. At the same time, while Assumption \ref{asm:4}\ref{asm:4-4} requires $\lambda_t$ to be serially independent (and independent of the individual factors), it does \emph{not} need to be stationary. Serial independence of $\lambda_t$ simplifies the exposition and proofs, but may be stronger than needed. For example, it is straightforward to allow for weakly dependent $\lambda_t$; see Remark~\ref{rem:relaxing_serial_independence}. Moreover, with a simple modification, our method can also explicitly incorporate additive time effects allowed to follow an arbitrary (potentially strongly dependent) process; see Remark~\ref{rem:additive_time_effects} for details. The final part of Assumption \ref{asm:4}\ref{asm:4-4} requires that with high probability, for each treated individual $i$, there exist untreated individuals in the sample whose latent factors are close to those of $i$, which  implies that the support of $\alpha_i$ is bounded.

Assumption \ref{asm:4}\ref{asm:4-2} is a standard overlap condition and would be required for regular estimation of the ATE even if $\alpha_i$ were observed. The bound away from one permits us to impute untreated potential outcomes for treated units using the outcomes of nearby untreated units, and it keeps the inverse weight $1/(1-p_{i,t})$ in the doubly robust score bounded. The bound away from zero ensures that the treated share $N_{1,t}/N$ is bounded away from zero, so that $\mathrm{ATT}_t$ can be estimated at the $\sqrt{N}$ rate.

Assumption \ref{asm:4}\ref{asm:4-3} enables us to apply particular concentration inequalities in order to obtain fast rates of convergence. The condition restricts the tail behavior of the error term $u_{i,t}$, requiring it to be uniformly sub-Gaussian conditional on some other variables. The assumption also imposes that, conditional on the latent factors $\alpha_i$ and $\lambda_t$, the residuals $u_{i,t}$ and $\epsilon_{i,t}$ are jointly independent of (a) the histories of these variables for all individuals other than $i$, (b) the individual-specific latent factors of individuals other than $i$, (c) lags of $u_{i,t}$, and (d) the aggregate factors $\lambda_s$ for $s\neq t$. In short, variables must be cross-sectionally independent after conditioning on the aggregate factors $\lambda_t$, and the current outcome and treatment residuals must be serially independent of the lagged outcome residuals after conditioning on individual-specific factors $\alpha_i$. One could relax this requirement and allow for serially correlated outcome residuals; see Remark~\ref{rem:relaxing_serial_independence} for details. The conditional independence restrictions in Assumption \ref{asm:4}\ref{asm:4-3} strengthen Assumption \ref{asm:3} to aid in the derivation of uniform convergence rates. While Assumption \ref{asm:3} is then redundant, we retain it for exposition because it plays a crucial role in identification.

Assumption \ref{asm:4}\ref{asm:4-5} imposes that the functions $\mu^{(0)}(\cdot,\cdot)$ and $p_t(\cdot)$ are bounded (on the relevant supports) and vary smoothly with the individual-specific factors.\footnote{Note that $p_t(\cdot)$ is bounded by construction. Boundedness of $\mu^{(0)}(\cdot,\cdot)$ is imposed for simplicity and can be relaxed. In particular, it is straightforward to allow for an unbounded common trend in the outcome model~\eqref{eq:Y_0 model}; see Remark~\ref{rem:additive_time_effects} for details.} This smoothness assumption ensures that individuals with similar latent factors also have similar conditional-mean potential outcomes and treatments.

Assumption \ref{asm:4}\ref{asm:4-6} stipulates  properties of the kernel $K(\cdot)$. Kernels that satisfy this condition are common in the literature, with the Epanechnikov kernel being a particularly prevalent choice.

Assumption \ref{asm:4}\ref{asm:4-8} relates the pseudo-distance to the distance between individual latent factors.\footnote{See \citet{sun2026flexible} for a detailed discussion of a network version of this assumption and primitive conditions sufficient for it to hold.} The assumption imposes, first, that the mean-squared distance between $\mu^{(0)}(\alpha_1,\cdot)$ and $\mu^{(0)}(\alpha_2,\cdot)$ is at least proportional to the squared Euclidean distance between $\alpha_1$ and $\alpha_2$. In addition, the assumption states that the population pseudo-distance is at least proportional to the square root of this mean-squared distance. Consider the special case in which $\mu^{(0)}(\alpha,\lambda)=\alpha'M\lambda$ for some fixed matrix $M$. In this case, the condition (\ref{lowB}) holds so long as the matrix $M\int_{\lambda\in\supp(\lambda)}\lambda\lambda'\dd\bar{\pi}_{T_0}(\lambda) M'$ is strictly positive definite, with its smallest eigenvalue bounded away from zero uniformly over all  sufficiently large $T_0$. If this is the case, then the second condition (\ref{closeness}) holds if there is a $c>0$ so that for any $\alpha_1,\alpha_2\in\mathcal{A}$ with $\alpha_1\neq\alpha_2$, we have $\frac{c}{\|\alpha_{1}-\alpha_{2}\|}(\alpha_{1}-\alpha_{2})\in \mathcal{A}$. This is true, for example, if $\mathcal{A}$ contains a Euclidean ball centered at zero.

\begin{thm}{1}[Rates of convergence] \label{thm:1}
Suppose Assumptions \ref{asm:1} to \ref{asm:4} hold and that $\frac{1}{h}\sqrt{\frac{\log(N)}{T_{0}}} \to 0$, $h\to0$, and $\frac{\log(N)}{h^{d_\alpha} N} \to 0$.  Then for $t>T_0$,\footnote{Because $T_0$ is increasing with $N$, $t$ is likewise understood to be a sequence increasing with $N$ so that $t>T_0$.}
\begin{align*}
\max_{i \in [N]} \left\vert
\hat{\mu}_{i,t}^{(0)} - \mu_{i,t}^{(0)}
\right\vert
= O_{p}\left(h+\sqrt{\frac{\log(N)}{h^{d_\alpha}N}}\right),\\
\max_{i \in [N]} \left\vert
\hat{p}_{i,t} - p_{i,t}
\right\vert
= O_{p}\left(h+\sqrt{\frac{\log(N)}{h^{d_\alpha}N}}\right).
\end{align*}
\end{thm}

Theorem~\ref{thm:1} establishes the uniform (over $i$) rate of convergence for the pseudo-distance based estimators \eqref{eq:feasible_nonparametric_estimators}. The derived rate is identical to the (pointwise) convergence rate of the infeasible Nadaraya-Watson estimators \eqref{eq:NW_observed_alpha} that use the unobserved factors as regressors, up to an extra log term needed to ensure uniformity over $i$. Setting $h\propto N^{-1/(d_\alpha+2)}$ (which automatically satisfies the restriction $\frac{\log(N)}{h^{d_\alpha} N} \to 0$), we get convergence rate $N^{-1/(d_\alpha+2)}\sqrt{\log(N)}$ which, up to a log term, is the optimal (pointwise) nonparametric learning rate under Lipschitz continuity of the regression function \citep{stone1980optimal}. However, the rate at which $h$ can converge to zero is also restricted by the condition that $\frac{1}{h}\sqrt{\frac{\log(N)}{T_{0}}} \to 0$. Thus the theorem above allows us to achieve the optimal rate, up to log terms, provided that  $T_{0}$ grows sufficiently quickly that $N^{2/(d_\alpha+2)}\log(N) / T_{0} \to 0$.

Recall that, in order to obtain $\sqrt{N}$-consistency and centered asymptotic normality of the doubly robust ATT estimator, it would be sufficient to demonstrate that \eqref{eq:sufficient_rates} holds. In particular, in the case of a one-dimensional individual latent factor ($d_\alpha=1$), Theorem \ref{thm:1} ensures that this holds for an appropriate choice of $h$ provided that $T_{0}$ grows quickly enough that $\frac{N^{1/2}\log(N)}{T_{0}} \to 0$. For $d_\alpha>1$, Theorem~\ref{thm:1} does not provide sufficiently fast rates for $\sqrt{N}$-consistency of the ${\rm ATT}_t$ estimator.

\begin{remark}{4}[Relaxing serial independence]
    \label{rem:relaxing_serial_independence}
    Assumption~\ref{asm:4} imposes serial independence of the time-specific factors $\lambda_t$ and residuals $u_{i,t}$. These conditions simplify the proof of Theorem~\ref{thm:1} but may be stronger than needed.

    In particular, these restrictions are used primarily to ensure that the sample pseudo-distance $\hat{d}_{i,j}$ is uniformly consistent for its population counterpart $d_{i,j}$, and to derive its rate of convergence. More precisely, the restrictions enable us to apply Bernstein-type inequalities (for independent random variables) to certain averages such as $\frac{1}{T_0}\sum_{t=1}^{T_0}u_{k,t}u_{i,t}$ and $\frac{1}{T_0}\sum_{t=1}^{T_0}\mu^{(0)}_{k,t}\mu^{(0)}_{i,t}$ to characterize their uniform (over $i$ and $k$) rates of convergence to their population means. Thus, one may invoke appropriate concentration inequalities for weakly dependent (rather than serially independent) variables instead.

\end{remark}

\begin{remark}{5}[Including additive time effects]
    \label{rem:additive_time_effects}
    Our assumptions do not rule out the possibility that some of the time-factors enter the model additively. However, the pseudo-distance in~\eqref{eq:hat_d_def} is not invariant to these additive factors, and our stated formal results require that they satisfy the same conditions we place on non-additive factors. However, by  adjusting  the pseudo-distance we can eliminate additive time-effects entirely and thus avoid the need to place any conditions on these objects. In particular, $\gamma_t$ can follow an arbitrary (potentially strongly dependent) process and be unbounded. To see this, suppose we include  additive time-effects $\gamma_t$ in model (\ref{outcome_model}). That is, suppose that
    \[
    Y_{i,t}(0)=\mu^{(0)}(\alpha_i,\lambda_t)+\gamma_t+u_{i,t}.
    \]
    The time-effects can be eliminated (in the pre-treatment period) in the usual way, by simply subtracting cross-sectional average outcomes. The corresponding sample and population pseudo-distances that no longer depend on the additive effects $\gamma_t$ can be constructed as
    \begin{align*}
    \hat{d}_{i,j}&= \frac{1}{T_0} \max\limits_{k \notin \{i, j\}} \vert \langle Y_{k,1:T_0}-\bar{Y}_{1:T_0}, Y_{i,1:T_0} - Y_{j,1:T_0} \rangle \vert,\\
    d_{i,j} &= \sup_{\alpha \in \mathcal{A}} \left\vert \int \big(\mu^{(0)}(\alpha, \lambda)-\bar{\mu}^{(0)}(\lambda)\big) \big( \mu^{(0)}(\alpha_i, \lambda) - \mu^{(0)}(\alpha_j, \lambda) \big) \dd \bar{\pi}_{T_0}(\lambda) \right\vert,
    \end{align*}
    where $\bar{Y}_{1:T_0}:=(\frac{1}{N}\sum_{i=1}^NY_{i,1},\frac{1}{N}\sum_{i=1}^NY_{i,2},...,\frac{1}{N}\sum_{i=1}^NY_{i,T_0})'$ and $\bar{\mu}^{(0)}(\lambda)=\int \mu^{(0)}(\alpha,\lambda)d\pi_{\alpha}$ for $\pi_\alpha$ the distribution of $\alpha_k$. One could adapt Theorem 1 (and the required assumptions) to apply with the alternative sample and population pseudo-distances above.
\end{remark}

\subsection{Asymptotic normality}%
\label{sec:3-2}
In this section, we apply well-established ideas from the literature on double machine learning to the studied setting and demonstrate that, under certain conditions, our ATT estimator (with cross-fitting as specified in Algorithm \ref{alg1}) is $\sqrt{N}$-consistent and has a centered asymptotically normal distribution.

\begin{thm}{2}[Asymptotic normality] \label{thm:2}
Suppose that ${\max_{i\in[N]} \left\vert \hat{p}_{i,t}-p_{i,t} \right\vert = o_{p}(1)}$, $\max_{i\in[N]} \left\vert \hat{\mu}^{(0)}_{i,t}-\mu^{(0)}_{i,t} \right\vert = o_{p}(1)$, and that ${\max_{i\in[N]} \left\vert (p_{i,t}-\hat{p}_{i,t})(\mu_{i,t}^{(0)}-\hat{\mu}_{i,t}^{(0)}) \right\vert = o_{p}\left(N^{-1/2}\right)}$. In addition, suppose that the number of folds $\mathcal{K}\geq2$ is fixed and, for each $k$, the size of fold $\mathcal{I}_{k}$  grows at the same rate as the sample size, that is, $\frac{|\mathcal{I}_{k}|}{N}\to c_{k}$ for some finite $c_{k}>0$. Then, under Assumptions \ref{asm:1}-\ref{asm:4}, we have that for $t>T_0$,
\begin{align*}
\sqrt{N}\hat{\mathrm{ATT}}_{t}=\frac{\sqrt{N}}{N_{1,t}}\sum_{i=1}^{N}\left(Y_{i,t}w_{i,t}-\frac{(1-w_{i,t})Y_{i,t}p_{i,t}+(w_{i,t}-p_{i,t})\mu_{i,t}^{(0)}}{1-p_{i,t}}\right) + o_{p}(1),
\end{align*}
in which case, if the usual conditions on the doubly robust ATT score $\psi_{i,t}$ apply conditional on $\lambda_t$,\footnote{Namely, that for each $t$, $\psi_{i,t}$ is iid across $i$ conditional on $\lambda_t$, with $\mathbb{E}_t[\psi_{i,t}^2]\geq c$ and $\mathbb{E}[|\psi_{i,t}|^{2+\rho}]\leq c^{-1}$ for a fixed $\rho>0$ and  $0<c<\infty$.}  we obtain $\sqrt{N} \left(\hat{\mathrm{ATT}}_{t}-\mathrm{ATT}_{t}\right)/\sqrt{V_t} \overset{d}{\rightarrow} N\left(0,1\right)$, where $V_t$ is the usual asymptotic variance of the doubly robust ATT estimator, i.e.,
\[V_{t}=\mathbb E_t[\psi_{i,t}^2],\,\,\,\,\psi_{i,t}:=\frac{1}{\mathbb E_{t}[w_{i,t}]}\left(Y_{i,t}w_{i,t}-\frac{(1-w_{i,t})Y_{i,t}p_{i,t}+(w_{i,t}-p_{i,t})\mu_{i,t}^{(0)}}{1-p_{i,t}} - w_{i,t} \mathrm{ATT}_{t} \right).\]
\end{thm}

The first result in Theorem \ref{thm:2} relates the estimate $\hat{\mathrm{ATT}}_t$ to an infeasible estimate in which the first stage estimates of $\mu_{i,t}^{(0)}$ and $p_{i,t}$ are replaced by their true values. In particular, the theorem states that under certain conditions, the difference between $\hat{\mathrm{ATT}}_t$ and the infeasible estimate disappears strictly faster than $N^{-1/2}$. It then follows from standard results that $\hat{\mathrm{ATT}}_t$ is $\sqrt{N}$-asymptotically normal and centered at $\mathrm{ATT}_t$.\footnote{Under the conditions of Theorem~\ref{thm:2}, $\hat V_t/V_t\overset{p}{\longrightarrow}1$, so the confidence interval in Algorithm~\ref{alg1} has asymptotic coverage $1-\zeta$.}

The proof of Theorem \ref{thm:2} proceeds by similar steps to the results in \cite{chernozhukov2018double} for DML2 estimators. A complicating factor is that, unlike in standard DML2 estimation of the ATT, in our setting the estimates $\hat{\mu}^{(0)}_{i,t}$ and $\hat{p}_{i,t}$ are constructed using the pseudo-distance which includes outcome data for individual $i$ in the pre-treatment period. That is, unlike in standard DML2 ATT estimation, the first-stage estimates $\hat{\mu}^{(0)}_{i,t}$ and $\hat{p}_{i,t}$ depend on some of the outcome data from the fold that contains individual $i$, which requires a more careful analysis.\footnote{In particular, it is straightforward to accommodate this dependence under the conditional serial independence of the errors imposed in Assumption \ref{asm:4}. The provided argument also naturally extends to the case when the errors are weakly dependent.}

Theorem \ref{thm:2} requires not only that the first-stage estimates are uniformly consistent over the sample, but also that the product of their estimation errors is uniformly $o_p(N^{-1/2})$. A sufficient condition,  mentioned earlier in this section, is that each of these estimates converges strictly faster than $N^{-1/4}$. The theorem further requires that each fold grows at the same rate as the sample size. This is satisfied if the number of folds is fixed and all of the folds are of equal size. In recent work,  \cite{velez2024asymptoticpropertiesdebiasedmachine} provides conditions under which DML2 estimates are centered root-$N$ asymptotically normal even if leave-one-out cross-fitting is applied, and so it may be possible to weaken this condition on the fold size in Theorem \ref{thm:2}.

\begin{remark}{6}[Semiparametric efficiency]
   The score $\psi_{i,t}$ in Theorem \ref{thm:2} is the well-known efficient influence function for the ATT in cross-sectional settings with \emph{observed} unit characteristics $\alpha_i$ \citep{hahn1998role}. Therefore, under the conditions of the theorem, our estimator achieves the semiparametric efficiency bound for estimation of the ATT in the cross-sectional case with observed $\alpha_i$, despite the fact that $\alpha_i$ is latent in our framework.
\end{remark}

\section{Simulation Study} \label{sec:4}

In this section, we illustrate the finite-sample properties of the proposed method in two sets of numerical experiments. The first compares the proposed method with the workhorse two-way fixed effects (TWFE) approach, both when the additive TWFE specification is correct and when it is misspecified. The second compares the proposed method with the local PCA method of \citet{feng2026causal} under the nonlinear data generating processes.

\subsection{Simulation with Linear DGPs} \label{sec:4-1}

The TWFE approach is widely used in applied research and is designed to adjust for unobserved additive individual and time fixed effects. In this section, we demonstrate that the proposed method performs comparably well in settings in which the TWFE approach is valid. We also show that the proposed method maintains good finite sample properties in a more complicated setting in which the TWFE approach fails.

We consider a panel with individuals $i = 1, \dots, N$ and time periods $t = 1, \dots, T$, where $N \in \{50, 250\}$, the number of pre-treatment periods is $T_0 \in \{50, 250\}$, and $T = T_0 + 1$, so that there is exactly one post-treatment period. For the outcome we consider the following data generating processes:
\begin{align*}
\begin{cases}
\text{Model 1 (Additive fixed effects): } &Y_{i,t} = \alpha_i + \lambda_t + 0.5 \cdot w_{i,t} + u_{i,t}, \\
\text{Model 2 (Interactive fixed effects): } &Y_{i,t} = \alpha_i \cdot \lambda_t + 0.5 \cdot w_{i,t} + u_{i,t},
\end{cases}
\end{align*}
where $\alpha_i \sim \mathrm{Uniform}(-1, 1)$, $\lambda_t \sim \mathrm{Uniform}(-1, 1)$, and $u_{i,t} \sim N(0, 0.5^2)$ are drawn independently of one another and i.i.d. across $i$ and $t$. The target estimand is the average treatment effect on the treated in the post-treatment period, $\mathrm{ATT}_T = 0.5$. Models 1 and 2 are standard additive and interactive fixed effects specifications.

The treatment assignment mechanism depends on unobserved latent characteristics for each $i$, and the treatment occurs only in the final period. $w_{i,t} = 0$ for all $t \leq T_0$, and the propensity score for the last period $p_{i,T}$ follows the model
\begin{align*}
w_{i,T} \mid \alpha_i \sim \mathrm{Bernoulli}(p_{i,T}),
\qquad
p_{i,T} = \frac{\exp(\alpha_i)}{1 + \exp(\alpha_i)}.
\end{align*}
Because $\alpha_i\in[-1,1]$, we have approximately $p_{i,T} \in [0.2689, 0.7311]$, which satisfies the overlap condition required for the proposed method.

We simulate $1000$ replications from these models. In these simulations, we utilize the Epanechnikov kernel, which satisfies the requirements of Theorem \ref{thm:1} and is defined as $K(x) = \frac{3}{4} (1 - x^2) \cdot \mathbf{1}\{\vert x \vert \leq 1 \}$. We use this kernel for both the outcome and propensity score imputations in our proposed method. Using this kernel, we compare the performance of different methods: (i) TWFE; (ii) the local PCA method proposed by \cite{feng2026causal} with 5-fold cross-validation for the choice of nearest neighbors;\footnote{The implementation of the method proposed by \cite{feng2026causal} follows the replication files available on the author's website: \url{https://github.com/yingjieum/replication-Feng_2026}.}\footnote{There are some minor numerical issues in the implementation of \cite{feng2026causal}. During the simulation, LocalPCA occasionally fails to produce estimates during certain replications. This issue occurred in fewer than 10 replications for each different simulation design across various $N,T_0$ settings. Since the overall results are not significantly affected by these minor errors and we prefer to maintain the original implementation, we have decided to drop the NAs while averaging for the summary statistics. This also explains why some of the coverage results are not multiples of 0.001. Also note that the proposed method's implementation produced no errors in any of the simulations' replications.} (iii) an infeasible doubly robust estimator in which the pseudo-distance is replaced by the oracle $\ell_2$ distance between the true $\alpha$'s; (iv)~our proposed method.

Two sample-splitting procedures are used for different purposes. \emph{Cross-fitting} refers to the sample splitting used when imputing the nuisance quantities $\hat\mu^{(0)}_{i,T}$ and $\hat p_{i,T}$ to control overfitting in doubly robust estimation. \emph{Cross-validation} refers to the choice of the bandwidth~$h$, which is further implemented with grid search and evaluated via post-treatment period squared prediction error of the imputed outcomes. Folds are drawn once per replication, and since $N$ can be quite small in some settings, redrawing is sometimes needed if any training set contains fewer than two treated or two control units.

In the numerical experiments, the bandwidth is selected between $0.05$ and $5$. In Tables \ref{tab:sim_1} and \ref{tab:sim_2}, we report the 95\% confidence interval coverage and other statistics across $1000$ replications in our simulation for each outcome model and each method. Additionally, estimated propensity scores close to one can generate large inverse-odds weights $\hat{p}_{i,T}/(1-\hat{p}_{i,T})$ on untreated observations and hence very large estimates. Such replications are rare but heavy enough to distort mean-based statistics, and the effect is more pronounced when $N$ is small. For this reason we report the median absolute error, the median of $\vert\text{estimate} - \text{true value}\vert$, and the median confidence interval length rather than their means. %

\begin{table}[h!]%
\begin{adjustwidth}{-1cm}{-1cm}
\begin{center}
\begin{threeparttable}
\caption{Simulation results for Model 1 (additive fixed effects)}\label{tab:sim_1}
\begin{small}
\addtolength{\tabcolsep}{-0pt}
\renewcommand{\arraystretch}{1.0}
\begin{tabular}{cc|c|c|ccc}
\toprule
                                                                                                       &                           & \multirow{2}{*}{\textbf{TWFE}} & \textbf{Local PCA}                   & \multicolumn{3}{c}{\textbf{Proposed Method}}                                         \\
                                                                                                       &                           &                                & \textbf{(\citealp{feng2026causal})} & \textbf{W/ true $\boldsymbol{\alpha_i}$} & \multicolumn{2}{c}{\textbf{Pseudo-Distance}} \\ \cmidrule(lr){3-3} \cmidrule(lr){4-4} \cmidrule(lr){5-5} \cmidrule(lr){6-7}
\multicolumn{2}{c|}{\textbf{Cross-Fitting}}                                                                                         &                                &                                      & \textbf{LOO}                             & \textbf{None}         & \textbf{5-Fold}     \\
\multicolumn{2}{c|}{\textbf{Bandwidth Selection}}                                                                                   &                                &                                      & \textbf{LOO}                             & \textbf{LOO}        & \textbf{5-Fold}     \\ \midrule
\multirow{3}{*}{\begin{tabular}[c]{@{}c@{}}$\boldsymbol{N=50}$\\ $\boldsymbol{T_0=50}$\end{tabular}}   & Median Abs. Err.          & 0.0951                         & 0.1359                               & 0.1233                                   & 0.1160              & 0.1342              \\
                                                                                                       & \textbf{95\% CI Coverage} & \textbf{0.9500}                & \textbf{0.8862}                      & \textbf{0.9580}                          & \textbf{0.9030}     & \textbf{0.9480}     \\
                                                                                                       & Median CI Length          & 0.5642                         & 0.6065                               & 0.6897                                   & 0.5490              & 0.7121              \\ \midrule
\multirow{3}{*}{\begin{tabular}[c]{@{}c@{}}$\boldsymbol{N=50}$\\ $\boldsymbol{T_0=250}$\end{tabular}}  & Median Abs. Err.          & 0.0935                         & 0.1340                               & 0.1191                                   & 0.1067              & 0.1271              \\
                                                                                                       & \textbf{95\% CI Coverage} & \textbf{0.9450}                & \textbf{0.8740}                      & \textbf{0.9540}                          & \textbf{0.8980}     & \textbf{0.9560}     \\
                                                                                                       & Median CI Length          & 0.5588                         & 0.5919                               & 0.7027                                   & 0.5568              & 0.7203              \\ \midrule
\multirow{3}{*}{\begin{tabular}[c]{@{}c@{}}$\boldsymbol{N=250}$\\ $\boldsymbol{T_0=50}$\end{tabular}}  & Median Abs. Err.          & 0.0442                         & 0.0508                               & 0.0490                                   & 0.0475              & 0.0483              \\
                                                                                                       & \textbf{95\% CI Coverage} & \textbf{0.9440}                & \textbf{0.9320}                      & \textbf{0.9570}                          & \textbf{0.9250}     & \textbf{0.9450}     \\
                                                                                                       & Median CI Length          & 0.2507                         & 0.2736                               & 0.2789                                   & 0.2608              & 0.2827              \\ \midrule
\multirow{3}{*}{\begin{tabular}[c]{@{}c@{}}$\boldsymbol{N=250}$\\ $\boldsymbol{T_0=250}$\end{tabular}} & Median Abs. Err.          & 0.0442                         & 0.0539                               & 0.0514                                   & 0.0498              & 0.0504              \\
                                                                                                       & \textbf{95\% CI Coverage} & \textbf{0.9590}                & \textbf{0.9410}                      & \textbf{0.9580}                          & \textbf{0.9400}     & \textbf{0.9520}     \\
                                                                                                       & Median CI Length          & 0.2488                         & 0.2713                               & 0.2801                                   & 0.2617              & 0.2828              \\
\bottomrule
\end{tabular}
\end{small}
\begin{tablenotes}
\footnotesize
\setlength\labelsep{0pt}
\item \emph{Notes}: 1,000 replications were conducted in each setting, where the estimand is $\mathrm{ATT}_T = 0.5$. LOO and 5-fold denote leave-one-out and five-fold procedures, respectively, for cross-fitting or bandwidth cross-validation as indicated by the corresponding row; the bandwidth is chosen from the grid by minimizing post-period squared prediction error. If the selected bandwidth yields an empty kernel neighborhood, the next larger grid value is used.  "None" indicates no cross-fitting.  Local PCA follows the replication package of \citet{feng2026causal} with $K$ chosen by 5-fold cross-validation.
\end{tablenotes}
\end{threeparttable}
\end{center}
\end{adjustwidth}
\end{table}

For Model 1, the additive fixed effects model, the proposed method shows performance similar to the two-way fixed effects method in terms of 95\% confidence interval coverage, which is under an ideal data generating process for TWFE. Additionally, the proposed method, which incorporates the pseudo-distance, effectively captures the distance between unobserved latent characteristics, which is evinced by the performance of our method relative to the oracle estimator which uses the true distance in the $\alpha$'s. However, there is some loss of efficiency when the sample size $N$ is small. This efficiency loss could perhaps be mitigated by increasing the number of folds or by using the leave-one-out procedure for bandwidth selection.

\begin{table}[h!]%
\begin{adjustwidth}{-1cm}{-1cm}
\begin{center}
\begin{threeparttable}
\caption{Simulation results for Model 2 (interactive fixed effects)}\label{tab:sim_2}
\begin{small}
\addtolength{\tabcolsep}{-0pt}
\renewcommand{\arraystretch}{1.0}
\begin{tabular}{cc|c|c|ccc}
\toprule
                                                                                                       &                           & \multirow{2}{*}{\textbf{TWFE}} & \textbf{Local PCA}                   & \multicolumn{3}{c}{\textbf{Proposed Method}}                                         \\
                                                                                                       &                           &                                & \textbf{(\citealp{feng2026causal})} & \textbf{W/ true $\boldsymbol{\alpha_i}$} & \multicolumn{2}{c}{\textbf{Pseudo-Distance}} \\ \cmidrule(lr){3-3} \cmidrule(lr){4-4} \cmidrule(lr){5-5} \cmidrule(lr){6-7}
\multicolumn{2}{c|}{\textbf{Cross-Fitting}}                                                                                         &                                &                                      & \textbf{LOO}                             & \textbf{None}         & \textbf{5-Fold}     \\
\multicolumn{2}{c|}{\textbf{Bandwidth Selection}}                                                                                   &                                &                                      & \textbf{LOO}                             & \textbf{LOO}        & \textbf{5-Fold}     \\ \midrule
\multirow{3}{*}{\begin{tabular}[c]{@{}c@{}}$\boldsymbol{N=50}$\\ $\boldsymbol{T_0=50}$\end{tabular}}   & Median Abs. Err.          & 0.1725                         & 0.1263                               & 0.1189                                   & 0.1041              & 0.1228              \\
                                                                                                       & \textbf{95\% CI Coverage} & \textbf{0.8490}                & \textbf{0.8999}                      & \textbf{0.9450}                          & \textbf{0.9110}     & \textbf{0.9450}     \\
                                                                                                       & Median CI Length          & 0.6780                         & 0.5885                               & 0.6290                                   & 0.5410              & 0.6415              \\ \midrule
\multirow{3}{*}{\begin{tabular}[c]{@{}c@{}}$\boldsymbol{N=50}$\\ $\boldsymbol{T_0=250}$\end{tabular}}  & Median Abs. Err.          & 0.1724                         & 0.1215                               & 0.1108                                   & 0.1036              & 0.1163              \\
                                                                                                       & \textbf{95\% CI Coverage} & \textbf{0.8120}                & \textbf{0.8899}                      & \textbf{0.9370}                          & \textbf{0.8920}     & \textbf{0.9320}     \\
                                                                                                       & Median CI Length          & 0.6728                         & 0.5850                               & 0.6337                                   & 0.5468              & 0.6432              \\ \midrule
\multirow{3}{*}{\begin{tabular}[c]{@{}c@{}}$\boldsymbol{N=250}$\\ $\boldsymbol{T_0=50}$\end{tabular}}  & Median Abs. Err.          & 0.1644                         & 0.0519                               & 0.0485                                   & 0.0507              & 0.0512              \\
                                                                                                       & \textbf{95\% CI Coverage} & \textbf{0.4500}                & \textbf{0.9410}                      & \textbf{0.9510}                          & \textbf{0.9250}     & \textbf{0.9380}     \\
                                                                                                       & Median CI Length          & 0.3014                         & 0.2693                               & 0.2687                                   & 0.2544              & 0.2704              \\ \midrule
\multirow{3}{*}{\begin{tabular}[c]{@{}c@{}}$\boldsymbol{N=250}$\\ $\boldsymbol{T_0=250}$\end{tabular}} & Median Abs. Err.          & 0.1501                         & 0.0506                               & 0.0498                                   & 0.0490              & 0.0504              \\
                                                                                                       & \textbf{95\% CI Coverage} & \textbf{0.5000}                & \textbf{0.9390}                      & \textbf{0.9480}                          & \textbf{0.9300}     & \textbf{0.9360}     \\
                                                                                                       & Median CI Length          & 0.2991                         & 0.2685                               & 0.2690                                   & 0.2555              & 0.2692              \\
\bottomrule
\end{tabular}
\end{small}
\begin{tablenotes}
\footnotesize
\setlength\labelsep{0pt}
\item \emph{Notes}: Details are provided in the notes to Table 1.
\end{tablenotes}
\end{threeparttable}
\end{center}
\end{adjustwidth}
\end{table}

For Model 2, the interactive fixed effects model, the TWFE estimator is misspecified. Its median absolute error remains between roughly $0.15$ and $0.17$ regardless of the sample size, while its confidence intervals shrink as $N$ grows. As a result, the TWFE method's coverage deteriorates sharply with $N$. In contrast, the proposed method maintains coverage close to the nominal level for all considered settings of $N$ and $T_0$.

\FloatBarrier
\subsection{Simulation with More Complex Nonlinear DGPs} \label{sec:4-2}

We compare the performance of our approach with the method proposed by \cite{feng2026causal}, which, to the best of our knowledge, is one of the few available inferential methods for average causal effects in nonlinear factor models. In this subsection, we use the same data generating processes employed in \cite{feng2026causal}\footnote{The replication package of \cite{feng2026causal} was originally released in 2024 on the GitHub page. The author later updated the package in 2026, making changes to both the simulation DGPs and the implementation of the LocalPCA method. Since the 2024 version includes a wider variety of nonlinear DGPs, we retain the simulation designs from that version. However, the LocalPCA implementation follows the exact same code as the most up-to-date 2026 paper and package. We also set the standard deviation of the pre-treatment disturbance to $\sigma = 0.5$, so that the pre-treatment signal-to-noise ratio is more favorable to all methods; results with $\sigma = 1$ are qualitatively similar.} but we consider different values of $N$ and $T$. The implementation is based on the replication code provided by the author. These numerical exercises are designed to investigate how the methods that we consider perform under relatively small sample sizes, which may be more typical of microeconomic applications. %

Again, consider a large panel data setting with individuals labeled $i = 1, ..., N$ over time periods $t = 1, ..., T$, where $N \in \{50, 250, 1000\}$ and $T_0 \in \{50, 250, 1000\}$ with only one post-treatment time period. Our target estimand is the counterfactual mean $\mathbb{E}[Y_{i,T}(0) | w_{i,T} = 1]$ which is the object of interest in the simulations of \cite{feng2026causal}. Estimation and inference for this counterfactual mean can be carried out in a manner analogous to those for the $\text{ATT}_t$. As with the $\text{ATT}_t$, one can plug our estimates of $\mu_{i,t}^{(0)}$ and $p_{i,t}$ into the standard doubly robust formulae.\footnote{To be precise, in this case the relevant estimator is
$\hat{\theta}_t=\frac{1}{N_{1,t}}\sum_{i=1}^{N}Y_{i,t}w_{i,t}-\hat{\rm ATT}_t$ and the variance estimate is
$
\frac{N}{N_{1,t}^{2}}\sum_{i=1}^{N}\bigg(\frac{(1-w_{i,t})Y_{i,t}\hat{p}_{i,t}+(w_{i,t}-\hat{p}_{i,t})\hat{\mu}_{i,t}^{(0)}}{1-\hat{p}_{i,t}}-w_{i,t}\hat{\theta}_{t}\bigg)^{2}
$. For TWFE, here and in Section~\ref{sec:5}, we estimate the counterfactual mean as the mean post-treatment outcome among treated units minus the estimated treatment coefficient, with standard errors based on the influence function of this difference.} We consider the following DGPs for the untreated potential outcomes:%
\begin{align*}
\begin{cases}
\text{Model 3 (Nonlinear factor model): } &Y_{i,t}(0) = (\alpha_i - \lambda_t)^2 + u_{i,t}, \\
\text{Model 4 (Gaussian kernel): } &Y_{i,t}(0) = \frac{1}{0.1 \cdot \sqrt{2 \pi}} \exp\left[ -100(\alpha_i - \lambda_t)^2 \right] + u_{i,t}, \\
\text{Model 5 (Exponential kernel): } &Y_{i,t}(0) = \exp\left( -10 \vert \alpha_i - \lambda_t \vert \right) + u_{i,t},
\end{cases}
\end{align*}
for $t \in \{1, ..., T_0\}$. In the post-treatment time period, all three models follow
\begin{align*}
Y_{i,T}(0) = \alpha_i + \alpha_i^2 + \varepsilon_{i,T}^{(0)},\,\,\,\,\,\,\,\,\,\,
Y_{i,T}(1) = 2 \alpha_i + \alpha_i^2 + 1 + \varepsilon_{i,T}^{(1)},
\end{align*}
where the latent factors and outcome residuals are drawn independently by $\alpha_i \sim \rm{Uniform}(0, 1)$, $\lambda_t \sim \rm{Uniform}(0, 1)$, $u_{i,t} \sim N(0, 0.5^2)$, and $\varepsilon_{i,t}^{(w)} \sim N(0, 1)$. Treatment is drawn independently across individuals by $w_{i,T}|\alpha_i\sim \text{Bernoulli}(p_{i,T})$, with the propensity score given by %
\begin{align*}
p_{i,T} = \frac{\exp\left[ (\alpha_i - 0.5) + (\alpha_i - 0.5)^2 \right]}{1 + \exp\left[ (\alpha_i - 0.5) + (\alpha_i - 0.5)^2 \right]}.
\end{align*}
The design of the DGP ensures that approximately $p_{i,T} \in [0.4378, 0.6792]$, which satisfies the required overlap condition.

The three designs differ only in the function that generates the pre-treatment history. The post-treatment outcome equations, the latent variables and the assignment mechanism are the same. The comparison therefore isolates how the quality of the pseudo-distance depends on the shape of the pre-treatment structure. Model 3 is a smooth polynomial factor model. Models 4 and 5 are Gaussian- and exponential-kernel designs, which are harder to approximate. A similar simulation design is also used in \cite{fernandez2021low}.%

The results are provided in Tables \ref{tab:sim_3}, \ref{tab:sim_4}, and \ref{tab:sim_5} below.\footnote{Note that the oracle benchmark uses the true individual factors $\alpha_i$ and post-treatment data. Within each model and for each fixed $N$, we use the same random seeds in corresponding replications to generate the individual factors, treatment assignments, and post-treatment potential outcomes across all $T_0$ settings. Since the benchmark does not use the pre-treatment history, its reported results remain identical as $T_0$ varies within each $N$ block of a given table.} We consider the same methods and report the same statistics as for the previous numerical experiments. The number of replications is $1000$.

\begin{table}[h!]%
\begin{adjustwidth}{-1cm}{-1cm}
\begin{center}
\begin{threeparttable}
\caption{Simulation results for Model 3 (nonlinear factor model)}\label{tab:sim_3}
\begin{small}
\addtolength{\tabcolsep}{-0pt}
\renewcommand{\arraystretch}{0.85}
\begin{tabular}{cc|c|c|ccc}
\toprule
                                                                                                       &                           & \multirow{2}{*}{\textbf{TWFE}} & \textbf{Local PCA}                   & \multicolumn{3}{c}{\textbf{Proposed Method}}                                         \\
                                                                                                       &                           &                                & \textbf{(\citealp{feng2026causal})} & \textbf{W/ true $\boldsymbol{\alpha_i}$} & \multicolumn{2}{c}{\textbf{Pseudo-Distance}} \\ \cmidrule(lr){3-3} \cmidrule(lr){4-4} \cmidrule(lr){5-5} \cmidrule(lr){6-7}
\multicolumn{2}{c|}{\textbf{Cross-Fitting}}                                                                                         &                                &                                      & \textbf{LOO}                             & \textbf{None}         & \textbf{5-Fold}     \\
\multicolumn{2}{c|}{\textbf{Bandwidth Selection}}                                                                                   &                                &                                      & \textbf{LOO}                             & \textbf{LOO}        & \textbf{5-Fold}     \\ \midrule
\multirow{3}{*}{\begin{tabular}[c]{@{}c@{}}$\boldsymbol{N=50}$\\ $\boldsymbol{T_0=50}$\end{tabular}}     & Median Abs. Err.          & 0.1992                         & 0.2147                               & 0.1832                                   & 0.1751              & 0.2004              \\
                                                                                                         & \textbf{95\% CI Coverage} & \textbf{0.8740}                & \textbf{0.8699}                      & \textbf{0.9330}                          & \textbf{0.8240}     & \textbf{0.9180}     \\
                                                                                                         & Median CI Length          & 0.8896                         & 0.9498                               & 1.0547                                   & 0.7292              & 1.0744              \\ \midrule
\multirow{3}{*}{\begin{tabular}[c]{@{}c@{}}$\boldsymbol{N=50}$\\ $\boldsymbol{T_0=250}$\end{tabular}}    & Median Abs. Err.          & 0.2013                         & 0.1916                               & 0.1832                                   & 0.1666              & 0.2010              \\
                                                                                                         & \textbf{95\% CI Coverage} & \textbf{0.8720}                & \textbf{0.8760}                      & \textbf{0.9330}                          & \textbf{0.8550}     & \textbf{0.9410}     \\
                                                                                                         & Median CI Length          & 0.8854                         & 0.9552                               & 1.0547                                   & 0.7823              & 1.0964              \\ \midrule
\multirow{3}{*}{\begin{tabular}[c]{@{}c@{}}$\boldsymbol{N=50}$\\ $\boldsymbol{T_0=1000}$\end{tabular}}   & Median Abs. Err.          & 0.2017                         & 0.1945                               & 0.1832                                   & 0.1744              & 0.1938              \\
                                                                                                         & \textbf{95\% CI Coverage} & \textbf{0.8750}                & \textbf{0.8899}                      & \textbf{0.9330}                          & \textbf{0.8730}     & \textbf{0.9420}     \\
                                                                                                         & Median CI Length          & 0.8857                         & 0.9487                               & 1.0547                                   & 0.8148              & 1.0959              \\ \midrule
\multirow{3}{*}{\begin{tabular}[c]{@{}c@{}}$\boldsymbol{N=250}$\\ $\boldsymbol{T_0=50}$\end{tabular}}    & Median Abs. Err.          & 0.1716                         & 0.0884                               & 0.0754                                   & 0.0840              & 0.0841              \\
                                                                                                         & \textbf{95\% CI Coverage} & \textbf{0.6430}                & \textbf{0.9160}                      & \textbf{0.9580}                          & \textbf{0.8950}     & \textbf{0.9380}     \\
                                                                                                         & Median CI Length          & 0.4119                         & 0.4264                               & 0.4353                                   & 0.3729              & 0.4289              \\ \midrule
\multirow{3}{*}{\begin{tabular}[c]{@{}c@{}}$\boldsymbol{N=250}$\\ $\boldsymbol{T_0=250}$\end{tabular}}   & Median Abs. Err.          & 0.1706                         & 0.0752                               & 0.0754                                   & 0.0761              & 0.0838              \\
                                                                                                         & \textbf{95\% CI Coverage} & \textbf{0.6420}                & \textbf{0.9400}                      & \textbf{0.9580}                          & \textbf{0.9150}     & \textbf{0.9540}     \\
                                                                                                         & Median CI Length          & 0.4103                         & 0.4299                               & 0.4353                                   & 0.3845              & 0.4485              \\ \midrule
\multirow{3}{*}{\begin{tabular}[c]{@{}c@{}}$\boldsymbol{N=250}$\\ $\boldsymbol{T_0=1000}$\end{tabular}}  & Median Abs. Err.          & 0.1707                         & 0.0745                               & 0.0754                                   & 0.0744              & 0.0746              \\
                                                                                                         & \textbf{95\% CI Coverage} & \textbf{0.6420}                & \textbf{0.9390}                      & \textbf{0.9580}                          & \textbf{0.9260}     & \textbf{0.9540}     \\
                                                                                                         & Median CI Length          & 0.4103                         & 0.4317                               & 0.4353                                   & 0.3969              & 0.4382              \\ \midrule
\multirow{3}{*}{\begin{tabular}[c]{@{}c@{}}$\boldsymbol{N=1000}$\\ $\boldsymbol{T_0=50}$\end{tabular}}   & Median Abs. Err.          & 0.1633                         & 0.0462                               & 0.0378                                   & 0.0423              & 0.0412              \\
                                                                                                         & \textbf{95\% CI Coverage} & \textbf{0.1490}                & \textbf{0.8910}                      & \textbf{0.9660}                          & \textbf{0.8790}     & \textbf{0.9090}     \\
                                                                                                         & Median CI Length          & 0.2060                         & 0.2109                               & 0.2130                                   & 0.1939              & 0.2094              \\ \midrule
\multirow{3}{*}{\begin{tabular}[c]{@{}c@{}}$\boldsymbol{N=1000}$\\ $\boldsymbol{T_0=250}$\end{tabular}}  & Median Abs. Err.          & 0.1635                         & 0.0383                               & 0.0378                                   & 0.0371              & 0.0367              \\
                                                                                                         & \textbf{95\% CI Coverage} & \textbf{0.1360}                & \textbf{0.9540}                      & \textbf{0.9660}                          & \textbf{0.9410}     & \textbf{0.9590}     \\
                                                                                                         & Median CI Length          & 0.2055                         & 0.2132                               & 0.2130                                   & 0.2009              & 0.2160              \\ \midrule
\multirow{3}{*}{\begin{tabular}[c]{@{}c@{}}$\boldsymbol{N=1000}$\\ $\boldsymbol{T_0=1000}$\end{tabular}} & Median Abs. Err.          & 0.1630                         & 0.0382                               & 0.0378                                   & 0.0375              & 0.0369              \\
                                                                                                         & \textbf{95\% CI Coverage} & \textbf{0.1420}                & \textbf{0.9630}                      & \textbf{0.9660}                          & \textbf{0.9520}     & \textbf{0.9620}     \\
                                                                                                         & Median CI Length          & 0.2054                         & 0.2139                               & 0.2130                                   & 0.2020              & 0.2109              \\
\bottomrule
\end{tabular}
\end{small}
\begin{tablenotes}
\footnotesize
\setlength\labelsep{0pt}
\item \emph{Notes}: 1,000 replications were conducted in each setting, where the estimand is ${\mathbb{E}[Y_{i,T}(0)\mid w_{i,T}=1]}$ with a true value around $0.9145$. LOO and 5-fold denote leave-one-out and five-fold procedures, respectively, for cross-fitting or bandwidth cross-validation as indicated by the corresponding row; the bandwidth is chosen from the grid by minimizing post-period squared prediction error. If the selected bandwidth yields an empty kernel neighborhood, the next larger grid value is used. "None" indicates no cross-fitting. Local PCA follows the replication package of \citet{feng2026causal} with $K$ chosen by 5-fold cross-validation.
\end{tablenotes}
\end{threeparttable}
\end{center}
\end{adjustwidth}
\end{table}

\begin{table}[h!]%
\begin{adjustwidth}{-1cm}{-1cm}
\begin{center}
\begin{threeparttable}
\caption{Simulation results for Model 4 (Gaussian kernel)}\label{tab:sim_4}
\begin{small}
\addtolength{\tabcolsep}{-0pt}
\renewcommand{\arraystretch}{0.85}
\begin{tabular}{cc|c|c|ccc}
\toprule
                                                                                                       &                           & \multirow{2}{*}{\textbf{TWFE}} & \textbf{Local PCA}                   & \multicolumn{3}{c}{\textbf{Proposed Method}}                                         \\
                                                                                                       &                           &                                & \textbf{(\citealp{feng2026causal})} & \textbf{W/ true $\boldsymbol{\alpha_i}$} & \multicolumn{2}{c}{\textbf{Pseudo-Distance}} \\ \cmidrule(lr){3-3} \cmidrule(lr){4-4} \cmidrule(lr){5-5} \cmidrule(lr){6-7}
\multicolumn{2}{c|}{\textbf{Cross-Fitting}}                                                                                         &                                &                                      & \textbf{LOO}                             & \textbf{None}         & \textbf{5-Fold}     \\
\multicolumn{2}{c|}{\textbf{Bandwidth Selection}}                                                                                   &                                &                                      & \textbf{LOO}                             & \textbf{LOO}        & \textbf{5-Fold}     \\ \midrule
\multirow{3}{*}{\begin{tabular}[c]{@{}c@{}}$\boldsymbol{N=50}$\\ $\boldsymbol{T_0=50}$\end{tabular}}     & Median Abs. Err.          & 0.2104                         & 0.2405                               & 0.1832                                   & 0.1812              & 0.2410              \\
                                                                                                         & \textbf{95\% CI Coverage} & \textbf{0.8680}                & \textbf{0.8679}                      & \textbf{0.9330}                          & \textbf{0.8330}     & \textbf{0.9310}     \\
                                                                                                         & Median CI Length          & 0.9113                         & 1.0565                               & 1.0547                                   & 0.8092              & 1.2549              \\ \midrule
\multirow{3}{*}{\begin{tabular}[c]{@{}c@{}}$\boldsymbol{N=50}$\\ $\boldsymbol{T_0=250}$\end{tabular}}    & Median Abs. Err.          & 0.2052                         & 0.2865                               & 0.1832                                   & 0.1822              & 0.2427              \\
                                                                                                         & \textbf{95\% CI Coverage} & \textbf{0.8570}                & \textbf{0.8529}                      & \textbf{0.9330}                          & \textbf{0.8420}     & \textbf{0.9440}     \\
                                                                                                         & Median CI Length          & 0.8875                         & 1.0731                               & 1.0547                                   & 0.8131              & 1.2789              \\ \midrule
\multirow{3}{*}{\begin{tabular}[c]{@{}c@{}}$\boldsymbol{N=50}$\\ $\boldsymbol{T_0=1000}$\end{tabular}}   & Median Abs. Err.          & 0.2063                         & 0.2567                               & 0.1832                                   & 0.1738              & 0.2388              \\
                                                                                                         & \textbf{95\% CI Coverage} & \textbf{0.8550}                & \textbf{0.8627}                      & \textbf{0.9330}                          & \textbf{0.8470}     & \textbf{0.9420}     \\
                                                                                                         & Median CI Length          & 0.8846                         & 1.0406                               & 1.0547                                   & 0.8094              & 1.2360              \\ \midrule
\multirow{3}{*}{\begin{tabular}[c]{@{}c@{}}$\boldsymbol{N=250}$\\ $\boldsymbol{T_0=50}$\end{tabular}}    & Median Abs. Err.          & 0.1803                         & 0.0829                               & 0.0754                                   & 0.0753              & 0.0841              \\
                                                                                                         & \textbf{95\% CI Coverage} & \textbf{0.6040}                & \textbf{0.9350}                      & \textbf{0.9580}                          & \textbf{0.9310}     & \textbf{0.9660}     \\
                                                                                                         & Median CI Length          & 0.4186                         & 0.4404                               & 0.4353                                   & 0.4116              & 0.4813              \\ \midrule
\multirow{3}{*}{\begin{tabular}[c]{@{}c@{}}$\boldsymbol{N=250}$\\ $\boldsymbol{T_0=250}$\end{tabular}}   & Median Abs. Err.          & 0.1810                         & 0.0820                               & 0.0754                                   & 0.0762              & 0.0814              \\
                                                                                                         & \textbf{95\% CI Coverage} & \textbf{0.6010}                & \textbf{0.9440}                      & \textbf{0.9580}                          & \textbf{0.9470}     & \textbf{0.9660}     \\
                                                                                                         & Median CI Length          & 0.4118                         & 0.4405                               & 0.4353                                   & 0.4140              & 0.4689              \\ \midrule
\multirow{3}{*}{\begin{tabular}[c]{@{}c@{}}$\boldsymbol{N=250}$\\ $\boldsymbol{T_0=1000}$\end{tabular}}  & Median Abs. Err.          & 0.1824                         & 0.0773                               & 0.0754                                   & 0.0751              & 0.0761              \\
                                                                                                         & \textbf{95\% CI Coverage} & \textbf{0.5930}                & \textbf{0.9430}                      & \textbf{0.9580}                          & \textbf{0.9470}     & \textbf{0.9610}     \\
                                                                                                         & Median CI Length          & 0.4096                         & 0.4445                               & 0.4353                                   & 0.4140              & 0.4630              \\ \midrule
\multirow{3}{*}{\begin{tabular}[c]{@{}c@{}}$\boldsymbol{N=1000}$\\ $\boldsymbol{T_0=50}$\end{tabular}}   & Median Abs. Err.          & 0.1704                         & 0.0384                               & 0.0378                                   & 0.0382              & 0.0380              \\
                                                                                                         & \textbf{95\% CI Coverage} & \textbf{0.1230}                & \textbf{0.9610}                      & \textbf{0.9660}                          & \textbf{0.9560}     & \textbf{0.9660}     \\
                                                                                                         & Median CI Length          & 0.2096                         & 0.2150                               & 0.2130                                   & 0.2105              & 0.2203              \\ \midrule
\multirow{3}{*}{\begin{tabular}[c]{@{}c@{}}$\boldsymbol{N=1000}$\\ $\boldsymbol{T_0=250}$\end{tabular}}  & Median Abs. Err.          & 0.1745                         & 0.0393                               & 0.0378                                   & 0.0378              & 0.0385              \\
                                                                                                         & \textbf{95\% CI Coverage} & \textbf{0.1060}                & \textbf{0.9580}                      & \textbf{0.9660}                          & \textbf{0.9620}     & \textbf{0.9620}     \\
                                                                                                         & Median CI Length          & 0.2064                         & 0.2160                               & 0.2130                                   & 0.2104              & 0.2178              \\ \midrule
\multirow{3}{*}{\begin{tabular}[c]{@{}c@{}}$\boldsymbol{N=1000}$\\ $\boldsymbol{T_0=1000}$\end{tabular}} & Median Abs. Err.          & 0.1730                         & 0.0394                               & 0.0378                                   & 0.0371              & 0.0367              \\
                                                                                                         & \textbf{95\% CI Coverage} & \textbf{0.0960}                & \textbf{0.9560}                      & \textbf{0.9660}                          & \textbf{0.9620}     & \textbf{0.9640}     \\
                                                                                                         & Median CI Length          & 0.2055                         & 0.2168                               & 0.2130                                   & 0.2101              & 0.2167              \\
\bottomrule
\end{tabular}
\end{small}
\begin{tablenotes}
\footnotesize
\setlength\labelsep{0pt}
\item \emph{Notes}:  Details are provided in the notes to Table 3.
\end{tablenotes}
\end{threeparttable}
\end{center}
\end{adjustwidth}
\end{table}

\begin{table}[h!]%
\begin{adjustwidth}{-1cm}{-1cm}
\begin{center}
\begin{threeparttable}
\caption{Simulation results for Model 5 (exponential kernel)}\label{tab:sim_5}
\begin{small}
\addtolength{\tabcolsep}{-0pt}
\renewcommand{\arraystretch}{0.85}
\begin{tabular}{cc|c|c|ccc}
\toprule
                                                                                                       &                           & \multirow{2}{*}{\textbf{TWFE}} & \textbf{Local PCA}                   & \multicolumn{3}{c}{\textbf{Proposed Method}}                                         \\
                                                                                                       &                           &                                & \textbf{(\citealp{feng2026causal})} & \textbf{W/ true $\boldsymbol{\alpha_i}$} & \multicolumn{2}{c}{\textbf{Pseudo-Distance}} \\ \cmidrule(lr){3-3} \cmidrule(lr){4-4} \cmidrule(lr){5-5} \cmidrule(lr){6-7}
\multicolumn{2}{c|}{\textbf{Cross-Fitting}}                                                                                         &                                &                                      & \textbf{LOO}                             & \textbf{None}         & \textbf{5-Fold}     \\
\multicolumn{2}{c|}{\textbf{Bandwidth Selection}}                                                                                   &                                &                                      & \textbf{LOO}                             & \textbf{LOO}        & \textbf{5-Fold}     \\ \midrule
\multirow{3}{*}{\begin{tabular}[c]{@{}c@{}}$\boldsymbol{N=50}$\\ $\boldsymbol{T_0=50}$\end{tabular}}     & Median Abs. Err.          & 0.2005                         & 0.1988                               & 0.1832                                   & 0.1868              & 0.1979              \\
                                                                                                         & \textbf{95\% CI Coverage} & \textbf{0.8620}                & \textbf{0.8730}                      & \textbf{0.9330}                          & \textbf{0.8040}     & \textbf{0.9240}     \\
                                                                                                         & Median CI Length          & 0.8850                         & 0.9312                               & 1.0547                                   & 0.7206              & 1.0791              \\ \midrule
\multirow{3}{*}{\begin{tabular}[c]{@{}c@{}}$\boldsymbol{N=50}$\\ $\boldsymbol{T_0=250}$\end{tabular}}    & Median Abs. Err.          & 0.2003                         & 0.2001                               & 0.1832                                   & 0.1883              & 0.2128              \\
                                                                                                         & \textbf{95\% CI Coverage} & \textbf{0.8620}                & \textbf{0.8769}                      & \textbf{0.9330}                          & \textbf{0.8210}     & \textbf{0.9380}     \\
                                                                                                         & Median CI Length          & 0.8818                         & 0.9294                               & 1.0547                                   & 0.7457              & 1.2056              \\ \midrule
\multirow{3}{*}{\begin{tabular}[c]{@{}c@{}}$\boldsymbol{N=50}$\\ $\boldsymbol{T_0=1000}$\end{tabular}}   & Median Abs. Err.          & 0.1997                         & 0.1881                               & 0.1832                                   & 0.1727              & 0.2100              \\
                                                                                                         & \textbf{95\% CI Coverage} & \textbf{0.8610}                & \textbf{0.8879}                      & \textbf{0.9330}                          & \textbf{0.8640}     & \textbf{0.9390}     \\
                                                                                                         & Median CI Length          & 0.8816                         & 0.9205                               & 1.0547                                   & 0.8073              & 1.1289              \\ \midrule
\multirow{3}{*}{\begin{tabular}[c]{@{}c@{}}$\boldsymbol{N=250}$\\ $\boldsymbol{T_0=50}$\end{tabular}}    & Median Abs. Err.          & 0.1788                         & 0.0942                               & 0.0754                                   & 0.0835              & 0.0869              \\
                                                                                                         & \textbf{95\% CI Coverage} & \textbf{0.6080}                & \textbf{0.8800}                      & \textbf{0.9580}                          & \textbf{0.8490}     & \textbf{0.9270}     \\
                                                                                                         & Median CI Length          & 0.4100                         & 0.4189                               & 0.4353                                   & 0.3565              & 0.4505              \\ \midrule
\multirow{3}{*}{\begin{tabular}[c]{@{}c@{}}$\boldsymbol{N=250}$\\ $\boldsymbol{T_0=250}$\end{tabular}}   & Median Abs. Err.          & 0.1781                         & 0.0804                               & 0.0754                                   & 0.0797              & 0.0857              \\
                                                                                                         & \textbf{95\% CI Coverage} & \textbf{0.6140}                & \textbf{0.9340}                      & \textbf{0.9580}                          & \textbf{0.9040}     & \textbf{0.9550}     \\
                                                                                                         & Median CI Length          & 0.4083                         & 0.4198                               & 0.4353                                   & 0.3729              & 0.4812              \\ \midrule
\multirow{3}{*}{\begin{tabular}[c]{@{}c@{}}$\boldsymbol{N=250}$\\ $\boldsymbol{T_0=1000}$\end{tabular}}  & Median Abs. Err.          & 0.1783                         & 0.0781                               & 0.0754                                   & 0.0733              & 0.0751              \\
                                                                                                         & \textbf{95\% CI Coverage} & \textbf{0.6040}                & \textbf{0.9230}                      & \textbf{0.9580}                          & \textbf{0.9360}     & \textbf{0.9630}     \\
                                                                                                         & Median CI Length          & 0.4079                         & 0.4132                               & 0.4353                                   & 0.4033              & 0.4400              \\ \midrule
\multirow{3}{*}{\begin{tabular}[c]{@{}c@{}}$\boldsymbol{N=1000}$\\ $\boldsymbol{T_0=50}$\end{tabular}}   & Median Abs. Err.          & 0.1695                         & 0.0589                               & 0.0378                                   & 0.0425              & 0.0413              \\
                                                                                                         & \textbf{95\% CI Coverage} & \textbf{0.1180}                & \textbf{0.7950}                      & \textbf{0.9660}                          & \textbf{0.8670}     & \textbf{0.9230}     \\
                                                                                                         & Median CI Length          & 0.2053                         & 0.2069                               & 0.2130                                   & 0.1858              & 0.2178              \\ \midrule
\multirow{3}{*}{\begin{tabular}[c]{@{}c@{}}$\boldsymbol{N=1000}$\\ $\boldsymbol{T_0=250}$\end{tabular}}  & Median Abs. Err.          & 0.1712                         & 0.0405                               & 0.0378                                   & 0.0379              & 0.0395              \\
                                                                                                         & \textbf{95\% CI Coverage} & \textbf{0.1100}                & \textbf{0.9350}                      & \textbf{0.9660}                          & \textbf{0.9420}     & \textbf{0.9650}     \\
                                                                                                         & Median CI Length          & 0.2047                         & 0.2090                               & 0.2130                                   & 0.1963              & 0.2238              \\ \midrule
\multirow{3}{*}{\begin{tabular}[c]{@{}c@{}}$\boldsymbol{N=1000}$\\ $\boldsymbol{T_0=1000}$\end{tabular}} & Median Abs. Err.          & 0.1710                         & 0.0397                               & 0.0378                                   & 0.0366              & 0.0367              \\
                                                                                                         & \textbf{95\% CI Coverage} & \textbf{0.1090}                & \textbf{0.9460}                      & \textbf{0.9660}                          & \textbf{0.9540}     & \textbf{0.9620}     \\
                                                                                                         & Median CI Length          & 0.2045                         & 0.2073                               & 0.2130                                   & 0.2063              & 0.2114              \\
\bottomrule
\end{tabular}
\end{small}
\begin{tablenotes}
\footnotesize
\setlength\labelsep{0pt}
\item \emph{Notes}:  Details are provided in the notes to Table 3.
\end{tablenotes}
\end{threeparttable}
\end{center}
\end{adjustwidth}
\end{table}

In general, the proposed method provides good coverage across different sample sizes. Specifically, the configuration using five-fold cross-fitting (and five-fold cross-validation) generally delivers confidence interval coverage closer to the nominal level than the configuration without cross-fitting (and with leave-one-out cross-validation). %

In Models 3, 4, and 5, the proposed method with five-fold cross-fitting achieves lower median absolute error than local PCA in 19 of the 27 designs, while delivering substantial improvements in confidence interval coverage in several settings. The coverage advantage is particularly pronounced when $N$ is small or the pre-treatment history is short relative to $N$. For $N=50$, local PCA has coverage below $90\%$ in every design, whereas the proposed method attains coverage between $91.8\%$ and $94.4\%$, accompanied by longer confidence intervals. In Model 5 with $N=1000$ and $T_0=50$, the proposed method reduces median absolute error by approximately $30\%$ relative to local PCA, from $0.0589$ to $0.0413$, while increasing coverage from $79.5\%$ to $92.3\%$. Also note that the proposed method delivers these gains while estimating propensity scores nonparametrically, whereas local PCA uses local logistic regression with the same link function as the treatment assignment mechanism in the DGP. Across all reported designs, the proposed method's coverage is within three percentage points of the nominal $95\%$ level in all but two cases, both of which occur in Model~3.

In summary, these simulations demonstrate that the proposed method performs comparably to TWFE when the additive TWFE specification is correct and continues to provide reliable inference in nonlinear settings where TWFE exhibits substantial undercoverage. Across the nonlinear designs, the proposed method often improves on local PCA in median absolute error and delivers substantial coverage gains in several settings, particularly when the cross-sectional sample size is small or the pre-treatment history is short relative to it. These results demonstrate that the proposed method combines simplicity with frequent gains in estimation accuracy and substantial improvements in coverage relative to the competing methods.

\FloatBarrier
\section{Empirical Illustration} %
\label{sec:5}

To evaluate the proposed method in settings that resemble real-world applications, we conduct an empirically calibrated simulation. We estimate the factors and loadings, fit the data nonparametrically, and then generate the panel by resampling from the estimated objects. This approach ensures that the generated panels reflect the underlying structure, heterogeneity, and noise levels of genuine data, while still allowing us full control over the sample size and potential smoothness.

\subsection{Data}

We calibrate the simulation to the Walmart weekly-sales panel used by \cite{abadie2021synthetic}, where the data are originally from \cite{prakash2023walmart}. The balanced panel records the weekly sales of $\tilde{N} = 45$ Walmart stores over $\tilde{T} = 143$ consecutive weeks, running from February 5, 2010, through October 26, 2012. The outcome variable is store-level weekly sales. After removing additive store and week fixed effects, the residualized panel retains a low-dimensional index and an empirical noise distribution; these are the features the calibrated DGP is designed to inherit. Note that since there is no treatment implemented in the actual data, the treatment mechanism for our simulation will be manually set up later. The data are standardized before the empirically calibrated panels are constructed, and we further denote this standardized outcome as $\tilde{Y}_{i,t}$.

\subsection{Empirically Calibrated Data Generating Process}

The DGP is calibrated nonparametrically in four steps: (i) estimate the additive fixed effects and the low-rank components; (ii) resample new individual-level factors and select the desired-length time-level loadings; (iii) fit the new factors back to outcomes nonparametrically and add resampled error terms; and (iv) assign treatment through a propensity score that depends on the latent features. A detailed description of each step is provided below.

 \smallskip

\noindent\textbf{Step (i) Extract factors:} In a real-world scenario, although traditional two-way fixed effects often account for a significant portion of the variation, some variation is still explained by the nonlinear combination of remaining factors. To reflect this, we first remove additive two-way fixed effects $\bar{\alpha}_i$ and $\bar{\lambda}_t$ from the standardized outcome $\tilde{Y}_{i,t}$. Then with the remaining outcome $\bar{Y}_{i,t} := \tilde{Y}_{i,t} - \bar{\alpha}_i - \bar{\lambda}_t$, following \cite{bai2003inferential}, we estimate $R$ factors and loadings, $\tilde{\alpha} \in \mathbb{R}^{\tilde{N} \times R}$ and $\tilde{\lambda} \in \mathbb{R}^{\tilde{T} \times R}$, and these extracted latent factors are the estimated low-dimensional indices used to calibrate the subsequent nonparametric outcome model. In this simulation, we set our number of factors to be $R=1$. Note that the additive fixed effects $\bar{\alpha}_i$ and $\bar{\lambda}_t$ are used only to isolate the remaining potentially nonlinear structure. As the original sales data are heavily driven by store fixed effects and seasonal effects, and the purpose of this empirically calibrated simulation is to study how the proposed method can further handle the nonlinear structure, we do not add those additive fixed effects back to the simulated panel.

 \smallskip

\noindent\textbf{Step (ii) Sample new factors:}
For a target panel of dimensions $N \times T$, new individual factors $\alpha = \{\alpha_i\}_{i \in \{1,...,N\}}$ are drawn by a kernel-density bootstrap\footnote{Here, we use the Epanechnikov kernel via R package \texttt{kernelboot::rmvk}, where each draw equals a uniformly resampled element of $\tilde\alpha$ plus Epanechnikov noise scaled by the package's default rule-of-thumb bandwidth.} from the empirical distribution of $\tilde\alpha$. This produces a smooth analogue of the empirical distribution that supports arbitrary cross-sectional sample sizes $N$, while still preserving the empirical co-movement between factors when $R > 1$ is used.

As it is more difficult to extrapolate the time dimension, here we only consider the $T \leq \tilde{T} = 143$ setting. Rather than resampling the time dimension, we simply select $T$ equally spaced calendar weeks from the original panel: simulated period $t$ is mapped to the original week
$\tau(t) := \mathrm{round}\!\left( t \cdot \tilde{T} / T \right)$ for
$t = 1, \dots, T$. The simulated time dimension therefore inherits the empirical time profile directly, including seasonality and trends, without a separately drawn time factor. Note that $\tau(T) = \tilde{T} = 143$, so the last simulated period corresponds to the same calendar week in every $T$ setting. In sum, the new individual factors are $\alpha \in \mathbb{R}^{N \times R}$.

 \smallskip

\noindent\textbf{Step (iii) Fitting the outcome nonparametrically:} With the empirically sampled factors, to generate the new outcome $Y_{i,t}(0)$ for $i \in \{1,...,N\}$ and $t \in \{1,...,T\}$, we then fit the outcome $\bar{Y}_{i,t}$ nonparametrically:
\begin{align*}
Y_{i,t}(0)
= \underbrace{\frac{\sum\limits^{\tilde{N}}_{j=1} K_h\left( \tilde{\alpha}_j - \alpha_i \right) \cdot \bar{Y}_{j,\tau(t)}}{\sum\limits^{\tilde{N}}_{j=1} K_h\left( \tilde{\alpha}_j - \alpha_i \right)}}_{=: \; \hat{m}\left( \alpha_i, \tau(t) \right)}
\;+\; \underbrace{u_{i,t}}_{\text{\ul{sampled from residuals}}}
\end{align*}
In the kernel regression part, the product Epanechnikov kernel is used:
\begin{align*}
K_h(\mathbf{v}) = \prod_{r=1}^{R} \frac{3}{4} \left(1 - (v_r/h_r)^2\right) \cdot \mathbf{1}\left\{(v_r/h_r)^2 \le 1\right\},
\end{align*}
The bandwidth in each dimension follows a rule of thumb $h_{r} = c \cdot s_{r} \cdot \tilde{N}^{-1/(R+4)}$ for $r = 1,\dots,R$,
where $s_r$ is the sample standard deviation of the $r$-th loading, $\tilde{N} = 45$ is
the number of \emph{original} stores, and $c$ is a tuning constant that indexes the
calibration regimes below. The shocks $u_{i,t}$ are sampled uniformly and independently (with replacement) from the pooled fitted residuals $\bar{Y}_{\ell,k} - \hat{m}(\tilde{\alpha}_{\ell},k)$ across all original stores $\ell$ and weeks $k$.

 \smallskip

\noindent\textbf{Step (iv) Treatment assignment mechanism and estimand:} Since no treatment happened in the actual data, here, we generate a simple mechanism. Treatment is confined to a single period, the last one, with the first $T - 1 =: T_0$ periods serving as the pre-treatment time periods. Assignment is governed by a propensity score that is logistic in the latent individual factors
$
p_i = 1/\left[1 + \exp(- a \cdot \sum_{r=1}^R \alpha_{i,r})\right]
$
where we set $a = 2.5$ so that the propensity scores lie approximately in $[0.3, 0.7]$. In this simulation study, no treatment effect is injected into the outcome. The target estimand is the counterfactual control mean for the treated units in the period $T$,
$\mathbb{E}\left[Y_{i,T}(0) \; \vert \; w_{i,T} = 1 \right]$.

 \smallskip

\noindent\textbf{Two calibration regimes of smoothness:} The constant $c$ in the kernel regression controls how smooth the nonparametric structure is. A small $c$ yields a small bandwidth, inducing fast singular-value decay. Here, we generate two settings for every sample size,
\begin{align*}
\text{Setting 1 (fast decay): } c = 2.5
\qquad
\text{Setting 2 (slow decay): } c = 6
\end{align*}
Figure \ref{fig:sv} reports the singular value spectra of the two calibrated regimes from one replication alongside the residualized original data.
\begin{figure}[ht]
\begin{center}
\includegraphics[width=0.85\textwidth]{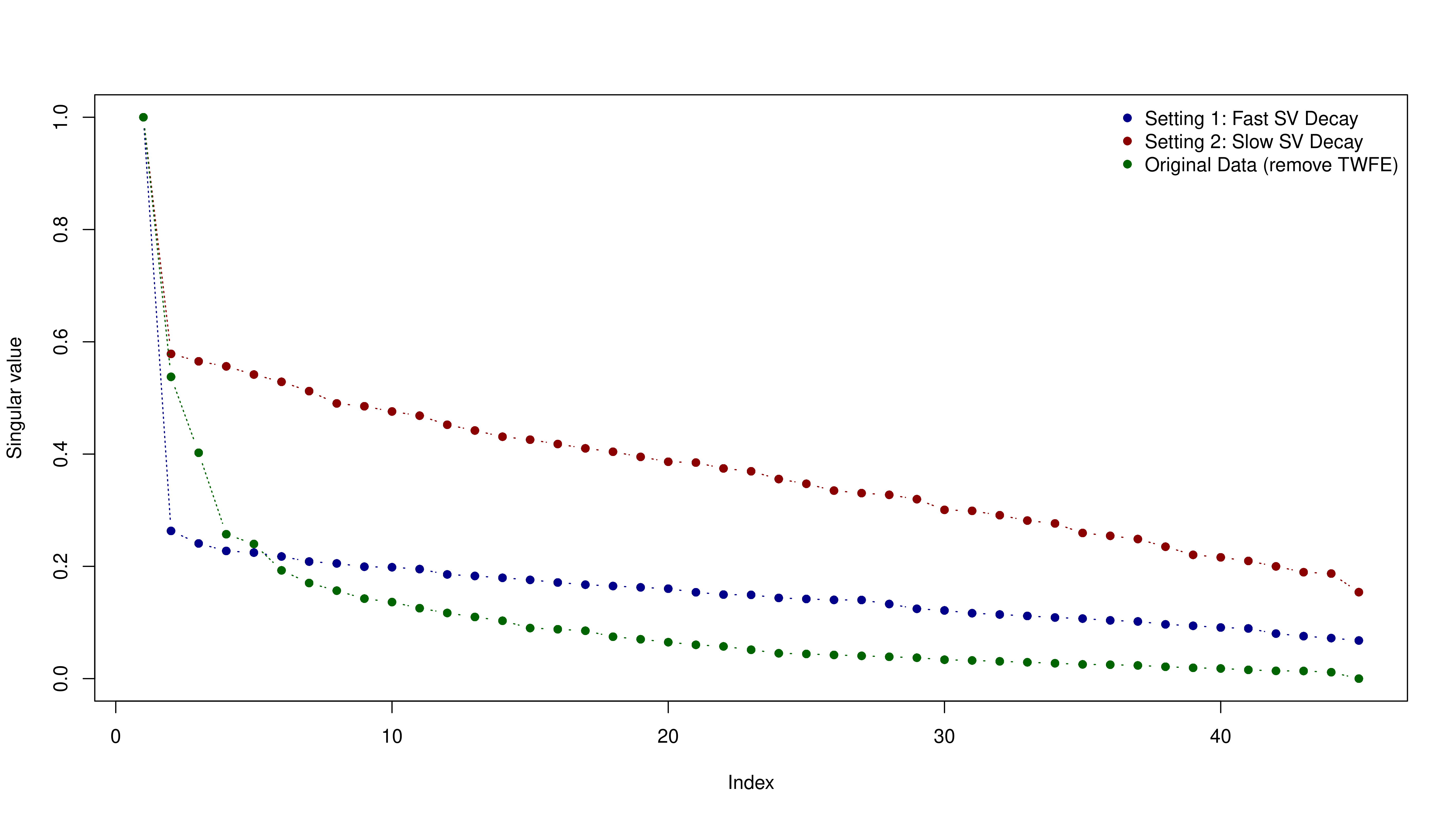}
\caption{Singular Value Spectra}\label{fig:sv}
\end{center}
\begin{minipage}{\textwidth}
\footnotesize
\emph{Notes}: Each line is normalized by its own largest singular value. Settings 1 and 2 are single calibrated panels of dimension $\tilde{N} = 45$ by $\tilde{T} = 143$ from one replication. The green line represents the original Walmart panel after standardization and removal of additive store and week effects.
\end{minipage}
\end{figure}
In the figure, all spectra are normalized so that the largest singular value equals one. Setting 1 falls sharply after the first component and then declines slowly; the residualized Walmart panel falls more gradually over the first few components but continues to decline, so that beyond roughly the sixth component it sits below Setting 1. Setting 2 remains well above both across the entire spectrum. The mechanism is that a smaller calibration bandwidth produces a closer fit and hence smaller resampled residuals. Therefore, in Setting 1 the fitted component accounts for a larger share of the variation, whereas in Setting 2 the resampled noise dominates. We use these two settings to compare the estimators across panels whose singular values decay at visibly different rates.

\FloatBarrier
\subsection{Results}

We vary the size of the panel dimensions over $N \in \{50, 100, 200\}$ and $T \in \{50, 100, 143\}$. Tables \ref{tab:empfast} and \ref{tab:empslow} report the results for the fast- and slow-decay regimes. In this empirical illustration, the bandwidth input for the proposed method is selected between $0.01$ and $9$.

\begin{table}[h!]%
\begin{adjustwidth}{-1cm}{-1cm}
\begin{center}
\begin{threeparttable}
\caption{Empirically calibrated simulation results (fast singular-value decay)}\label{tab:empfast}
\begin{small}
\addtolength{\tabcolsep}{-0pt}
\renewcommand{\arraystretch}{0.80}
\begin{tabular}{cc|c|c|ccc}
\toprule
                                                                                                       &                           & \multirow{2}{*}{\textbf{TWFE}} & \textbf{Local PCA}                   & \multicolumn{3}{c}{\textbf{Proposed Method}}                                         \\
                                                                                                       &                           &                                & \textbf{(\citealp{feng2026causal})} & \textbf{W/ true $\boldsymbol{\alpha_i}$} & \multicolumn{2}{c}{\textbf{Pseudo-Distance}} \\ \cmidrule(lr){3-3} \cmidrule(lr){4-4} \cmidrule(lr){5-5} \cmidrule(lr){6-7}
\multicolumn{2}{c|}{\textbf{Cross-Fitting}}                                                                                         &                                &                                      & \textbf{LOO}                             & \textbf{None}         & \textbf{5-Fold}     \\
\multicolumn{2}{c|}{\textbf{Bandwidth Selection}}                                                                                   &                                &                                      & \textbf{LOO}                             & \textbf{LOO}        & \textbf{5-Fold}     \\ \midrule
\multirow{3}{*}{\begin{tabular}[c]{@{}c@{}}$\boldsymbol{N=50}$\\ $\boldsymbol{T=50}$\end{tabular}}   & Median Abs. Err.          & 0.0224                         & 0.0210                               & 0.0201                                   & 0.0191              & 0.0200              \\
                                                                                                     & \textbf{95\% CI Coverage} & \textbf{0.9010}                & \textbf{0.8820}                      & \textbf{0.9450}                          & \textbf{0.9030}     & \textbf{0.9480}     \\
                                                                                                     & Median CI Length          & 0.1023                         & 0.1055                               & 0.1124                                   & 0.0944              & 0.1146              \\ \midrule
\multirow{3}{*}{\begin{tabular}[c]{@{}c@{}}$\boldsymbol{N=50}$\\ $\boldsymbol{T=100}$\end{tabular}}  & Median Abs. Err.          & 0.0247                         & 0.0220                               & 0.0206                                   & 0.0202              & 0.0223              \\
                                                                                                     & \textbf{95\% CI Coverage} & \textbf{0.8490}                & \textbf{0.8770}                      & \textbf{0.9210}                          & \textbf{0.8720}     & \textbf{0.9230}     \\
                                                                                                     & Median CI Length          & 0.1046                         & 0.1041                               & 0.1117                                   & 0.0950              & 0.1143              \\ \midrule
\multirow{3}{*}{\begin{tabular}[c]{@{}c@{}}$\boldsymbol{N=50}$\\ $\boldsymbol{T=143}$\end{tabular}}  & Median Abs. Err.          & 0.0218                         & 0.0226                               & 0.0198                                   & 0.0196              & 0.0214              \\
                                                                                                     & \textbf{95\% CI Coverage} & \textbf{0.8950}                & \textbf{0.8900}                      & \textbf{0.9460}                          & \textbf{0.8880}     & \textbf{0.9300}     \\
                                                                                                     & Median CI Length          & 0.1016                         & 0.1036                               & 0.1126                                   & 0.0943              & 0.1157              \\ \midrule
\multirow{3}{*}{\begin{tabular}[c]{@{}c@{}}$\boldsymbol{N=100}$\\ $\boldsymbol{T=50}$\end{tabular}}  & Median Abs. Err.          & 0.0183                         & 0.0157                               & 0.0146                                   & 0.0148              & 0.0151              \\
                                                                                                     & \textbf{95\% CI Coverage} & \textbf{0.8470}                & \textbf{0.9080}                      & \textbf{0.9250}                          & \textbf{0.9010}     & \textbf{0.9240}     \\
                                                                                                     & Median CI Length          & 0.0756                         & 0.0764                               & 0.0789                                   & 0.0704              & 0.0794              \\ \midrule
\multirow{3}{*}{\begin{tabular}[c]{@{}c@{}}$\boldsymbol{N=100}$\\ $\boldsymbol{T=100}$\end{tabular}} & Median Abs. Err.          & 0.0189                         & 0.0174                               & 0.0150                                   & 0.0151              & 0.0152              \\
                                                                                                     & \textbf{95\% CI Coverage} & \textbf{0.8380}                & \textbf{0.9110}                      & \textbf{0.9330}                          & \textbf{0.9070}     & \textbf{0.9370}     \\
                                                                                                     & Median CI Length          & 0.0755                         & 0.0777                               & 0.0783                                   & 0.0703              & 0.0806              \\ \midrule
\multirow{3}{*}{\begin{tabular}[c]{@{}c@{}}$\boldsymbol{N=100}$\\ $\boldsymbol{T=143}$\end{tabular}} & Median Abs. Err.          & 0.0187                         & 0.0154                               & 0.0143                                   & 0.0139              & 0.0149              \\
                                                                                                     & \textbf{95\% CI Coverage} & \textbf{0.8550}                & \textbf{0.9260}                      & \textbf{0.9400}                          & \textbf{0.9210}     & \textbf{0.9420}     \\
                                                                                                     & Median CI Length          & 0.0751                         & 0.0777                               & 0.0792                                   & 0.0716              & 0.0802              \\ \midrule
\multirow{3}{*}{\begin{tabular}[c]{@{}c@{}}$\boldsymbol{N=200}$\\ $\boldsymbol{T=50}$\end{tabular}}  & Median Abs. Err.          & 0.0154                         & 0.0097                               & 0.0092                                   & 0.0091              & 0.0096              \\
                                                                                                     & \textbf{95\% CI Coverage} & \textbf{0.7900}                & \textbf{0.9340}                      & \textbf{0.9570}                          & \textbf{0.9360}     & \textbf{0.9540}     \\
                                                                                                     & Median CI Length          & 0.0548                         & 0.0552                               & 0.0563                                   & 0.0518              & 0.0558              \\ \midrule
\multirow{3}{*}{\begin{tabular}[c]{@{}c@{}}$\boldsymbol{N=200}$\\ $\boldsymbol{T=100}$\end{tabular}} & Median Abs. Err.          & 0.0170                         & 0.0099                               & 0.0094                                   & 0.0095              & 0.0098              \\
                                                                                                     & \textbf{95\% CI Coverage} & \textbf{0.7700}                & \textbf{0.9370}                      & \textbf{0.9480}                          & \textbf{0.9290}     & \textbf{0.9440}     \\
                                                                                                     & Median CI Length          & 0.0542                         & 0.0552                               & 0.0561                                   & 0.0520              & 0.0557              \\ \midrule
\multirow{3}{*}{\begin{tabular}[c]{@{}c@{}}$\boldsymbol{N=200}$\\ $\boldsymbol{T=143}$\end{tabular}} & Median Abs. Err.          & 0.0173                         & 0.0098                               & 0.0100                                   & 0.0095              & 0.0098              \\
                                                                                                     & \textbf{95\% CI Coverage} & \textbf{0.7670}                & \textbf{0.9340}                      & \textbf{0.9510}                          & \textbf{0.9280}     & \textbf{0.9360}     \\
                                                                                                     & Median CI Length          & 0.0536                         & 0.0553                               & 0.0560                                   & 0.0516              & 0.0558              \\
\bottomrule
\end{tabular}
\end{small}
\begin{tablenotes}
\footnotesize
\setlength\labelsep{0pt}
\item \emph{Notes}: 1,000 replications were conducted in each setting, where the estimand is ${\mathbb{E}[Y_{i,T}(0)\mid w_{i,T}=1]}$, which is numerically approximated once per setting using 200,000 Monte Carlo draws of the noiseless calibrated regression function (values $\approx 0.008$). LOO and 5-fold denote leave-one-out and five-fold procedures, respectively, for cross-fitting or bandwidth cross-validation as indicated by the corresponding row; the bandwidth is chosen from the grid by minimizing post-period squared prediction error. If the selected bandwidth yields an empty kernel neighborhood, the next larger grid value is used.  "None" indicates no cross-fitting.  Local PCA follows the replication package of \citet{feng2026causal} with $K$ chosen by 5-fold cross-validation.
\end{tablenotes}
\end{threeparttable}
\end{center}
\end{adjustwidth}
\end{table}

\begin{table}[h!]%
\begin{adjustwidth}{-1cm}{-1cm}
\begin{center}
\begin{threeparttable}
\caption{Empirically calibrated simulation results (slow singular-value decay)}\label{tab:empslow}
\begin{small}
\addtolength{\tabcolsep}{-0pt}
\renewcommand{\arraystretch}{0.80}
\begin{tabular}{cc|c|c|ccc}
\toprule
                                                                                                       &                           & \multirow{2}{*}{\textbf{TWFE}} & \textbf{Local PCA}                   & \multicolumn{3}{c}{\textbf{Proposed Method}}                                         \\
                                                                                                       &                           &                                & \textbf{(\citealp{feng2026causal})} & \textbf{W/ true $\boldsymbol{\alpha_i}$} & \multicolumn{2}{c}{\textbf{Pseudo-Distance}} \\ \cmidrule(lr){3-3} \cmidrule(lr){4-4} \cmidrule(lr){5-5} \cmidrule(lr){6-7}
\multicolumn{2}{c|}{\textbf{Cross-Fitting}}                                                                                         &                                &                                      & \textbf{LOO}                             & \textbf{None}         & \textbf{5-Fold}     \\
\multicolumn{2}{c|}{\textbf{Bandwidth Selection}}                                                                                   &                                &                                      & \textbf{LOO}                             & \textbf{LOO}        & \textbf{5-Fold}     \\ \midrule
\multirow{3}{*}{\begin{tabular}[c]{@{}c@{}}$\boldsymbol{N=50}$\\ $\boldsymbol{T=50}$\end{tabular}}   & Median Abs. Err.          & 0.0211                         & 0.0248                               & 0.0215                                   & 0.0204              & 0.0204              \\
                                                                                                     & \textbf{95\% CI Coverage} & \textbf{0.9360}                & \textbf{0.9000}                      & \textbf{0.9420}                          & \textbf{0.9200}     & \textbf{0.9440}     \\
                                                                                                     & Median CI Length          & 0.1152                         & 0.1229                               & 0.1242                                   & 0.1092              & 0.1255              \\ \midrule
\multirow{3}{*}{\begin{tabular}[c]{@{}c@{}}$\boldsymbol{N=50}$\\ $\boldsymbol{T=100}$\end{tabular}}  & Median Abs. Err.          & 0.0248                         & 0.0277                               & 0.0248                                   & 0.0243              & 0.0251              \\
                                                                                                     & \textbf{95\% CI Coverage} & \textbf{0.8890}                & \textbf{0.8700}                      & \textbf{0.9200}                          & \textbf{0.8820}     & \textbf{0.9220}     \\
                                                                                                     & Median CI Length          & 0.1131                         & 0.1220                               & 0.1252                                   & 0.1100              & 0.1254              \\ \midrule
\multirow{3}{*}{\begin{tabular}[c]{@{}c@{}}$\boldsymbol{N=50}$\\ $\boldsymbol{T=143}$\end{tabular}}  & Median Abs. Err.          & 0.0213                         & 0.0249                               & 0.0217                                   & 0.0215              & 0.0224              \\
                                                                                                     & \textbf{95\% CI Coverage} & \textbf{0.9330}                & \textbf{0.8800}                      & \textbf{0.9480}                          & \textbf{0.9130}     & \textbf{0.9490}     \\
                                                                                                     & Median CI Length          & 0.1139                         & 0.1200                               & 0.1236                                   & 0.1091              & 0.1255              \\ \midrule
\multirow{3}{*}{\begin{tabular}[c]{@{}c@{}}$\boldsymbol{N=100}$\\ $\boldsymbol{T=50}$\end{tabular}}  & Median Abs. Err.          & 0.0159                         & 0.0173                               & 0.0160                                   & 0.0155              & 0.0160              \\
                                                                                                     & \textbf{95\% CI Coverage} & \textbf{0.9240}                & \textbf{0.9090}                      & \textbf{0.9340}                          & \textbf{0.9100}     & \textbf{0.9340}     \\
                                                                                                     & Median CI Length          & 0.0846                         & 0.0881                               & 0.0873                                   & 0.0807              & 0.0874              \\ \midrule
\multirow{3}{*}{\begin{tabular}[c]{@{}c@{}}$\boldsymbol{N=100}$\\ $\boldsymbol{T=100}$\end{tabular}} & Median Abs. Err.          & 0.0156                         & 0.0175                               & 0.0170                                   & 0.0164              & 0.0160              \\
                                                                                                     & \textbf{95\% CI Coverage} & \textbf{0.9180}                & \textbf{0.9180}                      & \textbf{0.9320}                          & \textbf{0.9190}     & \textbf{0.9340}     \\
                                                                                                     & Median CI Length          & 0.0857                         & 0.0895                               & 0.0892                                   & 0.0832              & 0.0894              \\ \midrule
\multirow{3}{*}{\begin{tabular}[c]{@{}c@{}}$\boldsymbol{N=100}$\\ $\boldsymbol{T=143}$\end{tabular}} & Median Abs. Err.          & 0.0162                         & 0.0178                               & 0.0158                                   & 0.0159              & 0.0166              \\
                                                                                                     & \textbf{95\% CI Coverage} & \textbf{0.9460}                & \textbf{0.9140}                      & \textbf{0.9560}                          & \textbf{0.9390}     & \textbf{0.9530}     \\
                                                                                                     & Median CI Length          & 0.0857                         & 0.0899                               & 0.0897                                   & 0.0831              & 0.0898              \\ \midrule
\multirow{3}{*}{\begin{tabular}[c]{@{}c@{}}$\boldsymbol{N=200}$\\ $\boldsymbol{T=50}$\end{tabular}}  & Median Abs. Err.          & 0.0123                         & 0.0122                               & 0.0118                                   & 0.0114              & 0.0121              \\
                                                                                                     & \textbf{95\% CI Coverage} & \textbf{0.9390}                & \textbf{0.9390}                      & \textbf{0.9470}                          & \textbf{0.9340}     & \textbf{0.9430}     \\
                                                                                                     & Median CI Length          & 0.0625                         & 0.0636                               & 0.0629                                   & 0.0602              & 0.0626              \\ \midrule
\multirow{3}{*}{\begin{tabular}[c]{@{}c@{}}$\boldsymbol{N=200}$\\ $\boldsymbol{T=100}$\end{tabular}} & Median Abs. Err.          & 0.0116                         & 0.0120                               & 0.0107                                   & 0.0111              & 0.0114              \\
                                                                                                     & \textbf{95\% CI Coverage} & \textbf{0.9200}                & \textbf{0.9300}                      & \textbf{0.9430}                          & \textbf{0.9310}     & \textbf{0.9350}     \\
                                                                                                     & Median CI Length          & 0.0613                         & 0.0645                               & 0.0625                                   & 0.0602              & 0.0625              \\ \midrule
\multirow{3}{*}{\begin{tabular}[c]{@{}c@{}}$\boldsymbol{N=200}$\\ $\boldsymbol{T=143}$\end{tabular}} & Median Abs. Err.          & 0.0126                         & 0.0132                               & 0.0115                                   & 0.0118              & 0.0120              \\
                                                                                                     & \textbf{95\% CI Coverage} & \textbf{0.9120}                & \textbf{0.9120}                      & \textbf{0.9210}                          & \textbf{0.9140}     & \textbf{0.9190}     \\
                                                                                                     & Median CI Length          & 0.0618                         & 0.0641                               & 0.0632                                   & 0.0606              & 0.0629              \\
\bottomrule
\end{tabular}
\end{small}
\begin{tablenotes}
\footnotesize
\setlength\labelsep{0pt}
\item \emph{Notes}:  Details are provided in the notes to Table 6.
\end{tablenotes}
\end{threeparttable}
\end{center}
\end{adjustwidth}
\end{table}

The oracle benchmark that uses the true loading attains coverage between $0.921$ and $0.957$ under fast singular value decay and between $0.920$ and $0.956$ under slow singular value decay. The proposed estimator with the pseudo-distance performs very closely to this infeasible benchmark, with coverage between $0.923$ and $0.954$ and between $0.919$ and $0.953$, respectively. The configuration using five-fold cross-fitting and five-fold cross-validation again delivers coverage closer to the nominal level. The configuration without cross-fitting and with leave-one-out cross-validation yields shorter intervals and exhibits greater undercoverage, with coverage as low as $0.872$; the coverage gap between the two configurations reaches five percentage points at $N=50$ and narrows as $N$ grows. %

Under fast singular value decay, TWFE covers only $0.767$--$0.790$ at $N = 200$, which suggests its bias may be large relative to its standard error. Local PCA does not suffer from this problem but undercovers in the smaller panels, with coverage between $0.877$ and $0.890$ at $N = 50$, and reaches approximately $0.934$ at $N = 200$. Under slow singular value decay, all three feasible estimators perform much more similarly: at $N = 200$, TWFE covers $0.912$--$0.939$ and local PCA $0.912$--$0.939$. Taken together, the proposed estimator is close to the oracle benchmark across all reported designs, while the configuration without cross-fitting and with leave-one-out cross-validation generally produces shorter intervals and greater undercoverage, especially for smaller $N$. %

\FloatBarrier
\addcontentsline{toc}{section}{References}
\renewcommand\refname{References}
\bibliographystyle{apalike}
\bibliography{ref}

\newpage
\appendix
\noindent{\textbf{\Large{Online Appendix}}}

\section{Proofs} \label{app:proof}

\subsection{Supporting Lemmas}

\begin{lemma}{A.1}[Rate of Convergence for the Pseudo-Distance Estimator] \label{lem:pseudo_dist_rate}
Suppose Assumption \ref{asm:4} holds and $\sqrt{\frac{\log(N)}{T_{0}}}\to0$.
Then
\begin{align*}
\max_{i,j} \left\vert
\hat{d}_{i,j} - d_{i,j}
\right\vert
=\mathcal{O}_{p}\left(\sqrt{\frac{\log(N)}{T_{0}}}+\bigg(\frac{\log(N)}{N}\bigg)^\frac{1}{d_{\alpha}}\right).
\end{align*}
\end{lemma}

\begin{proof}[Proof of Lemma \ref{lem:pseudo_dist_rate}]
By Assumptions \ref{asm:4} \ref{asm:4-4} we can write $d_{i,j}$ as follows:
\begin{align*}
d_{i,j} & =\sup_{\alpha_{k}\in\mathcal{A}}\left\vert  \frac{1}{T_0}\sum_{t=1}^{T_0}\mathbb{E}\big[\mu^{(0)}(\alpha_{k},\lambda_{t})\big(\mu^{(0)}(\alpha_{i},\lambda_{t})-\mu^{(0)}(\alpha_{j},\lambda_{t})\big)|\alpha_{i},\alpha_{j},\alpha_{k}\big]\right\vert .
\end{align*}
We further denote the oracle distance $\tilde{d}_{i,j}$, which will be helpful throughout the proof, as follows:
\begin{align*}
\tilde{d}_{i,j}
:= \max\limits_{k\notin\{i,j\}} \left\vert
\frac{1}{T_{0}}\sum_{t=1}^{T_{0}}\mu_{k,t}^{(0)}(\mu_{i,t}^{(0)}-\mu_{j,t}^{(0)})
\right\vert.
\end{align*}
The proof is then separated into two parts: (i) the rate of convergence from $\hat{d}_{i,j}$ to $\tilde{d}_{i,j}$ and (ii) the rate of convergence from $\tilde{d}_{i,j}$ to $d_{i,j}$.

\paragraph{From $\hat{d}_{i,j}$ to $\tilde{d}_{i,j}$}

By applying the reverse triangle inequality and multiplying out, we derive the max difference between estimated distance and oracle distance as
\begin{align}
&\hspace{5pt} \max_{i,j} \left\vert \hat{d}_{i,j}-\tilde{d}_{i,j} \right\vert \nonumber \\
=&\hspace{5pt} \max_{i,j}\left\vert  \max\limits_{k\notin\{i,j\}}
\left\vert \frac{1}{T_{0}}\sum_{t=1}^{T_{0}}(\mu_{k,t}^{(0)}+u_{k,t})(\mu_{i,t}^{(0)}-\mu_{j,t}^{(0)}+u_{i,t}-u_{j,t}) \right\vert
- \max\limits_{k\notin\{i,j\}} \left\vert\frac{1}{T_{0}}\sum_{t=1}^{T_{0}}\mu_{k,t}^{(0)}(\mu_{i,t}^{(0)}-\mu_{j,t}^{(0)}) \right\vert
\right\vert\nonumber \\
\leq&\hspace{5pt} \max\limits_{i,j,k:\,k\notin\{i,j\}} \left\vert
\frac{1}{T_{0}} \sum_{t=1}^{T_{0}} \left[ u_{k,t}\left(\mu_{i,t}^{(0)}-\mu_{j,t}^{(0)}\right) + \mu_{k,t}^{(0)}\left(u_{i,t}-u_{j,t}\right) + u_{k,t}\left(u_{i,t}-u_{j,t}\right) \right]
\right\vert. \label{eq:d_difference}
\end{align}
By (\ref{outcome_model}) and Assumption \ref{asm:4} \ref{asm:4-3}, there exist constants $C$, $a$, and $b$, so that for any $\eta>0$,
\begin{align}
\mathbb{P}\left(\left\vert\frac{1}{T_{0}}\sum_{t=1}^{T_{0}}\left[ (u_{k,t}(\mu_{i,t}^{(0)}-\mu_{j,t}^{(0)})+\mu_{k,t}^{(0)}(u_{i,t}-u_{j,t})+u_{k,t}(u_{i,t}-u_{j,t})\right]\right\vert>\eta\right)
\leq C\exp\left(-\frac{T_{0}\eta^{2}}{a+b\eta}\right).\label{imp1}
\end{align}
To see this, note that the independence conditions in Assumption \ref{asm:4} \ref{asm:4-3} imply that for any $t\leq T_0$ and any $r$ and $l$ with $l\neq r$ in $\{i,j,k\}$,  $u_{r,t}\indep\alpha_{l},\{\lambda_{s}\}_{s\neq t}\big\vert\hspace{3pt}\alpha_{r},\lambda_t$,\footnote{This and the other conditional independencies in this paragraph can be derived from the two conditional independence restrictions in \ref{asm:4} \ref{asm:4-3} by repeated application of the rule $A\indep B,C|D\implies(A\indep B|C,D\,\text{and }A\indep B|D)$, and the rule $(A\indep B|C,D\text{ and }A\indep C|D)\implies A\indep B,C|D$, for any random variables $A,B,C,D$.} which in turn implies $\mathbb{E}[u_{r,t}|\alpha_{l},\alpha_{r},\{\lambda_{s}\}_{s=1}^{T_0}]=\mathbb{E}[u_{r,t}|\alpha_{r},\lambda_t]$ and this is zero by (\ref{outcome_model}), and so  $\mathbb{E}[u_{r,t}|\alpha_{l},\alpha_{r},\{\lambda_{s}\}_{s=1}^{T_0}]=0$. Additionally, by Assumption \ref{asm:4} \ref{asm:4-3}, $u_{r,t}\indep u_{l,t}\big\vert\hspace{3pt}\alpha_{r},\alpha_{l},\{\lambda_{s}\}_{s=1}^{T_0}$. Moreover, by Assumption \ref{asm:4} \ref{asm:4-5}, $\mu_{i,t}^{(0)}$, $\mu_{j,t}^{(0)}$, and $\mu_{k,t}^{(0)}$ are bounded. By Assumption \ref{asm:4} \ref{asm:4-3}, for each $t\leq T_0$ and $r,l\in\{i,j,k\}$,  $(u_{r,t},u_{l,t})\indep\{(u_{r,s},u_{l,s})\}_{s=1}^{t-1}\big\vert\hspace{3pt}\alpha_{r},\alpha_{l},\{\lambda_{s}\}_{s=1}^{T_{0}}$, and $u_{r,t}$ is sub-Gaussian conditional on $\alpha_{r},\alpha_{l}$, and $\{\lambda_{s}\}_{s=1}^{T_{0}}$. Thus, from the Bernstein inequalities in Lemma \ref{lem:bi_subexp}, we get that there exist constants $C$, $a$, and $b$, so that for any $\eta>0$,
\begin{align*}
\mathbb{P}\left(\left|\frac{1}{T_{0}}\sum_{t=1}^{T_{0}}u_{k,t}\mu_{r,t}^{(0)}\right|>\frac{\eta}{6}\hspace{3pt}\big\vert\hspace{3pt}\alpha_{r},\alpha_{k},\{\lambda_{t}\}_{t=1}^{T_0}\right) & \leq\frac{C}{6}\exp\left(-\frac{T_{0}\eta^{2}}{a+b\eta}\right),\,r=i,j\\
\mathbb{P}\left(\left|\frac{1}{T_{0}}\sum_{t=1}^{T_{0}}\mu_{k,t}^{(0)}u_{r,t}\right|>\frac{\eta}{6}\hspace{3pt}\big\vert\hspace{3pt}\alpha_{r},\alpha_{k},\{\lambda_{t}\}_{t=1}^{T_0}\right) & \leq\frac{C}{6}\exp\left(-\frac{T_{0}\eta^{2}}{a+b\eta}\right),\,r=i,j\\
\mathbb{P}\left(\left|\frac{1}{T_{0}}\sum_{t=1}^{T_{0}}u_{k,t}u_{r,t}\right|>\frac{\eta}{6}\hspace{3pt}\big\vert\hspace{3pt}\alpha_{r},\alpha_{k},\{\lambda_{t}\}_{t=1}^{T_0}\right) & \leq\frac{C}{6}\exp\left(-\frac{T_{0}\eta^{2}}{a+b\eta}\right),\,r=i,j
\end{align*}
These inequalities hold for some $a$, $b$, and $C$ by the Bernstein bounds in Lemma \ref{lem:bi_subexp}. Applying the law of iterated expectations to integrate out the factors, and then using the triangle inequality and the union bound yields (\ref{imp1}).

By taking the union bound and (\ref{eq:d_difference}), we then get from (\ref{imp1}) that
\[
\mathbb{P}\left(\max_{i,j} \left\vert \hat{d}_{i,j}-\tilde{d}_{i,j} \right\vert > \eta\right)
\leq N^{3}C\exp\left(-\frac{T_{0}\eta^{2}}{a+b\eta}\right).
\]
Set $\eta=\sqrt{\dfrac{Ma\log(N)}{T_{0}}}$ for some fixed $M>3$, in which case by $\dfrac{\log(N)}{T_{0}}\to0$, we have $\eta\to 0$. Then by the above inequality
\begin{align*}
\mathbb{P}\left(\max_{i,j} \left\vert \hat{d}_{i,j}-\tilde{d}_{i,j} \right\vert > \sqrt{\frac{Ma\log(N)}{T_{0}}}\right)
\leq C\exp\left( \frac{(3-M)a+3b\eta}{a+b\eta}\log(N)\right)\to 0,
\end{align*}
and so we conclude that
\begin{align*}
\max_{i,j} \left\vert \hat{d}_{i,j}-\tilde{d}_{i,j} \right\vert
= \mathcal{O}_{p}\left(\sqrt{\frac{\log(N)}{T_{0}}}\right).
\end{align*}

\paragraph{From $\tilde{d}_{i,j}$ to $d_{i,j}$}

By the triangle inequality and reverse triangle inequality, the maximum difference between the population and oracle distances satisfies
\begin{align}
&\hspace{5pt} \max_{i,j} \left| d_{i,j}-\tilde{d}_{i,j} \right| \nonumber \\
\leq&\hspace{5pt}
\max_{i,j,k:\,k\notin\{i,j\}}
\left|
\frac{1}{T_0}\sum_{t=1}^{T_0}
\mu_{k,t}^{(0)}\left(\mu_{i,t}^{(0)}-\mu_{j,t}^{(0)}\right)
-
\mathbb{E}\left[
\frac{1}{T_0}\sum_{t=1}^{T_0}
\mu_{k,t}^{(0)}\left(\mu_{i,t}^{(0)}-\mu_{j,t}^{(0)}\right)
\mid \alpha_i,\alpha_j,\alpha_k
\right]
\right| \label{eq:muline1}\\
+&\hspace{5pt}
\max_{i,j}\left[
\begin{aligned}
&\sup_{\alpha\in\mathcal A}
\left|
\mathbb{E}\left[
\frac{1}{T_0}\sum_{t=1}^{T_0}
\mu^{(0)}(\alpha,\lambda_t)
\left(\mu_{i,t}^{(0)}-\mu_{j,t}^{(0)}\right)
\mid \alpha_i,\alpha_j
\right]
\right|\\
&-
\max_{k:\,k\notin\{i,j\}}
\left|
\mathbb{E}\left[
\frac{1}{T_0}\sum_{t=1}^{T_0}
\mu_{k,t}^{(0)}
\left(\mu_{i,t}^{(0)}-\mu_{j,t}^{(0)}\right)
\mid \alpha_i,\alpha_j,\alpha_k
\right]
\right|
\end{aligned}
\right]. \label{eq:muline2}
\end{align}
We first derive the rate for (\ref{eq:muline1}).
By Assumption \ref{asm:4}, conditions \ref{asm:4-4} and
\ref{asm:4-5}, the vectors
$(\mu_{i,t}^{(0)},\mu_{j,t}^{(0)},\mu_{k,t}^{(0)})$
are uniformly bounded and independent across $t$, conditional on
$\alpha_i,\alpha_j,\alpha_k$.
Therefore, by the Bernstein inequality for bounded random variables
(Lemma \ref{lem:bi}), there are positive constants $C$, $a$, and $b$,
uniform in $i,j,k,N,T_0$, such that, for any $\eta>0$,
\begin{align*}
&\mathbb{P}\left(
\left|
\frac{1}{T_0}\sum_{t=1}^{T_0}
\left[
\mu_{k,t}^{(0)}
\left(\mu_{i,t}^{(0)}-\mu_{j,t}^{(0)}\right)
-
\mathbb{E}\left[
\mu_{k,t}^{(0)}
\left(\mu_{i,t}^{(0)}-\mu_{j,t}^{(0)}\right)
\mid \alpha_i,\alpha_j,\alpha_k
\right]
\right]
\right|>\eta
\;\middle|\;
\alpha_i,\alpha_j,\alpha_k
\right)\\
&\hspace{20pt}\leq
C\exp\left(-\frac{T_0\eta^2}{a+b\eta}\right).
\end{align*}
Integrating out the individual factors gives
\begin{align*}
&\mathbb{P}\left(
\left|
\frac{1}{T_0}\sum_{t=1}^{T_0}
\left[
\mu_{k,t}^{(0)}
\left(\mu_{i,t}^{(0)}-\mu_{j,t}^{(0)}\right)
-
\mathbb{E}\left[
\mu_{k,t}^{(0)}
\left(\mu_{i,t}^{(0)}-\mu_{j,t}^{(0)}\right)
\mid \alpha_i,\alpha_j,\alpha_k
\right]
\right]
\right|>\eta
\right)\\
&\hspace{20pt}\leq
C\exp\left(-\frac{T_0\eta^2}{a+b\eta}\right).
\end{align*}
Applying the union bound over at most $N^3$ triples yields
\begin{align*}
&\mathbb{P}\left(
\max_{i,j,k:\,k\notin\{i,j\}}
\left|
\frac{1}{T_0}\sum_{t=1}^{T_0}
\left[
\mu_{k,t}^{(0)}
\left(\mu_{i,t}^{(0)}-\mu_{j,t}^{(0)}\right)
-
\mathbb{E}\left[
\mu_{k,t}^{(0)}
\left(\mu_{i,t}^{(0)}-\mu_{j,t}^{(0)}\right)
\mid \alpha_i,\alpha_j,\alpha_k
\right]
\right]
\right|>\eta
\right)\\
&\hspace{20pt}\leq
CN^3\exp\left(-\frac{T_0\eta^2}{a+b\eta}\right).
\end{align*}
Set $\eta=\sqrt{Ma\log(N)/T_0}$ for some fixed $M>3$.
Since $\log(N)/T_0\to0$, we have $\eta\to0$, and the last bound becomes
\[
C\exp\left[
\left(3-\frac{Ma}{a+b\eta}\right)\log(N)
\right]\to0.
\]
Consequently,
\begin{align*}
&\max_{i,j,k:\,k\notin\{i,j\}}
\left|
\frac{1}{T_0}\sum_{t=1}^{T_0}
\left[
\mu_{k,t}^{(0)}
\left(\mu_{i,t}^{(0)}-\mu_{j,t}^{(0)}\right)
-
\mathbb{E}\left[
\mu_{k,t}^{(0)}
\left(\mu_{i,t}^{(0)}-\mu_{j,t}^{(0)}\right)
\mid \alpha_i,\alpha_j,\alpha_k
\right]
\right]
\right|\\
&\hspace{20pt}=
\mathcal O_p\left(\sqrt{\frac{\log(N)}{T_0}}\right).
\end{align*}

We next derive the rate for (\ref{eq:muline2}).
The terms with $i=j$ are zero, so it suffices to consider distinct pairs.
By the reverse triangle inequality, boundedness of $\mu^{(0)}$, and
the Lipschitz condition in Assumption \ref{asm:4}\ref{asm:4-5},
there is a constant $C$ such that
\begin{align*}
&\max_{i\ne j}\left[
\begin{aligned}
&\sup_{\alpha\in\mathcal A}
\left|
\mathbb{E}\left[
\frac{1}{T_0}\sum_{t=1}^{T_0}
\mu^{(0)}(\alpha,\lambda_t)
\left(\mu_{i,t}^{(0)}-\mu_{j,t}^{(0)}\right)
\mid \alpha_i,\alpha_j
\right]
\right|\\
&-
\max_{k:\,k\notin\{i,j\}}
\left|
\mathbb{E}\left[
\frac{1}{T_0}\sum_{t=1}^{T_0}
\mu_{k,t}^{(0)}
\left(\mu_{i,t}^{(0)}-\mu_{j,t}^{(0)}\right)
\mid \alpha_i,\alpha_j,\alpha_k
\right]
\right|
\end{aligned}
\right]\\
&\leq
\max_{i\ne j}\sup_{\alpha\in\mathcal A}
\min_{k\notin\{i,j\}}
\left|
\mathbb{E}\left[
\frac{1}{T_0}\sum_{t=1}^{T_0}
\left(
\mu^{(0)}(\alpha,\lambda_t)
-\mu^{(0)}(\alpha_k,\lambda_t)
\right)
\left(\mu_{i,t}^{(0)}-\mu_{j,t}^{(0)}\right)
\mid \alpha_i,\alpha_j,\alpha_k
\right]
\right|\\
&\leq
C\max_{i\ne j}\sup_{\alpha\in\mathcal A}
\min_{k\notin\{i,j\}}\|\alpha-\alpha_k\|.
\end{align*}
Define
\[
R_N^{(-2)}
:=
\max_{i\ne j}\sup_{\alpha\in\mathcal A}
\min_{k\notin\{i,j\}}\|\alpha-\alpha_k\|.
\]

By the small-ball condition in Assumption \ref{asm:4}\ref{asm:4-4},
for all sufficiently small $\eta>0$, there is a deterministic
$\eta$-covering of $\mathcal A$ with centres
$\{\bar\alpha_\ell\}_{\ell=1}^J\subset\mathcal A$ and
$J\leq C\eta^{-d_\alpha}$.
The triangle inequality therefore gives
\[
R_N^{(-2)}
\leq
\eta+
\max_{i\ne j}\max_{\ell\in[J]}
\min_{k\notin\{i,j\}}\|\bar\alpha_\ell-\alpha_k\|.
\]

The same small-ball condition implies that there exist
$\delta>0$ and $\epsilon_0>0$ such that, for every
$\alpha\in\mathcal A$ and $0<\epsilon\leq\epsilon_0$,
\[
\mathbb P(\|\alpha-\alpha_k\|>\epsilon)
\leq 1-\delta\epsilon^{d_\alpha}.
\]
For each fixed pair $i\ne j$ and each $\ell\in[J]$,
independence of the individual factors gives
\begin{align*}
\mathbb P\left(
\min_{k\notin\{i,j\}}
\|\bar\alpha_\ell-\alpha_k\|>\epsilon
\right)
&=
\prod_{k\notin\{i,j\}}
\mathbb P\left(
\|\bar\alpha_\ell-\alpha_k\|>\epsilon
\right)\\
&\leq
(1-\delta\epsilon^{d_\alpha})^{N-2}.
\end{align*}
Taking a union bound over both the excluded pairs and the covering
centres, we obtain
\begin{align*}
\mathbb P\left(R_N^{(-2)}>\eta+\epsilon\right)
&\leq
\mathbb P\left(
\max_{i\ne j}\max_{\ell\in[J]}
\min_{k\notin\{i,j\}}
\|\bar\alpha_\ell-\alpha_k\|>\epsilon
\right)\\
&\leq
\sum_{i\ne j}\sum_{\ell=1}^J
\mathbb P\left(
\min_{k\notin\{i,j\}}
\|\bar\alpha_\ell-\alpha_k\|>\epsilon
\right)\\
&\leq
CN^2\eta^{-d_\alpha}
(1-\delta\epsilon^{d_\alpha})^{N-2}.
\end{align*}

Set
\[
\eta_N=\frac1{N-2},
\qquad
\epsilon_N=
\left(
\frac{M\log(N-2)}{\delta(N-2)}
\right)^{1/d_\alpha},
\]
for some fixed $M>d_\alpha+2$.
For all sufficiently large $N$, these radii lie within the ranges
required above, and
\begin{align*}
\mathbb P\left(R_N^{(-2)}>\eta_N+\epsilon_N\right)
&\leq
CN^2(N-2)^{d_\alpha}
\left(
1-\frac{M\log(N-2)}{N-2}
\right)^{N-2}\\
&\leq
CN^2(N-2)^{d_\alpha-M}
\longrightarrow0.
\end{align*}
Since $\eta_N=o(\epsilon_N)$ for $d_\alpha\geq1$, it follows that
\[
\max_{i\ne j}\sup_{\alpha\in\mathcal A}
\min_{k\notin\{i,j\}}\|\alpha-\alpha_k\|
=
\mathcal O_p\left(
\left(\frac{\log(N)}{N}\right)^{1/d_\alpha}
\right).
\]

Combining the bounds for (\ref{eq:muline1}) and
(\ref{eq:muline2}) with the preceding bound for
$\max_{i,j}|\hat d_{i,j}-\tilde d_{i,j}|$, and applying the triangle
inequality, we conclude that
\[
\max_{i,j}|\hat d_{i,j}-d_{i,j}|
=
\mathcal O_p\left(
\sqrt{\frac{\log(N)}{T_0}}
+
\left(\frac{\log(N)}{N}\right)^{1/d_\alpha}
\right).
\]
\end{proof}

\begin{lemma}{A.2} \label{lem:dist_char}
Under Assumption \ref{asm:4} \ref{asm:4-8}, there exists a $c>0$ so that for all sufficiently large $T_0$ and $\alpha_{1},\alpha_{2}\in\mathcal{A}$:
\[
\sup_{\alpha\in\mathcal{A}}\big|\int_{\lambda\in\supp(\lambda)}\mu^{(0)}(\alpha,\lambda)\big(\mu^{(0)}(\alpha_{1},\lambda)-\mu^{(0)}(\alpha_{2},\lambda)\big)\dd\bar{\pi}_{T_0}(\lambda)\big|\geq c\|\alpha_{1}-\alpha_{2}\|.
\]
\end{lemma}

\begin{proof}[Proof of Lemma \ref{lem:dist_char}]
By Assumption \ref{asm:4} \ref{asm:4-8}, there exists a constant $c_2>0$ such that
\begin{align*}
&\hspace{3pt} c_2\sqrt{\int_{\lambda\in\supp(\lambda)}\left(\mu^{(0)}(\alpha_{1},\lambda)-\mu^{(0)}(\alpha_{2},\lambda)\right)^{2}\dd\bar{\pi}_{T_0}(\lambda)}\\
\leq&\hspace{3pt}  \sup_{\alpha\in\mathcal{A}}\big|\int_{\lambda\in\supp(\lambda)}\mu^{(0)}(\alpha,\lambda)\left(\mu^{(0)}(\alpha_{1},\lambda)-\mu^{(0)}(\alpha_{2},\lambda)\right)\dd\bar{\pi}_{T_0}(\lambda)\big|,
\end{align*}
and the object on the LHS is bounded below by $c_2\sqrt{c_1}\|\alpha_1-\alpha_2\|$. Thus letting $c=c_2\sqrt{c_1}>0$ gives the result.
\end{proof}

\begin{lemma}{A.3} \label{lem:dist_kernel_bound}
Suppose Assumption \ref{asm:4} holds and that $\frac{1}{h}\sqrt{\frac{\log(N)}{T_{0}}}$ and $\frac{\log(N)}{Nh^{d_{\alpha}}}$ go to zero. Then there exists a $\delta>0$ so that with probability approaching $1$,
for all $i,j\in[N]$
\begin{align*}
K\left(\frac{\hat{d}_{i,j}}{h}\right)
= \mathbf{1} \left\{\Vert\alpha_{i}-\alpha_{j}\Vert \leq \delta h \right\} K\left(\frac{\hat{d}_{i,j}}{h}\right).
\end{align*}
\end{lemma}

\begin{proof}[Proof of Lemma \ref{lem:dist_kernel_bound}]

By Lemma \ref{lem:dist_char}, there exists a constant $C$, such that for any finite $\delta>0$, we have
\begin{align*}
\Vert\alpha_{i}-\alpha_{j}\Vert > \delta h\implies d_{i,j}>C\delta h
\end{align*}
Now, note that for all $i,j$,
\begin{align*}
\frac{\hat{d}_{i,j}}{h}\geq & \frac{d_{i,j}}{h}-\frac{\max_{i,j} \left\vert \hat{d}_{i,j}-d_{i,j} \right\vert}{h}.
\end{align*}
Thus, it follows that
\begin{align*}
\Vert\alpha_{i}-\alpha_{j}\Vert > \delta h
\implies
\frac{\hat{d}_{i,j}}{h} > \delta C - \frac{\max_{i,j} \left\vert\hat{d}_{i,j}-d_{i,j}\right\vert}{h}
\end{align*}
By Lemma \ref{lem:pseudo_dist_rate},  we have
\begin{align*}
\frac{\max_{i,j} \left\vert\hat{d}_{i,j}-d_{i,j}\right\vert}{h}
= \mathcal{O}_{p}\left(\frac{1}{h}\sqrt{\frac{\log(N)}{T_{0}}}+\bigg(\frac{\log(N)}{N h^{d_{\alpha}}}\bigg)^\frac{1}{d_{\alpha}}\right)
\end{align*}
By supposition, $\frac{1}{h}\sqrt{\frac{\log(N)}{T_{0}}}$ and $\frac{\log(N)}{Nh^{d_{\alpha}}}$ go to zero, and so for any fixed $\delta>0$, with probability approaching one, we have $\frac{\max_{i,j} \left\vert\hat{d}_{i,j}-d_{i,j}\right\vert}{h} \leq \frac{1}{2}\delta C$. Hence, for any fixed $\delta>0$, with probability approaching one, for all $i,j$, it follows that
\begin{align*}
\Vert\alpha_{i}-\alpha_{j}\Vert > \delta h
\implies
\frac{\hat{d}_{i,j}}{h} > \frac{1}{2}\delta C
\end{align*}
Assumption \ref{asm:4} \ref{asm:4-6} states that the kernel is zero everywhere other than a compact (and therefore bounded) set. Hence, there is a constant $c$ such that
\begin{align*}
\frac{\hat{d}_{i,j}}{h} > c
\implies
K\left(\frac{\hat{d}_{i,j}}{h}\right) = 0
\end{align*}
Finally, fixing some $\delta>\frac{2c}{C}$ gives the desired result.
\end{proof}

\begin{lemma}{A.4} \label{lem:char_dist_bound}
Suppose Assumption \ref{asm:4} holds, $\frac{\log N}{Nh^{d_\alpha}} \to 0$, and $h\to0$. Then for any fixed $\delta>0$, there are constants $0<\underline{c}<\bar{c}<\infty$ and $0<\tilde{c}$, such that for sufficiently large $T_0$,  with probability approaching one and for all $i\in[N]$, it follows that
\begin{align*}
\underline{c}h^{d_{\alpha}}
\leq \frac{1}{N}\sum_{j=1}^N (1-w_{j,t})\mathbf{1}\left\{\Vert \alpha_{i} - \alpha_{j} \Vert \leq \delta h\right\}
\leq \bar{c} h^{d_{\alpha}}
\end{align*}
and
\begin{align*}
\underline{c}h^{d_{\alpha}}
\leq \frac{1}{N}\sum_{j=1}^N \mathbf{1}\left\{\Vert \alpha_{i} - \alpha_{j} \Vert \leq \delta h\right\}
\leq \bar{c} h^{d_{\alpha}}
\end{align*}
and in addition,
$
\tilde{c}\leq\frac{1}{Nh^{d_{\alpha}}}\sum_{j=1}^N (1-w_{j,t})K(d_{i,j}/h)
$ and $
\tilde{c}\leq\frac{1}{Nh^{d_{\alpha}}}\sum_{j=1}^N K(d_{i,j}/h).
$
\end{lemma}

\begin{proof}[Proof of Lemma \ref{lem:char_dist_bound}]

Define $X_{i,j}= \mathbf{1} \left\{\Vert\alpha_{i}-\alpha_{j}\Vert\leq\delta h \right\}$ and $X^{(w)}_{i,j}=(1-w_{j,t}) \mathbf{1} \left\{\Vert\alpha_{i}-\alpha_{j}\Vert\leq\delta h \right\}$.
By the upper bound in Assumption \ref{asm:4} \ref{asm:4-2}, for some $c>0$, $\mathbb{E}_t[1-w_{j,t}|\alpha_{j}]=1-p_{j,t}\geq c$ almost surely.  Using the fact that $w_{j,t}$ and $\alpha_{j}$ are jointly independent of $\alpha_{i}$ and that $\{\alpha_{j}\}_{j=1}^{N}$ are i.i.d. and independent of $\lambda_t$ (by Assumptions  \ref{asm:4} \ref{asm:4-3} and \ref{asm:4} \ref{asm:4-4}), we have
\begin{align*}
\mathbb{E}_t \left[X^{(w)}_{i,j} \hspace{3pt}\big\vert\hspace{3pt} \alpha_{i}\right]
& = \mathbb{E}_t\left[\mathbb{E}_t[1-w_{j,t}|\alpha_{j}] \mathbf{1} \left\{\|\alpha_{i}-\alpha_{j}\|\leq\delta h\right\} \hspace{3pt}\big\vert\hspace{3pt} \alpha_{i}\right]\\
& \geq c \cdot \mathbb{E}_t\left[ \mathbf{1} \left\{\Vert\alpha_{i}-\alpha_{j}\Vert \leq \delta h \right\} \hspace{3pt}\big\vert\hspace{3pt} \alpha_{i}\right]\\
&= c \cdot \mathbb{P}\left(\Vert\alpha_{i}-\alpha_{j}\Vert\leq\delta h \hspace{3pt}\big\vert\hspace{3pt} \alpha_{i}\right).
\end{align*}
To upper bound the quantity note that
\begin{align*}
\mathbb{E}_t \left[X^{(w)}_{i,j} \hspace{3pt}\big\vert\hspace{3pt} \alpha_{i}\right]
& \leq \mathbb{E}_t\left[\mathbf{1} \left\{\|\alpha_{i}-\alpha_{j}\|\leq\delta h\right\} \hspace{3pt}\big\vert\hspace{3pt} \alpha_{i}\right]\\
&=
\mathbb{P}\left(\Vert\alpha_{i}-\alpha_{j}\Vert\leq\delta h \hspace{3pt}\big\vert\hspace{3pt} \alpha_{i}\right).
\end{align*}
With Assumption \ref{asm:4} \ref{asm:4-4}, we then see that for any fixed $\delta$ and sufficiently small $h>0$, there are constants $0<c_{-}<c_{+}<\infty$ such that
\begin{align*}
c_{-} h^{d_{\alpha}}
\leq \mathbb{E}_t \left[X^{(w)}_{i,j} \hspace{3pt}\big\vert\hspace{3pt} \alpha_{i}\right]
\leq c_{+}h^{d_{\alpha}}.
\end{align*}
By Assumption \ref{asm:4} \ref{asm:4-4},  $\mathbb{E}_t \left[X_{i,j} \hspace{3pt}\big\vert\hspace{3pt} \alpha_{i}\right]=\mathbb{P}\left(\Vert\alpha_{i}-\alpha_{j}\Vert\leq\delta h \hspace{3pt}\big\vert\hspace{3pt} \alpha_{i}\right)$ and so we see that the above also holds for $X^{(w)}_{i,j}$ replaced by $X_{i,j}$.
Note that the variables $X^{(w)}_{i,j}$ for $j\neq i$ are independent Bernoulli conditional on $\alpha_i$ and $\lambda_t$, and likewise for $X_{i,j}$, and so the above also demonstrates that they have conditional variance below $c_{+}h^{d_{\alpha}}$. So applying the Bernstein inequality in Lemma \ref{lem: A bi unb}, it follows that, for a sufficiently small fixed $a>0$ and setting $C=\frac{a^2}{2(c_{+} + a/3)}$,
\begin{align*}
&\hspace{5pt} \mathbb{P}_t\left(\max_{i\in[N]}\left\vert\frac{1}{N}\sum_{j\neq i}\left(X^{(w)}_{i,j}
- \mathbb{E}_t\left[X^{(w)}_{i,j}|\alpha_i\right]\right)\right\vert > a h^{d_{\alpha}}\right)\\
\leq&\hspace{5pt} 2N \cdot \exp\left(-C N h^{d_{\alpha}}\right)\\
=&\hspace{5pt} 2 \cdot \exp\left(\log(N)-CNh^{d_{\alpha}}\right),
\end{align*}
and similarly with $X_{i,j}^{(w)}$ replaced by $X_{i,j}$.
The RHS goes to zero if $\frac{\log N}{Nh^{d_{\alpha}}}\to0$, which holds by supposition, and so with probability approaching one, for all $i\in[N]$,
\begin{align*}
(\frac{N-1}{N}{c}_{-}-a)h^{d_{\alpha}}
\leq \frac{1}{N}\sum_{j\neq i} X^{(w)}_{i,j}
\leq ({c}_{+}+a)h^{d_{\alpha}},
\end{align*}
in which case, including the $j=i$ term in the sum,
\begin{align*}
(\frac{N-1}{N}{c}_{-}-a)h^{d_{\alpha}}
\leq \frac{1}{N}\sum_{j=1}^N X^{(w)}_{i,j}
\leq ({c}_{+}+a)h^{d_{\alpha}}+1/N,
\end{align*}
and the same steps apply with $X^{(w)}_{i,j}$ replaced by $X_{i,j}$. By supposition $1/N=o(h^{d_{\alpha}})$, and so by setting $a$ sufficiently small, we get the first two results in the Lemma.

Now for the final two statements in Lemma \ref{lem:char_dist_bound}, by Assumption \ref{asm:4} \ref{asm:4-5}, it follows that $\mu^{(0)}(\cdot,\lambda)$
is Lipschitz continuous and bounded (uniformly over $\lambda$ in
both cases). Thus, for some constant $C<\infty$, we have
\begin{align*}
\Vert\alpha_{i}-\alpha_{j}\Vert\leq x
\implies
d_{i,j}\leq Cx.
\end{align*}
Since Assumption \ref{asm:4} \ref{asm:4-6} requires the kernel to be weakly positive, strictly positive at zero, and Lipschitz continuous, there are constants $a,b>0$ such that for all $x\leq a$, $K(x)\geq b$. Combining these results gives
\begin{align*}
\Vert\alpha_{i}-\alpha_{j}\Vert\leq\frac{a}{C}h
\implies
d_{i,j}\leq ah
\implies
K(d_{i,j}/h)\geq b.
\end{align*}
Then with the positivity of the kernel, it follows that
\begin{align*}
b \cdot \frac{1}{N} \sum_{j=1}^N(1-w_{j,t}) \mathbf{1} \left\{\Vert\alpha_{i}-\alpha_{j}\Vert \leq \frac{a}{C} \cdot h \right\}
\leq \frac{1}{N}\sum_{j=1}^N(1-w_{j,t}) \cdot K(d_{i,j}/h)
\end{align*}
and
\begin{align*}
b \cdot \frac{1}{N} \sum_{j=1}^N \mathbf{1} \left\{\Vert\alpha_{i}-\alpha_{j}\Vert \leq \frac{a}{C} \cdot h \right\}
\leq \frac{1}{N}\sum_{j=1}^N  K(d_{i,j}/h).
\end{align*}
We already established that for any fixed $\delta$, there is
a constant $\underline{c}>0$ such that $\underline{c}h^{d_\alpha}\leq\frac{1}{N}\sum_{j=1}^N(1-w_{j,t}) \mathbf{1} \left\{\Vert\alpha_{i}-\alpha_{j}\Vert \leq \delta h \right\}$ and $\underline{c}h^{d_\alpha}\leq\frac{1}{N}\sum_{j=1}^N \mathbf{1} \left\{\Vert\alpha_{i}-\alpha_{j}\Vert \leq \delta h \right\}$ for all $i\in[N]$ with probability approaching one, therefore, there is a constant $\tilde{c}>0$ such that with probability approaching one, for all $i$,
\begin{align*}
\tilde{c}\leq\frac{1}{Nh^{d_{\alpha}}}\sum_{j=1}^N(1-w_{j,t}) \cdot K(d_{i,j}/h)
\end{align*}
and similarly
$
\tilde{c}\leq\frac{1}{Nh^{d_{\alpha}}}\sum_{j=1}^N K(d_{i,j}/h)$.
\end{proof}

\subsection{Proof of Theorem \ref{thm:1}}

\begin{proof}[Proof of Theorem \ref{thm:1}]
We prove the result for $\hat{\mu}^{(0)}_{i,t}$. The result for $\hat{p}_{i,t}$ follows by almost identical steps, replacing $u_{j,t}$ by $\epsilon_{j,t}$, dropping the restriction to $\{j:w_{j,t}=0\}$ in the kernel sums, using that $\epsilon_{j,t}\in[-1,1]$ is bounded and hence sub-Gaussian, and removing conditioning on current treatments throughout the argument.

By Lemma \ref{lem:dist_kernel_bound}, there exists a $\delta < \infty$ such that with probability approaching one, for all $i,j\in[N]$, we have
\begin{align*}
K\left(\hat{d}_{i,j} / h\right)
= K \left(\hat{d}_{i,j}/h\right) \cdot \mathbf{1} \left\{ \Vert \alpha_{i}-\alpha_{j} \Vert \leq\delta h \right\}
\end{align*}
and by Lemma \ref{lem:dist_char} the same holds with $\hat{d}_{i,j}$ replaced by ${d}_{i,j}$. Fix such a $\delta$. In the event that these equalities hold, we can then decompose $\hat{\mu}_{i,t}^{(0)}-\mu_{i,t}^{(0)}$
as follows:
\begin{align}
\hat{\mu}_{i,t}^{(0)}-\mu_{i,t}^{(0)}
&= \frac{\sum\limits_{j;w_{j,t}=0} \mathbf{1}\left\{ \Vert \alpha_{i}-\alpha_{j} \Vert \leq \delta h\right\} K\left(\hat{d}_{i,j}/h\right) \cdot \mu_{j,t}^{(0)}}{\sum\limits_{j;w_{j,t}=0}\mathbf{1}\left\{ \Vert \alpha_{i}-\alpha_{j} \Vert \leq \delta h\right\} K\left(\hat{d}_{i,j}/h\right)} - \mu_{i,t}^{(0)}\nonumber \\
&+ \frac{\dfrac{1}{h^{d_{\alpha}}N}\sum\limits_{j;w_{j,t}=0} \mathbf{1}\left\{ \Vert \alpha_{i}-\alpha_{j} \Vert \leq \delta h\right\} K\left(\hat{d}_{i,j}/h\right) \cdot u_{j,t}}{\dfrac{1}{h^{d_{\alpha}}N}\sum\limits_{j;w_{j,t}=0}K\left(\hat{d}_{i,j}/h\right)}. \label{eq:secon_line}
\end{align}
The quantity on the first line of the RHS above is a weighted average of $\mu_{j,t}^{(0)}-\mu_{i,t}^{(0)}$ over untreated individuals $j$ for whom $\Vert\alpha_i-\alpha_j\Vert\leq\delta h$.
Using the Lipschitz condition in Assumption \ref{asm:4} \ref{asm:4-5}, we have
\begin{align*}
\left\vert\frac{\sum\limits_{j;w_{j,t}=0} \mathbf{1} \{\|\alpha_{i}-\alpha_{j}\| \leq \delta h\}K(\hat{d}_{i,j}/h) \cdot \mu_{j,t}^{(0)}}{\sum\limits_{j;w_{j,t}=0} \mathbf{1} \left\{\Vert \alpha_{i} - \alpha_{j} \Vert \leq \delta h\right\}K(\hat{d}_{i,j}/h)}
- \mu_{i,t}^{(0)} \right\vert
& \leq\sup_{\Vert\alpha-\alpha_{i}\Vert \leq \delta h}\left\vert \mu^{(0)}(\alpha,\lambda_{t})-\mu^{(0)}(\alpha_{i},\lambda_{t})\right\vert \\
 & \leq L_{0}\delta h
\end{align*}

We now consider the term in the second line (\ref{eq:secon_line}). First
consider the numerator. By Lemma \ref{lem:char_dist_bound}, there exist constants $0<\underline{c}<\bar{c}<\infty$ such that, with probability approaching one, for all $i$, $\underline{c}h^{d_{\alpha}}\leq\frac{1}{N}\sum\limits_{j;w_{j,t}=0}\mathbf{1}\left\{\Vert\alpha_{i}-\alpha_{j}\Vert\leq\delta h\right\}\leq\bar{c}h^{d_{\alpha}}$. Call the event that this holds $E_h$, then we have
\begin{align*}
&\hspace{5pt} \mathbb{P}_{t}\left(\max_{i\in[N]} \left\vert \frac{1}{Nh^{d_{\alpha}}}\sum\limits_{j;w_{j,t}=0} \mathbf{1} \left\{\Vert\alpha_{i}-\alpha_{j}\Vert \leq \delta h \right\}K\left(\hat{d}_{i,j}/h\right) \cdot u_{j,t} \right\vert > \eta \right)\\
\leq&\mathbb{P}_{t}\left(\max_{i\in[N]} \left\vert \frac{1}{Nh^{d_{\alpha}}}\sum\limits_{j;w_{j,t}=0} \mathbf{1} \left\{\Vert\alpha_{i}-\alpha_{j}\Vert \leq \delta h \right\}K\left(\hat{d}_{i,j}/h\right) \cdot u_{j,t} \right\vert > \eta \big\vert\hspace{3pt}E_h\right)+1-\mathbb{P}_t(E_h)
\end{align*}
By applying union bound to the first probability on the second line above, we obtain
\begin{align*}
&\hspace{5pt} \mathbb{P}_{t}\left(\max_{i\in[N]} \left\vert \frac{1}{Nh^{d_{\alpha}}}\sum\limits_{j;w_{j,t}=0} \mathbf{1} \left\{\Vert\alpha_{i}-\alpha_{j}\Vert \leq \delta h \right\}K\left(\hat{d}_{i,j}/h\right) \cdot u_{j,t} \right\vert > \eta \right)\\
\leq&\hspace{5pt} N \cdot \max_i\mathbb{P}_{t} \left( \left\vert \frac{1}{Nh^{d_{\alpha}}}\sum\limits_{j;w_{j,t}=0} \mathbf{1} \left\{\Vert\alpha_{i}-\alpha_{j}\Vert \leq \delta h \right\}K\left(\hat{d}_{i,j}/h\right) \cdot u_{j,t} \right\vert > \eta \big\vert\hspace{3pt}E_h\right)+1-\mathbb{P}_t(E_h).
\end{align*}
 Suppose that this $E_h$ holds, it follows that
\begin{align*}
&\hspace{5pt} \mathbb{P}\left( \left\vert \frac{1}{Nh^{d_{\alpha}}}\sum\limits_{j;w_{j,t}=0} \mathbf{1}\left\{\Vert\alpha_{i}-\alpha_{j}\Vert\leq\delta h\right\}K\left(\hat{d}_{i,j}/h\right) \cdot u_{j,t} \right\vert > \eta \hspace{3pt}\big\vert\hspace{3pt} \left\{\alpha_{j},\hat{d}_{i,j},w_{j,t}\right\}_{j=1}^{N}, \lambda_{t} \right)\\
\leq&\hspace{5pt} \mathbb{P} \left(\left\vert \frac{1}{\sum_{j;w_{j,t}=0} \mathbf{1}\left\{\Vert\alpha_{i}-\alpha_{j}\Vert\leq\delta h\right\}} \sum\limits_{j;w_{j,t}=0, \Vert\alpha_{i}-\alpha_{j}\Vert\leq\delta h} K\left(\hat{d}_{i,j}/h\right) \cdot u_{j,t}\right\vert > \frac{1}{\bar{c}} \cdot \eta \hspace{3pt}\big\vert\hspace{3pt} \left\{\alpha_{j},\hat{d}_{i,j},w_{j,t}\right\}_{j=1}^{N},\lambda_{t} \right)
\end{align*}
First, note that, given $\left\{\alpha_{j},\hat{d}_{i,j},w_{j,t}\right\}_{j=1}^{N}$ and $\lambda_{t}$, the terms $\left\{K(\hat{d}_{i,j}/h)\right\}_{j=1}^{N}$ are constants, and these are bounded above by Assumption \ref{asm:4} \ref{asm:4-6}. By Assumptions \ref{asm:1} and \ref{asm:4} \ref{asm:4-3}, $\{u_{j,t}\}_{j=1}^{N}$ are sub-Gaussian, zero-mean, and independent random variables conditional on $\left\{\alpha_{j},\hat{d}_{i,j},w_{j,t}\right\}_{j=1}^{N}$ and $\lambda_{t}$. Thus, from the Bernstein inequality in Lemma \ref{lem:bi_subexp}, it follows
that for constants $C$, $a$, and $b$, for any $\eta>0$,
\begin{align*}
&\hspace{5pt} \mathbb{P}_{t}\left(\left\vert\frac{1}{\sum_{j;w_{j,t}=0}\mathbf{1}\left\{\Vert\alpha_{i}-\alpha_{j}\Vert\leq\delta h\right\}}\sum\limits_{j;w_{j,t}=0,\Vert\alpha_{i}-\alpha_{j}\Vert\leq\delta h}K\left(\hat{d}_{i,j}/h\right) \cdot u_{j,t}\right\vert > \frac{1}{\bar{c}} \cdot \eta \hspace{3pt}\big\vert\hspace{3pt} \left\{\alpha_{j},\hat{d}_{i,j},w_{j,t}\right\}_{j=1}^{N}, \lambda_{t} \right)\\
\leq&\hspace{5pt} C \cdot \exp\left(-\frac{\sum\limits_{j:w_{j,t}=0} \mathbf{1} \left\{\Vert\alpha_{i}-\alpha_{j}\Vert\leq\delta h\right\} \cdot \bar{c}^{-2}\eta^{2}}{a+b\eta}\right)\\
\leq&\hspace{5pt} C \cdot \exp\left(-\frac{\underline{c}Nh^{d_{\alpha}}\bar{c}^{-2}\eta^{2}}{a+b\eta}\right),
\end{align*}
where the last line again uses that $E_h$ holds, i.e., $\underline{c}h^{d_{\alpha}}\leq\frac{1}{N}\sum_{j;w_{j,t}=0} \mathbf{1}\left\{\Vert\alpha_{i}-\alpha_{j}\Vert\leq\delta h\right\}\leq\bar{c}h^{d_{\alpha}}$.
Then by the law of iterated expectations, we have
\begin{align*}
&N \cdot \max_i\mathbb{P}_{t} \left(\left\vert\frac{1}{Nh^{d_{\alpha}}}\sum_{j;w_{j,t}=0}\mathbf{1}\left\{\Vert\alpha_{i}-\alpha_{j}\Vert\leq\delta h\right\}K\left(\hat{d}_{i,j}/h\right) \cdot u_{j,t} \right\vert > \eta\big\vert\hspace{3pt}E_h\right)
\\ \leq & C \cdot \exp\left(\log(N)-\frac{\underline{c}Nh^{d_{\alpha}}\bar{c}^{-2}\eta^{2}}{a+b\eta}\right).
\end{align*}
And so, combining with our earlier result we get
\begin{align*}
&\mathbb{P}_{t}\left(\max_{i\in[N]} \left\vert \frac{1}{Nh^{d_{\alpha}}}\sum\limits_{j;w_{j,t}=0} \mathbf{1} \left\{\Vert\alpha_{i}-\alpha_{j}\Vert \leq \delta h \right\}K\left(\hat{d}_{i,j}/h\right) \cdot u_{j,t} \right\vert > \eta \right)
\\ \leq & C \cdot \exp\left(\log(N)-\frac{\underline{c}Nh^{d_{\alpha}}\bar{c}^{-2}\eta^{2}}{a+b\eta}\right)+1-\mathbb{P}_t(E_h).
\end{align*}
By setting $\eta = \bar{c}\sqrt{\frac{Ma \log(N)}{\underline{c}h^{d_{\alpha}}N}}$ in the
above inequality for a fixed $M>1$, and noting that $\eta\to0$ and $\mathbb{P}_t(E_h)\to 1$, we obtain
\begin{align*}
&\hspace{5pt} \mathbb{P}_{t}\left(\max_{i\in[N]} \left\vert \frac{1}{Nh^{d_{\alpha}}}\sum\limits_{j;w_{j,t}=0} \mathbf{1} \left\{\Vert\alpha_{i}-\alpha_{j}\Vert \leq \delta h \right\}K\left(\hat{d}_{i,j}/h\right) \cdot u_{j,t} \right\vert > \eta \right) \\
\leq&\hspace{5pt} C \cdot \exp\left(\log(N) - \frac{M a \log(N)}{a+b\eta}\right)+1-\mathbb{P}_t(E_h)
\to 0,
\end{align*}
and therefore we have
\begin{align*}
&\hspace{5pt} \mathbb{P}_{t} \left(\max_{i\in[N]}\left\vert\frac{1}{Nh^{d_{\alpha}}}\sum_{j;w_{j,t}=0}\mathbf{1}\left\{\Vert\alpha_{i}-\alpha_{j}\Vert\leq\delta h\right\}K\left(\hat{d}_{i,j}/h\right) \cdot u_{j,t} \right\vert > \bar{c}\sqrt{\frac{Ma \log(N)}{\underline{c}h^{d_{\alpha}}N}} \right) \\
=&\hspace{5pt}  o(1).
\end{align*}

Finally, we consider the denominator term $\frac{1}{Nh^{d_{\alpha}}}\sum\limits_{j;w_{j,t}=0} K\left(\hat{d}_{i,j}/h\right)$. By adding and subtracting terms, we have
\begin{align*}
&\hspace{5pt} \frac{1}{Nh^{d_{\alpha}}}\sum\limits_{j;w_{j,t}=0}\mathbf{1}\left\{\Vert\alpha_{i}-\alpha_{j}\Vert\leq\delta h\right\}K\left(\hat{d}_{i,j}/h\right) \\
=&\hspace{5pt} \frac{1}{Nh^{d_{\alpha}}}\sum\limits_{j;w_{j,t}=0}\mathbf{1}\left\{\Vert\alpha_{i}-\alpha_{j}\Vert\leq\delta h\right\}K\left(d_{i,j}/h\right) \\
+&\hspace{5pt} \frac{1}{Nh^{d_{\alpha}}}\sum\limits_{j;w_{j,t}=0}\mathbf{1}\left\{\Vert\alpha_{i}-\alpha_{j}\Vert\leq\delta h\right\}\left[K\left(\hat{d}_{i,j}/h\right) - K\left(d_{i,j}/h\right)\right].
\end{align*}
Using that the kernel is Lipschitz by Assumption \ref{asm:4} \ref{asm:4-6}, we have
\begin{align*}
&\hspace{5pt} \max_{i\in[N]} \left\vert \frac{1}{Nh^{d_{\alpha}}}\sum\limits_{j;w_{j,t}=0} \mathbf{1} \left\{\Vert\alpha_{i}-\alpha_{j}\Vert\leq\delta h\right\}\left[K\left(\hat{d}_{i,j}/h\right) - K\left(d_{i,j}/h\right) \right] \right\vert \\
\leq&\hspace{5pt} \bar{K}' \frac{\max_{i,j} \left\vert \hat{d}_{i,j}-d_{i,j}\right\vert}{h} \max_{i\in[N]} \frac{1}{Nh^{d_{\alpha}}} \sum\limits_{j;w_{j,t}=0} \mathbf{1} \left\{\Vert\alpha_{i}-\alpha_{j}\Vert\leq\delta h \right\}.
\end{align*}
By Lemma \ref{lem:pseudo_dist_rate}, $\frac{\max\limits_{j, i} \left\vert\hat{d}_{i,j}-d_{i,j} \right\vert}{h} = o_{p}(1)$ and by Lemma \ref{lem:char_dist_bound} $\max\limits_{i\in[N]}\frac{1}{Nh^{d_{\alpha}}}\sum\limits_{j;w_{j,t}=0} \mathbf{1} \left\{\Vert\alpha_{i}-\alpha_{j}\Vert\leq\delta h \right\} = \mathcal{O}_{p}(1)$, it follows that
\begin{align*}
&\hspace{5pt} \min_{i\in[N]} \frac{1}{Nh^{d_{\alpha}}} \sum\limits_{j;w_{j,t}=0} \mathbf{1}\left\{\Vert\alpha_{i}-\alpha_{j}\Vert\leq\delta h\right\} K\left(\hat{d}_{i,j}/h\right) \\
=&\hspace{5pt} \min_{i\in[N]}\frac{1}{Nh^{d_{\alpha}}} \sum\limits_{j;w_{j,t}=0} \mathbf{1} \left\{\Vert\alpha_{i}-\alpha_{j}\Vert\leq\delta h\right\} K\left(d_{i,j}/h\right) + o_{p}(1).
\end{align*}
Finally, by Lemma \ref{lem:char_dist_bound}, there exists a $\tilde{c}>0$ such that, with probability
approaching one,
\begin{align*}
\min_{i\in[N]} \frac{1}{Nh^{d_{\alpha}}} \sum\limits_{j;w_{j,t}=0} \mathbf{1} \left\{\Vert\alpha_{i}-\alpha_{j}\Vert\leq\delta h \right\}K\left(d_{i,j}/h\right)
= \min_{i\in[N]} \frac{1}{Nh^{d_{\alpha}}} \sum\limits_{j;w_{j,t}=0} K\left(d_{i,j}/h\right)
> \tilde{c}.
\end{align*}
Therefore, combining the bias bound $L_{0}\delta h$ on the first line of the decomposition with the numerator and denominator bounds just established, we conclude that
\begin{align*}
\max_{i\in[N]} \left\vert \hat{\mu}_{i,t}^{(0)}-\mu_{i,t}^{(0)} \right\vert \
= \mathcal{O}_{p} \left( h + \sqrt{\frac{\log(N)}{h^{d_{\alpha}}N}} \right).
\end{align*}
\end{proof}

\subsection{Proof of Theorem \ref{thm:2}}

\begin{proof}[Proof of Theorem \ref{thm:2}]
Recall the definition of the ATT estimate,
\begin{align*}
\hat{\mathrm{ATT}}_{t} & =\frac{1}{N_{1,t}}\sum_{i=1}^{N}\left(Y_{i,t}w_{i,t}-\frac{(1-w_{i,t})Y_{i,t}\hat{p}_{i,t}+(w_{i,t}-\hat{p}_{i,t})\hat{\mu}_{i,t}^{(0)}}{1-\hat{p}_{i,t}}\right).
\end{align*}
Using $(1-w_{i,t})Y_{i,t}=(1-w_{i,t})\mu^{(0)}_{i,t}+(1-w_{i,t})u_{i,t}$ and $w_{i,t}=p_{i,t}+\epsilon_{i,t}$, we can express this as
\begin{align*}
\hat{\mathrm{ATT}}_{t} & =\frac{1}{N_{1,t}}\sum_{i=1}^{N}\left(Y_{i,t}w_{i,t}-\frac{(1-w_{i,t})Y_{i,t}p_{i,t}+(w_{i,t}-p_{i,t})\mu_{i,t}^{(0)}}{1-p_{i,t}}\right)\\
 & +\frac{1}{N_{1,t}}\sum_{i=1}^{N}\frac{1}{1-\hat{p}_{i,t}}(p_{i,t}-\hat{p}_{i,t})(\mu_{i,t}^{(0)}-\hat{\mu}_{i,t}^{(0)})\\
 & +\frac{1}{N_{1,t}}\sum_{i=1}^{N}\epsilon_{i,t} \left(\frac{\mu_{i,t}^{(0)}-\hat{\mu}_{i,t}^{(0)}}{1-\hat{p}_{i,t}} \right)\\
 & +\frac{1}{N_{1,t}}\sum_{i=1}^{N}\frac{(1-w_{i,t})u_{i,t}}{(1-p_{i,t})}\left(\frac{p_{i,t}-\hat{p}_{i,t}}{1-\hat{p}_{i,t}}\right).
\end{align*}
We tackle each line separately.

By the upper bound in Assumption \ref{asm:4} \ref{asm:4-2},  for some $c>0$, we have that $\frac{1}{1-p_{i,t}}\leq 1/c$ with probability $1$. And by supposition, we have $\max_{i\in[N]} \left\vert \hat{p}_{i,t}-p_{i,t} \right\vert = o_{p}(1)$ and so together with the previous bound we see that $\max_{i\in[N]}\left\vert \frac{1}{1-\hat{p}_{i,t}}-\frac{1}{1-p_{i,t}}\right\vert = o_{p}(1)$ and thus $\max_{i\in[N]}\frac{1}{1-\hat p_{i,t}}=O_{p}(1)$. Additionally, by supposition $\max_{i\in[N]} \left\vert \left(p_{i,t}-\hat{p}_{i,t}\right) \left(\mu_{i,t}^{(0)}-\hat{\mu}_{i,t}^{(0)}\right) \right\vert = o_{p}\left(N^{-1/2}\right)$. With the lower bound in Assumption \ref{asm:4} \ref{asm:4-2} and the law of large numbers, $\frac{N}{N_{1,t}}$ is bounded above uniformly with probability approaching $1$. Hence, it follows that
\begin{align*}
\frac{1}{N_{1,t}}\sum_{i=1}^{N}\frac{1}{1-\hat{p}_{i,t}} \left(p_{i,t}-\hat{p}_{i,t}\right) \left(\mu_{i,t}^{(0)}-\hat{\mu}_{i,t}^{(0)}\right) = o_{p}\left(N^{-1/2}\right).
\end{align*}

Define $\mathcal{D}_{k,t} := \left( \left\{ \left\{(Y_{j,s}(0),w_{j,s}) \right\}_{s=1}^{t} \right\}_{j\in\mathcal{I}_{-k}},\left\{ \left\{ Y_{i,s} \right\}_{s=1}^{T_{0}} \right\}_{i\in\mathcal{I}_{k}}\right)$. By Assumptions \ref{asm:2}  and \ref{asm:4} \ref{asm:4-3}, for $t>T_0$ and $i\in\mathcal{I}_k$, we have $
\epsilon_{i,t}\indep\mathcal{D}_{k,t} \hspace{3pt}\big\vert\hspace{3pt} \left\{ \alpha_{\ell} \right\}_{\ell=1}^{N}, \lambda_{t}$ and this implies
\begin{align*}
\mathbb{E} \left[\epsilon_{i,t} \hspace{3pt}\big\vert\hspace{3pt} \mathcal{D}_{k,t}, \left\{\alpha_{\ell}\right\}_{\ell=1}^{N},\lambda_{t} \right] = 0.
\end{align*}
Thus, by the law of iterated expectations, and using the fact that $\hat{\mu}_{i,t}^{(0)}$ and $\hat{p}_{i,t}$ are non-stochastic given $\mathcal{D}_{k,t}$, we have
\begin{align*}
\mathbb{E} \left[\epsilon_{i,t}\left(\frac{\mu_{i,t}^{(0)}-\hat{\mu}_{i,t}^{(0)}}{1-\hat{p}_{i,t}}\right) \hspace{3pt}\bigg\vert\hspace{3pt} \mathcal{D}_{k,t},\{\alpha_{\ell}\}_{\ell=1}^{N},\lambda_{t}\right]
&= \mathbb{E} \left[\epsilon_{i,t} \hspace{3pt}\big\vert\hspace{3pt} \mathcal{D}_{k,t},\{\alpha_{\ell}\}_{\ell=1}^{N},\lambda_{t}\right] \left(\frac{\mu_{i,t}^{(0)}-\hat{\mu}_{i,t}^{(0)}}{1-\hat{p}_{i,t}}\right)\\
&= 0.
\end{align*}
Also, by assumptions \ref{asm:2}, \ref{asm:3}, and \ref{asm:4} \ref{asm:4-3}, for any $i,j \in \mathcal{I}_{k}$ with $i\neq j$, we have
 $
\epsilon_{i,t}\indep\epsilon_{j,t} \hspace{3pt}\big\vert\hspace{3pt} \{\alpha_{\ell}\}_{\ell=1}^{N},\lambda_{t},\mathcal{D}_{k,t}
$. Using this and the fact that $\epsilon_{i,t}\in[-1,1]$ for all $i$, it follows that
\begin{align*}
\mathbb{E} \left[\epsilon_{i,t}\epsilon_{j,t} \left(\frac{\mu_{i,t}^{(0)}-\hat{\mu}_{i,t}^{(0)}}{1-\hat{p}_{i,t}}\right) \left(\frac{\mu_{j,t}^{(0)}-\hat{\mu}_{j,t}^{(0)}}{1-\hat{p}_{j,t}} \right) \hspace{3pt}\bigg\vert\hspace{3pt} \mathcal{D}_{k,t}, \{\alpha_{\ell}\}_{\ell=1}^{N}, \lambda_{t}\right]
&= \mathbb{E} \left[ \epsilon_{i,t}\epsilon_{j,t} \hspace{3pt}\big\vert\hspace{3pt} \mathcal{D}_{k,t},\{\alpha_{\ell}\}_{\ell=1}^{N},\lambda_{t} \right] \left(\frac{\mu_{i,t}^{(0)}-\hat{\mu}_{i,t}^{(0)}}{1-\hat{p}_{i,t}}\right) \left(\frac{\mu_{j,t}^{(0)}-\hat{\mu}_{j,t}^{(0)}}{1-\hat{p}_{j,t}}\right) \\
&\leq \mathbf{1} \left\{ j = i \right\} \cdot  \left(\frac{\mu_{i,t}^{(0)}-\hat{\mu}_{i,t}^{(0)}}{1-\hat{p}_{i,t}} \right)^{2}.
\end{align*}
Using the results above and Markov's inequality, we have
\begin{align*}
\mathbb{P} \left(\left\vert \frac{1}{\sqrt{\left\vert \mathcal{I}_{k}\right\vert}}\sum_{i\in\mathcal{I}_{k}}\epsilon_{i,t} \left(\frac{\mu_{i,t}^{(0)}-\hat{\mu}_{i,t}^{(0)}}{1-\hat{p}_{i,t}}\right) \right\vert > \epsilon \hspace{3pt}\bigg\vert\hspace{3pt} \mathcal{D}_{k,t},\{\alpha_{i}\}_{i=1}^{N},\lambda_{t}\right)
\leq\frac{\frac{1}{\left\vert \mathcal{I}_{k}\right\vert} \sum_{i \in \mathcal{I}_{k}}\left(\frac{\mu_{i,t}^{(0)}-\hat{\mu}_{i,t}^{(0)}}{1-\hat{p}_{i,t}}\right)^{2}}{\epsilon^{2}}=o_{p}(1),
\end{align*}
where the final equality uses $\max_i|\hat\mu^{(0)}_{i,t}-\mu^{(0)}_{i,t}|=o_p(1)$ together with $\max_i 1/(1-\hat p_{i,t})=O_p(1)$. Hence, by using $\vert \mathcal{I}_{k}\vert \asymp N$, it follows that
\begin{align*}
\frac{1}{|\mathcal{I}_{k}|}\sum_{i\in\mathcal{I}_{k}}\epsilon_{i,t} \left(\frac{\mu_{i,t}^{(0)}-\hat{\mu}_{i,t}^{(0)}}{1-\hat{p}_{i,t}} \right) = o_{p}\left(N^{-1/2}\right).
\end{align*}
Since the number of folds $\mathcal{K}$ is fixed, summing over $k\in[\mathcal{K}]$ and using $N/N_{1,t}=O_{p}(1)$ gives $\frac{1}{N_{1,t}}\sum_{i=1}^{N}\epsilon_{i,t}\big(\mu^{(0)}_{i,t}-\hat\mu^{(0)}_{i,t}\big)/(1-\hat p_{i,t})=o_{p}(N^{-1/2})$.

Similarly, by Assumptions \ref{asm:2} and \ref{asm:4} \ref{asm:4-3}, we have
$
u_{i,t} \indep \mathcal{D}_{k,t} \hspace{3pt}\big\vert\hspace{3pt} \{\alpha_{l}\}_{l=1}^{N},\lambda_{t},w_{i,t}=0$,
and for any $i,j \in \mathcal{I}_{k}$ with $i\neq j$,
\begin{align*}
u_{i,t}\indep u_{j,t} \hspace{3pt}\big\vert\hspace{3pt} \{\alpha_{l}\}_{l=1}^{N},\lambda_{t}, \mathcal{D}_{k,t}, w_{i,t}=w_{j,t}=0.
\end{align*}
Therefore, by the law of iterated expectations, alongside the fact that $\hat{p}_{i,t}$ is non-stochastic given $\mathcal{D}_{k,t}$, it follows that
\begin{align*}
\mathbb{E} \left[\frac{u_{i,t}}{(1-p_{i,t})} \left(\frac{p_{i,t}-\hat{p}_{i,t}}{1-\hat{p}_{i,t}}\right) \hspace{3pt}\bigg\vert\hspace{3pt} \mathcal{D}_{k,t},\{\alpha_{l}\}_{l=1}^{N},\lambda_{t},w_{i,t}=0\right]
&= \mathbb{E} \left[u_{i,t} \hspace{3pt}\big\vert\hspace{3pt} \{\alpha_{l}\}_{l=1}^{N},\lambda_{t},w_{i,t}=0 \right] \frac{1}{1-p_{i,t}}\left(\frac{p_{i,t}-\hat{p}_{i,t}}{1-\hat{p}_{i,t}}\right) \\
&= 0
\end{align*}
Moreover, we have
\begin{align*}
&\hspace{5pt} \mathbb{E} \left[ \frac{u_{i,t}}{(1-p_{i,t})}\frac{u_{j,t}}{(1-p_{j,t})} \left( \frac{p_{i,t}-\hat{p}_{i,t}}{1-\hat{p}_{i,t}} \right) \left( \frac{p_{j,t}-\hat{p}_{j,t}}{1-\hat{p}_{j,t}} \right) \hspace{3pt}\bigg\vert\hspace{3pt} \mathcal{D}_{k,t},\{\alpha_{k}\}_{k=1}^{N},\lambda_{t},w_{i,t}=w_{j,t}=0 \right] \\
=&\hspace{5pt} \mathbb{E} \left[ \frac{u_{i,t}}{(1-p_{i,t})}\frac{u_{j,t}}{(1-p_{j,t})} \hspace{3pt}\bigg\vert\hspace{3pt} \mathcal{D}_{k,t}, \{\alpha_{k}\}_{k=1}^{N}, \lambda_{t}, w_{i,t}=w_{j,t}=0 \right] \left( \frac{p_{i,t}-\hat{p}_{i,t}}{1-\hat{p}_{i,t}} \right) \left( \frac{p_{j,t}-\hat{p}_{j,t}}{1-\hat{p}_{j,t}} \right) \\
\leq&\hspace{5pt} \mathbf{1}\{j=i\} \cdot \mathbb{E} \left[u_{i,t}^{2} \hspace{3pt}\big\vert\hspace{3pt} w_{i,t}=0,\{\alpha_{k}\}_{k=1}^{N},\lambda_{t} \right] \frac{1}{(1-p_{i,t})^{2}} \left(\frac{p_{i,t}-\hat{p}_{i,t}}{1-\hat{p}_{i,t}}\right)^{2}\\
\leq&\hspace{5pt} \mathbf{1}\{j=i\} \cdot \bar{\sigma}_{Y}^{2} \cdot \frac{1}{(1-p_{i,t})^{2}} \left(\frac{p_{i,t}-\hat{p}_{i,t}}{1-\hat{p}_{i,t}} \right)^{2},
\end{align*}
for some finite constant $\bar{\sigma}_{Y}^{2}$, which exists by Assumption \ref{asm:4} \ref{asm:4-3} (by standard properties of sub-Gaussian random variables $u_{i,t}$ has conditional variance at most $2\delta$ where $\delta$ is the `variance proxy' referred to in the assumption). Thus, by Markov's inequality, it follows that%
\begin{align*}
\mathbb{P} \left( \left\vert \frac{1}{\sqrt{\vert\mathcal{I}_{k}\vert}}\sum_{i\in\mathcal{I}_{k}}(1-w_{i,t})\frac{u_{i,t}}{1-p_{i,t}}\left(\frac{p_{i,t}-\hat{p}_{i,t}}{1-\hat{p}_{i,t}}\right) \right\vert > \epsilon \hspace{3pt}\bigg\vert\hspace{3pt} \mathcal{D}_{k,t}, \{\alpha_{i}\}_{i=1}^{N}, \lambda_{t} \right)
\leq\frac{\frac{1}{\vert\mathcal{I}_{k}\vert} \sum_{i\in\mathcal{I}_{k}}\bar{\sigma}_{Y}^{2}\frac{1}{(1-p_{i,t})^2}\left(\frac{p_{i,t}-\hat{p}_{i,t}}{1-\hat{p}_{i,t}}\right)^{2}}{\epsilon^{2}}
= o_{p}(1)
\end{align*}
The conditional probability is $o_p(1)$ which implies that the unconditional probability is $o(1)$. Hence, we have
\begin{align*}
\frac{1}{\vert\mathcal{I}_{k}\vert}\sum_{i\in\mathcal{I}_{k}} (1-w_{i,t}) \frac{u_{i,t}}{(1-p_{i,t})} \left( \frac{p_{i,t}-\hat{p}_{i,t}}{1-\hat{p}_{i,t}} \right)
= o_{p}\left(N^{-1/2}\right).
\end{align*}

By combining all previous results, we have shown that
\begin{align*}
\sqrt{N}\hat{\mathrm{ATT}}_{t}
= \frac{\sqrt{N}}{N_{1,t}} \sum_{i=1}^{N} \left(Y_{i,t}w_{i,t}-\frac{(1-w_{i,t})Y_{i,t}p_{i,t}+(w_{i,t}-p_{i,t})\mu_{i,t}^{(0)}}{1-p_{i,t}} \right) + o_{p}(1).
\end{align*}
The remaining steps are a standard application of asymptotic theory for the ${\rm ATT}$ estimate with known nuisance functions, and so we omit them here.
\end{proof}

\subsection{Bernstein Inequalities}

In this section, we specify the Bernstein inequalities, which we refer to in the proofs. %

\subsubsection{Bernstein inequality for bounded random variables}

\begin{lemma}{A.5}[Bernstein inequality for bounded random variables; see \cite{bennett1962probability}] \label{lem:bi}
Let $Z_1, ..., Z_n$ be mean zero independent random variables. Assume there exists a positive constant $M$ such that $|Z_i| \leq M$ with probability one for each $i$. Also let $\sigma^2 := \frac{1}{n} \sum_{i=1}^n \mathbb{E}[Z_i^2]$. Then for all $\epsilon > 0$
\begin{align*}
\mathbb{P} \left( \left\vert \frac{1}{n} \sum_{i=1}^n Z_i \right\vert \geq \epsilon \right) \leq 2 \exp \left( - \frac{n \epsilon^2}{2(\sigma^2 + \frac{1}{3} M \epsilon)} \right).
\end{align*}
\end{lemma}

\subsubsection{Bernstein inequality for unbounded random variables}

\begin{lemma}{A.6}[Moments of a sub-Gaussian random variable]
	\label{lem: A subG moments}
	Let $Z$ be a mean zero random variable satisfying $\mathbb E[{e^{\lambda Z}}] \leqslant e^{v \lambda^2/2}$ for all $\lambda \in \mathbb R$, for some $v > 0$. Then for every integer $q \geqslant 1$, we have $\mathbb E [{Z^{2q}}] \leqslant q! (4v)^{q}$.
\end{lemma}

\begin{proof}[Proof of Lemma \ref{lem: A subG moments}]
	See Theorem 2.1 in \citet{boucheron2013concentrationinequalities}.
\end{proof}

\begin{lemma}{A.7}[Bernstein inequality for unbounded random variables]
    \label{lem: A bi unb}
    Let $Z_1, ..., Z_n$ be independent random variables. Assume that there exist some positive constants $\nu$ and $c$ such that ${\frac{1}{n}\sum_{i=1}^{n}\mathbb{E}[Z_i^2]\leq\nu}$ and for all integers $q \geq 3$,
    \begin{align*}
    \frac{1}{n}\sum_{i=1}^{n}\mathbb{E}[|Z_i|^q] \leq \frac{q! c^{q-2}}{2}\nu.
    \end{align*}
    Then for all $\epsilon > 0$,
    \begin{align*}
    \mathbb{P} \left( \left\vert \frac{1}{n} \sum_{i=1}^n (Z_i - \mathbb{E}[Z_i]) \right\vert \geq \epsilon \right) \leq 2 \exp \left( - \frac{n \epsilon^2}{2(\nu + c \epsilon)} \right).
    \end{align*}
\end{lemma}

\begin{proof}[Proof of Lemma \ref{lem: A bi unb}]
    Apply the one-sided bound in Corollary 2.11 in \cite{boucheron2013concentrationinequalities} to $Z_1,\ldots,Z_n$ and to $-Z_1,\ldots,-Z_n$, with variance parameter $n\nu$ and threshold $n\epsilon$. The moment assumptions hold for both signs, so the union bound gives the stated two-sided inequality, with the factor of two accounting for the two tails.
\end{proof}

\begin{lemma}{A.8}[Bernstein inequality for sub-Gaussian random variables]%
\label{lem:bi_subexp}
Let $Z_1, ..., Z_n$ be mean-zero independent random variables. Assume that there exists some $v > 0$ such that $\mathbb E[{e^{\lambda Z_i}}] \leqslant e^{v \lambda^2/2}$ for all $\lambda \in \mathbb{R}$ and for all $i \in \{1, \dots, n\}$. Fix $\bar\beta\in(0,\infty)$. Then there exist positive constants $C$, $a$, and $b$, depending only on $v$ and $\bar\beta$, such that for all constants $\beta_1, ..., \beta_n$ satisfying $\max_i |\beta_i | < \bar{\beta}$ and for all $\epsilon > 0$,
\begin{align*}
\mathbb{P} \left( \left\vert \frac{1}{n} \sum_{i=1}^n \beta_i Z_i \right\vert \geq \epsilon \right) \leq C \exp \left( - \frac{n \epsilon^2}{a + b \epsilon} \right)
\end{align*}
and
\begin{align*}
\mathbb{P} \left( \left\vert \frac{1}{n} \sum_{i=1}^n \beta_i (Z_i^2 - \mathbb{E}[Z_i^2]) \right\vert \geq \epsilon \right) \leq C \exp \left( - \frac{n \epsilon^2}{a + b \epsilon} \right).
\end{align*}
Additionally, let $W_1, ..., W_n$ be mean-zero random variables such that the pairs $(Z_i,W_i)$ are independent across $i$, $\mathbb E[e^{\lambda W_i}]\leq e^{v\lambda^2/2}$ for all $\lambda\in\mathbb R$ and all $i$, and $\mathbb E[Z_iW_i]=0$ for all $i$. Then, after enlarging $C$, $a$, and $b$ if necessary, for all $\epsilon>0$,
\[
\mathbb P\left(\left|\frac1n\sum_{i=1}^n Z_iW_i\right|\geq\epsilon\right)
\leq C\exp\left(-\frac{n\epsilon^2}{a+b\epsilon}\right).
\]
\end{lemma}

\begin{proof}[Proof of Lemma \ref{lem:bi_subexp}]
    The first inequality follows from the sub-Gaussian tail bound in Section 2.3 of \citet{boucheron2013concentrationinequalities}, since $n^{-1}\sum_{i=1}^n\beta_iZ_i$ has variance factor at most $v\bar\beta^2/n$.

    For the second inequality, Lemma \ref{lem: A subG moments} gives $\mathbb E [|\beta_iZ_i^2|^q] \leq q!(4v\bar\beta)^q$ for every integer $q\geq2$, so Lemma \ref{lem: A bi unb} applies to $\beta_iZ_i^2$.

    For the third inequality, Cauchy--Schwarz and Lemma \ref{lem: A subG moments} give
    \begin{align*}
      \mathbb E [|Z_iW_i|^q] \leq(\mathbb E[Z_i^{2q}]\mathbb E[W_i^{2q}])^{1/2}\leq q!(4v)^q
    \end{align*}
    for every integer $q\geq2$. Apply Lemma \ref{lem: A bi unb} to the independent, mean-zero products $Z_iW_i$. The constants in all three bounds can be chosen to depend only on $v$ and $\bar\beta$.
\end{proof}

\end{document}